\documentclass[groupedaddress, showpacs, superscriptaddress]{revtex4-2}

\usepackage{xcolor}  
\usepackage{graphicx}         
\usepackage{tikz}             
\usetikzlibrary{graphs}       
\usepackage[table]{xcolor}    

\usepackage{amssymb,amsmath,amsthm}
\usepackage{bbold}            

\usepackage{hyperref}

\hypersetup{colorlinks,citecolor={blue},linkcolor={red},urlcolor={teal}}

\usepackage{array}            
\usepackage{multirow,makecell} 

\usepackage{times}            

\theoremstyle{plain}

\newtheorem{theorem}{Theorem}
\newtheorem{lemma}[theorem]{Lemma}
\newtheorem{remark}{Remark}

\theoremstyle{definition}

\usepackage[most]{tcolorbox}  
\tcbuselibrary{theorems}

\newtcbtheorem[number within=section]{theorem-box}{Theorem Box}{%
	colback=blue!05!white,
	colframe=red!85!black,
	fonttitle=\bfseries,
	before upper={\parindent15pt\relax},
	enforce breakable
}{th}

\newtheorem{result}{result}

\theoremstyle{definition}

\makeatletter
\newsavebox{\@brx}
\newcommand{\llangle}[1][]{\savebox{\@brx}{\(\m@th{#1\langle}\)}%
  \mathopen{\copy\@brx\mkern2mu\kern-0.9\wd\@brx\usebox{\@brx}}}
\newcommand{\rrangle}[1][]{\savebox{\@brx}{\(\m@th{#1\rangle}\)}%
  \mathclose{\copy\@brx\mkern2mu\kern-0.9\wd\@brx\usebox{\@brx}}}
\makeatother

\begin{document}

\title{Quantum unital Otto heat engines: using Kirkwood-Dirac quasi-probability for the engine's coherence to stay alive}
\author{Abdelkader El Makouri }
\email{Corresponding author: abdelkader\_elmakouri@um5.ac.ma}\affiliation{LPHE-Modeling and Simulation, Faculty of Sciences, Mohammed V University in Rabat, Morocco.} 
\author{Abdallah Slaoui }\affiliation{LPHE-Modeling and Simulation, Faculty of Sciences, Mohammed V University in Rabat, Morocco.}\affiliation{Centre of Physics and Mathematics, CPM, Faculty of Sciences, Mohammed V University in Rabat, Rabat, Morocco.}
\author{Rachid Ahl Laamara}\affiliation{LPHE-Modeling and Simulation, Faculty of Sciences, Mohammed V University in Rabat, Morocco.}\affiliation{Centre of Physics and Mathematics, CPM, Faculty of Sciences, Mohammed V University in Rabat, Rabat, Morocco.}

\begin{abstract}
In this work, we consider \textit{quantum unital Otto heat engines}. The latter refers to the fact that both the unitaries of the adiabatic strokes and the source of the heat provided to the engine preserve the maximally mixed state. We show how to compute the cumulants of either the dephased or undephased engine. For a qubit, we give the analytical expressions of the averages and variances for arbitrary unitaries and unital channels. We do a detailed comparative study between the dephased and undephased heat engines. More precisely, we focus on the effect of the parameters on the average work and its reliability and efficiency. As a case study of unital channels, we consider a quantum projective measurement. We show on which basis we should projectively measure the qubit, either the dephased or undephased heat engine, to extract higher amounts of work, increase the latter's reliability, and increase efficiency. Further, we show that non-adiabatic transitions \textit{are not always detrimental} to thermodynamic quantities. Our results, we believe, are important for heat engines fueled by \textit{quantum measurement}.\par

$\mathbf{Keywords:}$ Cumulants of the unmonitored engine, Kirkwood-Dirac quasi-probability, Measurement-based quantum engine, Universal thermodynamic bounds.
\end{abstract}
\date{\today}

\maketitle
\section{Introduction}
\textit{Quantum mechanics} and \textit{thermodynamics} are two of the best theories that humankind has developed. With today's ability of experimentalists to control quantum systems—something that one even could not imagine at the time of Schrödinger—one starts wondering if and how these two great theories can be fit together into a single framework. Thermodynamics has emerged as a consequence of people's interest in understanding heat engines. More precisely to answer the question of how one can use \textit{heat} and efficiently convert it into useful energy, i.e., \textit{work}. Nowadays, our technologies are getting smaller, so we want to know how we can manage heat at the small scale, e.g., can heat generated at the quantum scale be used advantageously for useful things such as cooling quantum systems? These kinds of questions have led to the field known today as $\textit{quantum thermodynamics}$\cite{Scully,Anders,Huber,Binder,Deffner,Myers1,Nicole, Matteo}. \par

Inspired by $\textit{Maxwell demon}$ and $\textit{Szilard engine}$ \cite{Maxwell}, quantum thermal machines based on \textit{quantum measurement} \cite{Kim,Das,Elouard,Buffoni,Anka,Behzadi,Prasanna1, KimT,Sagawa,Jordan,Chand,Lisboa,Shanhe, Bresque,JuzarSong,JordanNM,Yamamoto,BiswasP,SantosJ,BhandariJ,Perna,Robert} are now under extensive study. The fact that quantum measurement can fuel quantum thermal machines and do quantum computation, see, e.g., Ref. \cite{Robert}, has no classical analogy. This is because, in principle, a measurement in the classical world can extract information without disturbing the state of the measured system. On the other hand, this is not the case when we consider quantum systems such as electrons and photons—\textit{their state gets changed after the measurement}. It is this change in the state of the system that is responsible for the possibility of thermal machines fueled by quantum measurement. \par

Recently, in Ref. \cite{Abdelkader}, we have considered a single qubit quantum Otto cycle where we neither assumed the cycle to be time-reversal symmetric nor specified the unital channel replacing the hot heat bath. Similarly to some works \cite{Mohanta,Watanabe,Gerry,Sandipan, Gerry1,Saryal,Watanabe2,Gerry2}, we have proved that the ratio of the fluctuations of the stochastic work ($W$) and the stochastic heat absorbed $(Q_{M})$ is \textit{lower and upper bounded}. The lower bound was the square of the efficiency of the engine, and the upper bound was 1. In Ref. \cite{Abdelkader2}, we put forward this work and also considered the fluctuations of the stochastic heat released (\textit{denoted by} \( Q_{C}' \)). In general, there is another stochastic quantity denoted by \( Q_{C} \) that has the same average as \( Q_{C}' \) but not the same higher-order cumulants; see Ref. \cite{Abdelkader2} for more details.

In \cite{Abdelkader2}, we have proved that the relative fluctuations (RFs) of $W$, $Q_{M}$, and $Q_{C}$ obey a thermodynamic uncertainty relation given by $2/\llangle \Sigma\rrangle-1$ (see Ref. \cite{Sacchi} and Refs. \cite{Barato,Timpanaro,Horowitz,Horowitz2,England, Falasco,Brandner,Brandner2,LiuSegal, Proesmans,Hasegawa} for more details about thermodynamic uncertainty relations), where $\llangle \Sigma\rrangle$ is the average of the stochastic entropy production.
 We analytically showed that the RFs of $Q_{C}$ always bound the RFs of $W$ and $Q_{M}$, which is better than the bound $2/\llangle \Sigma\rrangle-1$. The latter has the flaw that it becomes negative when $\llangle\Sigma\rrangle>2$. Actually, a lot of thermodynamic uncertainty relations have the flaw that when entropy production diverges, i.e., \textit{in the low-temperature regime}, they tend to zero. However, we know that even in this regime, thermodynamic quantities such as work and heat have \textit{non zero} relative fluctuations. We show this below, numerically. \par

One should note that because of the \textit{projective measurement} between the strokes in Refs. \cite{Abdelkader,Abdelkader2}, coherence was erased. Thus, one would wonder what is the fate of those bounds and the relationship between the RFs of $W$, $Q_{M}$, and $Q_{C}$ in the presence of \textit{quantum coherence}, and if and how the latter can be used to enhance work and its reliability as well as efficiency. Here, we consider the quantum Otto cycle, where the working medium is a single two-level system, similar to what we did in Refs. \cite{Abdelkader,Abdelkader2}. In this paper, the main things that we focus on are: \par 

\begin{enumerate}
\item We show how to derive the cumulants — \textit{in the presence of quantum coherence} — of all thermodynamic quantities from the characteristic function (CF) of the stochastic energies. In Ref. \cite{Abdelkader2}, we have shown how one can do this for the projectively measured engine. In the presence of coherence, we see that the cumulants follow from a \textit{quasiprobability}.

\item We analytically show how one can derive the first and second cumulants only in terms of six transition probabilities (\textit{to be defined below}), independently of the Hamiltonians, the unitaries, and the unital channel. More precisely, we give the compact expressions of the averages and fluctuations of work and heat.

\item Considering as a unital channel a quantum projective measurement, we give a detailed comparison between the dephased and undephased engines. More precisely, we focus on the effect of the parameters (i.e., inverse temperature $\beta$, gaps $2\nu_{1}$ and $2\nu_{2}$, angles of measurement $\chi$ and $\alpha$, the angle $\phi$, and the non-adiabatic parameter $\delta'$), on the average work and its reliability, and also on efficiency. In the main text below, it will be clear what the difference is between the dephased and undephased engines. In general, it is shown that contrary to common wisdom, non-adiabatic transitions are not always detrimental to thermodynamic quantities such as work and efficiency.

\item We also comment on the relationship between the RFs proved in Refs. \cite{Abdelkader,Abdelkader2} when coherence is not erased.

\end{enumerate}

 This paper has been organized as follows: In Section \ref{1}, we explain what we mean by dephased and undephased engines, and we show how one can derive the cumulants of the latter engines and prove that both of these engines can never work as a \textit{refrigerator}. In Section \ref{casestu}, we give the qubit model to which we apply our analytical results. In Section \ref{dede}, we explain in detail which parameters have a good influence on the average work, work reliability, and efficiency of the dephased engine. We show that their highest values are only achieved in the adiabatic regime. Then, in Section \ref{undeunde}, we focus on the undephased engine. We compute the averages and fluctuations of work and heat for an arbitrary unital channel, then we apply them to the qubit model that we show in Section \ref{casestu}. We also give the common and different features between the two engines. In Section \ref{conc}, we give a summary of our results. Finally, in the appendices, we give our arguments to convince the reader to use Eq. (\ref{chiFF}) to derive the cumulants of the undephased engine as well as the proof of some analytical results presented in the main text. We set $\hbar=k_{B}=1$ throughout this paper. \par

\section{Dephased and undephased Quantum Otto heat engines}
\label{1}

Although there are many different thermodynamic cycles, the quantum cycle that we consider in this paper is the quantum version of the Otto cycle \cite{Johal,Rezek,Quan2,Quan,Kosloff,Kieu}. This is because \textit{heat and work exchanges} are separated under the assumption of \textit{ the weak coupling limit} between the working medium and its environment. Also, when the heat and work exchanges are separated, this helps in considering \textit{their higher cumulants} without troubles i.e., work and heat are distinguished even at the stochastic level. The Otto cycle is composed of two adiabatic strokes and two isochoric strokes. In the adiabatic strokes, we have an exchange of work with the external world, while in the isochoric strokes, we have an exchange of heat.

\begin{figure}[hbtp]
\centering
\includegraphics[scale=0.95]{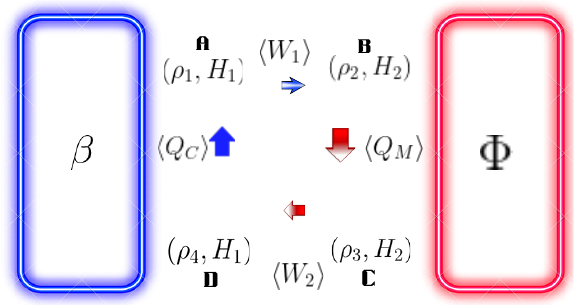}
\caption{Schematic of the four steps of the quantum Otto cycle. The four strokes of the cycle are: $\mathbf{A}\rightarrow \mathbf{B}$, $\mathbf{B}\rightarrow \mathbf{C}$, $\mathbf{C}\rightarrow \mathbf{D}$, and $\mathbf{D}\rightarrow A$. See the main text for details.}\label{Otto}
\end{figure}

\subsection{Undephased engine}\label{ude}

Consider a quantum system with an arbitrary finite dimension $d$.
The cycle steps of the $\textit{undephased engine}$ are as follows; see Fig. \ref{Otto}. $\mathbf{A}:$ First, we assume that the system with Hamiltonian $H_{1}$, is initialized in thermal equilibrium with a heat bath at inverse temperature $\beta$, thus the state is given by $\rho_{1}:=e^{-\beta H_{1}}/Z$. $Z:=\mathrm{Tr}\left[ e^{-\beta H_{1}}\right]$ is the partition function, and $\mathrm{Tr}[\,\cdot\,]
$ is the trace operation. Without loss of generality, and for this moment, let's not specify the expression of $H_{1}$. The initial average energy of the system is given by: $\langle E_{1}\rangle:=\mathrm{Tr}\left[\rho_{1}H_{1}\right]$. From $\mathbf{A}$ to $\mathbf{B}$, we apply the first unitary transformation given by $U$. The Hamiltonian gets changed from $H_{1}$ into $H_{2}$ and the state into $\rho_{2}:=U\rho_{1} U^{\dagger}$. After this unitary transformation, the average energy of the system becomes $\langle E_{2}\rangle:=\mathrm{Tr}\left[\rho_{2}H_{2}\right]$. Since $U$ is entropy preserving, this change in energy is $\textit{work}$. That is, we have, $\langle W_{1}\rangle:=\langle E_{2}\rangle-\langle E_{1}\rangle=\mathrm{Tr}\left[\rho_{2}H_{2}\right]-\mathrm{Tr}\left[\rho_{1}H_{1}\right]$. Then, from $\mathbf{B}$ to $\mathbf{C}$, we fix the Hamiltonian at $H_{2}$, but we apply a unital channel denoted by $\Phi$ \textit{to fuel the system}. In this case, the state of the system evolves into $\rho_{3}:=\Phi(\rho_{2})$. The average energy after applying $\Phi$ is given by $\langle E_{3}\rangle:=\mathrm{Tr}\left[\rho_{3}H_{2}\right]$. In this case, the change in energy is classified as heat, and it is given by, $\langle Q_{M}\rangle:=\langle E_{3}\rangle-\langle E_{2}\rangle=\mathrm{Tr}\left[\rho_{3}H_{2}\right]-\mathrm{Tr}\left[\rho_{2}H_{2}\right]$. Then, from $\mathbf{C}$ to $\mathbf{D}$, we bring back the Hamiltonian from $H_{2}$ into $H_{1}$. In this second adiabatic stroke, the state of the system becomes $\rho_{4}:=V\rho_{3}V^{\dagger}$. The average energy is given by $\langle E_{4}\rangle=\mathrm{Tr}\left[\rho_{4}H_{1}\right]$ and the change in energy is classified as work, and it is given by, $\langle W_{2}\rangle:=\langle E_{4}\rangle-\langle E_{3}\rangle=\mathrm{Tr}\left[\rho_{4}H_{1}\right]-\mathrm{Tr}\left[\rho_{3}H_{2}\right]$. Note that it is not necessary for $V$ to be the time-reversal of $U$, i.e., not necessarily $V:=\Theta U^{\dagger}\Theta^{\dagger}$, where $\Theta$ is the time-reversal operator. Finally, from $\mathbf{D}$ to $\mathbf{A}$, we let the system interact with the initial bath until it is fully in equilibrium, thus closing the cycle. The change in energy in this stroke is classified as heat, and it is given by, $\langle Q_{C}\rangle:=\langle E_{1}\rangle-\langle E_{4}\rangle=\mathrm{Tr}\left[\rho_{1}H_{1}\right]-\mathrm{Tr}\left[\rho_{4}H_{1}\right]$. The first law of thermodynamics states that $\langle W_{1}\rangle+\langle W_{2}\rangle+\langle Q_{M}\rangle+\langle Q_{C}\rangle=0$. From this, the total amount of work extracted is given by,
\begin{equation}
\langle W\rangle:=\langle Q_{M}\rangle+\langle Q_{C}\rangle=-(\langle W_{1}\rangle+\langle W_{2}\rangle).
\end{equation}

By a unital channel $\Phi$, we mean that for any finite dimension $d$, $\Phi(\mathbb{1}_{d}/d)=\mathbb{1}_{d}/d$, where $\mathbb{1}_{d}$ is the $d\times d$ identity matrix.

\begin{remark}
Before we continue, let's clarify the nature of the energy change after applying the unital channel $\Phi$, i.e., the nature of $\langle E_{3}\rangle-\langle E_{2}\rangle$. When the channel $\Phi$ is unitary, then in this case, the energy change is work, since $\Phi$ preserves the entropy of the system; see Refs. \cite{Alipour,Chen,Bernardo,Behzadi,Ahmadi}. Please note that unitary operations are an example of unital channels, since for an arbitrary $U$ we have $U(\mathbb{1}_{d}/d)U^{\dagger}=\mathbb{1}_{d}/d$. In this paper, on the other hand, we are considering unital channels that change the entropy of the system, and thus the energy provided is \textit{heat}.  An example of quantum unital channels that have been used to fuel quantum thermal machines (see, e.g., Refs. \cite{Kim,Das,Elouard,Buffoni,Anka,Behzadi,Prasanna1, KimT,Sagawa,Jordan,Chand,Lisboa,Shanhe, Bresque,JuzarSong,JordanNM,Yamamoto,BiswasP,SantosJ,BhandariJ,Perna,Robert}) is quantum projective measurements. 
\end{remark}

\subsection{Dephased engine}\label{de}

As above, consider as a working medium a quantum system with a finite dimension $d$, i.e., $d<\infty$. The Hamiltonian of the working medium can be expressed as follows:
\begin{equation}
H:=\sum_{i=1}^{d}\nu_{i}\Pi_{i}=\sum_{i=1}^{d}\nu_{i}|i\rangle\langle i|.\label{ehehehe}
\end{equation}
Where $\nu_{i}$ and $\Pi_{i}$ are, respectively, the corresponding eigenvalues and (assumed) rank 1 eigenprojectors of $H$. Even though not explicitly written, the eigenvalues and eigenprojectors can be functions of parameters such as time, coupling, or magnetic field. But let us not go into details.

Now let's define what we mean by the $\textit{dephased engine}$. Following the two-point measurement (TPM) scheme \cite{Campisi1, Esposito,Tasaki}, the system is now monitored between the \textit{strokes}. Below, we use the $\textit{(un)monitored}$ and $\textit{(un)dephased}$ engines interchangeably. First measuring the initial Hamiltonian $H_{1}$ at $\mathbf{A}$ in Fig. \ref{Otto}, gives one of its eigenvalues denoted from now on by $\nu_{n}^{(1)}$ (see, Eq. (\ref{ehehehe})), where the superscript $(1)$ refers to the fact that we are measuring $H_{1}$. In this case, according to the postulates of quantum mechanics, the system would be projected onto $\Pi_{n}^{(1)}\left(:=|n\rangle_{11}\langle n|\right)$. Following the same reasoning, we denote the measured eigenvalues of the Hamiltonians at $\mathbf{B}$, $\mathbf{C}$, and $\mathbf{D}$ in Fig. \ref{Otto}, by $\nu_{m}^{(2)}$, $\nu_{k}^{(2)}$, and $\nu_{l}^{(1)}$, respectively. Since we are considering projective measurements, the system would be in one of the eigenstates of the measured Hamiltonian after the measurement. Note that the quantum numbers $n$, $m$, $k$, and $l$ vary from $1$ into $d$. We call four successive measured energies, i.e., $\nu_{n}^{(1)}$, $\nu_{m}^{(2)}$, $\nu_{k}^{(2)}$, and $\nu_{l}^{(1)}$, $\textit{a stochastic cycle}$. Then, the collapsed states of the system in an arbitrary stochastic cycle are given as follows:
\begin{equation}
\Pi_{n}^{(1)}\rightarrow \Pi_{m}^{(2)}\rightarrow \Pi_{k}^{(2)} \rightarrow \Pi_{l}^{(1)}.\label{scc}
\end{equation}
In this case, the $\textit{stochastic work}$ and the $\textit{stochastic heat absorbed}$, are defined, respectively, as follows: $W:=\nu_{n}^{(1)}-\nu_{m}^{(2)}+\nu_{k}^{(2)}-\nu_{l}^{(1)}$, $Q_{M}:=\nu_{k}^{(2)}-\nu_{m}^{(2)}$.
On the other hand, by imposing the first law of thermodynamics at the stochastic level, we can define $Q_{C}:=\nu_{n}^{(1)}-\nu_{l}^{(1)}$. This quantity reflects only the true average of the heat released and not the higher cumulants. See Ref. \cite{Abdelkader2} for more details on this statement.
 Further, the probability (denoted by $p_{n,m,k,l}$ from now on) of an arbitrary stochastic cycle to be followed by the working medium is given by
\begin{equation}
p_{n,m,k,l}:=\frac{e^{-\beta \nu_{n}^{(1)}}}{Z}|{}_{2}\langle m|U|n\rangle_{1}|^{2}\left(
{}_{2}\langle k |\Phi(|m\rangle_{22}\langle m|)|k\rangle_{2}\right)
 |{}_{1}\langle l|V|k\rangle_{2}|^{2} .\label{pnmkl}
\end{equation}
Here we employ the fact that $\Pi_{m}^{(2)}$, $\Pi_{k}^{(2)}$, and $\Pi_{l}^{(1)}$ are considered to be rank 1 and thus defined similarly to $\Pi_{n}^{(1)}$. Furthermore, it should be emphasized that $p_{n,m,k,l}$ verifies the next two properties: $0\leq p_{n,m,k,l} \leq1$ for all stochastic cycles, and $\sum_{n,m,k,l} p_{n,m,k,l}=1$.
	
Equation (\ref{pnmkl}) can be further compressed into a nice form written as,
\begin{equation}
p_{n,m,k,l}=\mathrm{Tr}\left[\Pi_{l}^{(1)}V\Pi_{k}^{(2)}\Phi\left(\Pi_{m}^{(2)}U\Pi_{n}^{(1)}\rho_{1}\Pi_{n}^{(1)}U^{\dagger}\Pi_{m}^{(2)}\right)\Pi_{k}^{(2)}V^{\dagger} \right].\label{pnmkl1}
\end{equation}

We define the next four stochastic energies by: $E_{1}:=\nu_{n}^{(1)}$, $E_{2}:=\nu_{m}^{(2)}$, $E_{3}:=\nu_{k}^{(2)}$, and $E_{4}:=\nu_{l}^{(1)}$. The joint probability distribution (PD) that describes the measured energies $E_{1}$, $E_{2}$, $E_{3}$, and $E_{4}$ of our engine is given by
\begin{equation}
\begin{split}
P(E_{1},E_{2},E_{3},E_{4}) :=\sum_{n,m,k,l}p_{n,m,k,l}\delta(E_{1}-\nu_{n}^{(1)})\delta(E_{2}-\nu_{m}^{(2)})\delta(E_{3}-\nu_{k}^{(2)})\delta(E_{4}-\nu_{l}^{(1)}).
\end{split}
\label{PWQ}
\end{equation}
From a computational point of view, it is easier to work with the characteristic function than with the PD. Nevertheless, note that they contain the same information. The CF, denoted from now on by $\chi(\gamma_{1},\gamma_{2},\gamma_{3},\gamma_{4})$, is the fourier transform of $P(E_{1},E_{2},E_{3},E_{4})$, i.e.
\begin{equation}
\chi(\gamma_{1},\gamma_{2},\gamma_{3},\gamma_{4}):=\int P(E_{1},E_{2},E_{3},E_{4}) e^{i(\gamma_{1}\nu_{n}^{(1)}+\gamma_{2}\nu_{m}^{(2)}+\gamma_{3}\nu_{k}^{(2)}+\gamma_{4}\nu_{l}^{(1)})} dE_{1}dE_{2}dE_{3}dE_{4}.
\end{equation}
Here, $\gamma_{1}$, $\gamma_{2}$, $\gamma_{3}$, and $\gamma_{4}$ are the Fourier conjugates of $E_{1}$, $E_{2}$, $E_{3}$, and $E_{4}$, respectively. Since the outcomes of our engine are discrete, the integral would be replaced by a summation. Therefore, after simple algebra, one can arrive at the next expression:
\begin{equation}
\chi(\gamma_{1},\gamma_{2},\gamma_{3},\gamma_{4})  =\sum_{n,m,k,l}\mathrm{Tr}\left[\Pi_{l}^{(1)}V\Pi_{k}^{(2)}\Phi\left(\Pi_{m}^{(2)}U\Pi_{n}^{(1)}\rho_{1}\Pi_{n}^{(1)}U^{\dagger}\Pi_{m}^{(2)}\right)\Pi_{k}^{(2)}V^{\dagger} \right]
e^{i(\gamma_{1}\nu_{n}^{(1)}+\gamma_{2}\nu_{m}^{(2)}+\gamma_{3}\nu_{k}^{(2)}+\gamma_{4}\nu_{l}^{(1)})}.\label{Chi}
\end{equation}
The cumulants of $E_{1}$, $E_{2}$, $E_{3}$, and $E_{4}$ and their covariances can be derived from Eq. (\ref{Chi}) as follows:
\begin{equation}
\llangle E_{1}^{s1}E_{2}^{s2}E_{3}^{s3}E_{4}^{s4}\rrangle_{c}:=\frac{\partial^{s1}\partial^{s2}\partial^{s3}\partial^{s4}\rm log(\chi(\gamma_{1},\gamma_{2},\gamma_{3},\gamma_{4}))}{\partial(i\gamma_{1})^{s1}\partial(i\gamma_{2})^{s2}\partial(i\gamma_{3})^{s3}\partial(i\gamma_{4})^{s4}}\bigg\rvert_{\gamma_{1},\gamma_{2},\gamma_{3},\gamma_{4}=0}.\label{TTT}
\end{equation}  
Here, $s_{1}$, $s_{2}$, $s_{3}$, and $s_{4}$ are positive integers, $\log$ denotes the natural logarithm, and the subscript $c$ refers to the fact that $\llangle x^s\rrangle_c$ is the corresponding cumulant of $x$ and not its moment.

We should emphasize that the first derivative, with respect to, e.g., $\gamma_{1}$ in Eq. (\ref{Chi}), gives the average of $E_{1}$,
\begin{equation} 
\llangle E_{1}\rrangle_{c}=\frac{\partial\rm log(\chi(\gamma_{1},\gamma_{2},\gamma_{3},\gamma_{4}))}{\partial(i\gamma_{1})}\bigg\rvert_{\gamma_{1},\gamma_{2},\gamma_{3},\gamma_{4}=0},
\end{equation}
while the second derivative gives its variance (also so-called $\textit{second cumulant}$)
\begin{equation} 
\llangle E_{1}^{2}\rrangle_{c}=\frac{\partial^{2}\rm log(\chi(\gamma_{1},\gamma_{2},\gamma_{3},\gamma_{4}))}{\partial(i\gamma_{1})^{2}}\bigg\rvert_{\gamma_{1},\gamma_{2},\gamma_{3},\gamma_{4}=0}. 
\end{equation} 
In this latter equation, if we eliminate the log function, we obtain the $\textit{second moment}$ of $E_{1}$.The same thing applies to $\gamma_{2}$, $\gamma_{3}$, and $\gamma_{4}$. Other derivatives give the correlations between the stochastic energies. For example, 
\begin{equation} 
\llangle E_{1}E_{2}\rrangle_{c}=\frac{\partial^{2}\rm log(\chi(\gamma_{1},\gamma_{2},\gamma_{3},\gamma_{4}))}{\partial(i\gamma_{1})\partial(i\gamma_{2})}\bigg\rvert_{\gamma_{1},\gamma_{2},\gamma_{3},\gamma_{4}=0},\label{cov} 
\end{equation} 
gives the covariance of $E_{1}$ and $E_{2}$. Please note that in what follows we drop the subscript $c$ from the averages.

We note that one can compress Eq. (\ref{Chi}) to obtain,
\begin{equation}
\begin{split}
\chi(\gamma_{1},\gamma_{2},\gamma_{3},\gamma_{4})=\mathrm{Tr}\left[e^{i\gamma_{4}H_{1}} V e^{i\gamma_{3}H_{2}} \Delta_{2}\left(\Phi \left(\Delta_{2}\left( e^{i\gamma_{2}H_{2}} U e^{i(\gamma_{1}+i\beta)H_{1}} U^{\dagger}\right)\right)\right) V^{\dagger} \right]/Z.\label{chiF}
\end{split}
\end{equation}
 Here $\Delta_{2}\left(\cdot\right):=\sum_{i=1}^{d}|i\rangle_{22}\langle i|(\cdot)|i\rangle_{22}\langle i|$ is a complete dephasing channel in the eigenbasis of the Hamiltonian $H_{2}$. This shows the effect of the projective measurement. From Eq. (\ref{chiF}) the average energies along the cycle are given by: $\llangle E_{1}\rrangle=\mathrm{Tr}\left[ \rho_{1} H_{1} \right]$, $\llangle E_{2}\rrangle=\mathrm{Tr}\left[\rho_{2}H_{2}\right]$, $\llangle E_{3}\rrangle=\mathrm{Tr}\left[\Phi(\Delta_{2}(U\rho_{1}U^{\dagger}))H_{2}\right]$ and $\llangle E_{4}\rrangle=\mathrm{Tr}\left[V \Delta_{2}(\Phi(\Delta_{2}(U\rho_{1}U^{\dagger})))V^{\dagger}H_{1}\right]$. We see that only the third and fourth averages are affected by measurement. That is, the averages $\llangle E_{3}\rrangle$ and $\llangle E_{4}\rrangle$ are computed with respect to the evolution of the dephased states. This is because the initial state is incoherent with respect to $H_{1}$, and thus the average energy $\llangle E_{2}\rrangle$ is not affected.

The cumulants of work and heats follow from Eq. (\ref{Chi}) as follows:
\begin{enumerate}
\item $Q_{M}$ cumulants: we set $\gamma_{1}=\gamma_{4}=0$ and $-\gamma_{2}=\gamma_{3}=\gamma_{M}$,
\item $Q_{C}$ cumulants: we set $\gamma_{2}=\gamma_{3}=0$ and $\gamma_{1}=-\gamma_{4}=\gamma_{C}$,
\item $W$ cumulants: we set $\gamma_{1}=-\gamma_{4}=\gamma_{W}$ and $\gamma_{3}=-\gamma_{2}=\gamma_{W}$.
\end{enumerate}\label{defin1}
Here, $\gamma_{M}$, $\gamma_{C}$, and $\gamma_{W}$ are the Fourier conjugates of $Q_{M}$, $Q_{C}$, and $W$, respectively. The work and heat averages for the monitored engine are defined as follows: 
\begin{equation}
\begin{split}
& 
\llangle Q_{M}\rrangle  =\mathrm{Tr}\left[\Phi(\Delta_{2}(U\rho_{1}U^{\dagger}))H_{2}\right]-\mathrm{Tr}\left[U\rho_{1}U^{\dagger}H_{2}\right],
\\ &
\llangle Q_{C}\rrangle  =\mathrm{Tr}\left[\rho_{1} H_{1}\right]-\mathrm{Tr}\left[V \Delta_{2}(\Phi(\Delta_{2}(U\rho_{1}U^{\dagger})))V^{\dagger}H_{1}\right],
\\ &
\llangle W\rrangle  =\mathrm{Tr}\left[\rho_{1} H_{1}\right]-\mathrm{Tr}\left[V \Delta_{2}(\Phi(\Delta_{2}(U\rho_{1}U^{\dagger})))V^{\dagger}H_{1}\right]+\mathrm{Tr}\left[\Phi(\Delta_{2}(U\rho_{1}U^{\dagger}))H_{2}\right]-\mathrm{Tr}\left[U\rho_{1}U^{\dagger}H_{2}\right].
\end{split}
\label{FFF}
\end{equation}

\subsection{Cumulants of the undephased engine}

In the previous subsection, we showed how one could obtain the cumulants of the thermodynamic quantities from the characteristic function. However, from Eqs. (\ref{chiF}) and (\ref{FFF}), we see that the cumulants are computed with respect to the dephased states of the system. This is a result of \textit{measuring projectively} the working medium between the strokes. Actually, in this case, all the coherence created in the energy eigenbasis gets killed by the measurement. Thus, the important question now is: how can one compute the cumulants of the undephased engine? 

A quick answer to this question is: for the averages, one can see that when we eliminate the dephasing channel in Eq. (\ref{chiF}), then we would obtain the average energies of the undephased engine, given in section (\ref{ude}). This would motivate us to do the same thing for higher cumulants; see Appendix \ref{hjkjkl} for more details. In this situation, the CF Eq. (\ref{Chi}) becomes,
\begin{equation}
\begin{split}
\chi_{UDE}(\gamma_{1},\gamma_{2},\gamma_{3},\gamma_{4}):=\mathrm{Tr}\left[e^{i\gamma_{4}H_{1}} V e^{i\gamma_{3}H_{2}} \Phi \left( e^{i\gamma_{2}H_{2}} U e^{i(\gamma_{1}+i\beta)H_{1}} U^{\dagger}\right) V^{\dagger} \right]/Z.\label{chiFF}
\end{split}
\end{equation}
The subscript $UDE$ refers to the fact that this CF is for the undephased engine. Similarly to Eq. (\ref{PWQ}), the CF in Eq. (\ref{chiFF}) follows from the next $\textit{ quasiprobability distribution}$,
\begin{equation}
\begin{split}
P_{UDE}(E_{1},E_{2},E_{3},E_{4}) 
:=\sum_{n,m,k,l}p_{n,m,k,l}^{UDE}\delta(E_{1}-\nu_{n}^{(1)})\delta(E_{2}-\nu_{m}^{(2)})\delta(E_{3}-\nu_{k}^{(2)})\delta(E_{4}-\nu_{l}^{(1)}),
\end{split}
\label{pnmkl2}
\end{equation}
with $p_{n,m,k,l}^{UDE}:=\mathrm{Tr}\left[\Pi_{l}^{(1)}V\Pi_{k}^{(2)}\Phi\left(\Pi_{m}^{(2)}U\Pi_{n}^{(1)}\rho_{1}U^{\dagger}\right)V^{\dagger} \right]$. The fact that $P_{UDE}(E_{1},E_{2},E_{3},E_{4}) $ is a $\textit{quasiprobaility}$ and not a true probability, is because $p_{n,m,k,l}^{UDE}$ is not positive in general, and even more, it can be a complex number. Note that Eq. (\ref{pnmkl2}) is known in the literature as $\textit{Kirwkood-Dirac}$ quasiprobability \cite{Kirkwood,Dirac,MatteoLevy,NicoleY,Chiara,Francica2,Francica1,HTQuan,Dipojono,Santini,Belenchia,Bievre} — {the real part of $p_{n,m,kl}^{UDE}$ is known as the Margenau-Hill distribution \cite{Hill}. However, note that this quasiprobability still satisfies, $\sum_{n,m,k,l} p_{n,m,k,l}^{UDE}=1$.

The cumulants of the energies of the undephased engine follow from Eq. (\ref{chiFF}) as follows:
\begin{equation}
\langle E_{1}^{s1}E_{2}^{s2}E_{3}^{s3}E_{4}^{s4}\rangle_{c}:=\frac{\partial^{s1}\partial^{s2}\partial^{s3}\partial^{s4}\rm log(\chi_{UDE}(\gamma_{1},\gamma_{2},\gamma_{3},\gamma_{4}))}{\partial(i\gamma_{1})^{s1}\partial(i\gamma_{2})^{s2}\partial(i\gamma_{3})^{s3}\partial(i\gamma_{4})^{s4}}\bigg\rvert_{\gamma_{1},\gamma_{2},\gamma_{3},\gamma_{4}=0}.\label{EEEE}
\end{equation} 
Note that we reserve the notations $\langle \cdot\rangle$ and $\langle (\cdot)^{2}\rangle_{c}$, respectively, for the first and second cumulants. Again, note that we dropped the subscript $c$ from the averages. Finally, note that the cumulants of work and heat follow from Eq. (\ref{EEEE}) similarly to the dephased engine \ref{defin1}. However, note that the second cumulants of $W$ and $Q_{M}$ have a real and imaginary part. In what follows, this statement will be clear.

\subsection{Can the cycle work as a refrigerator?}\label{Entropy}
In Ref. \cite{Abdelkader}, we have shown that when the working medium is a single qubit, and independently of the parameters, the system can never work as a $\textit{refrigerator}$. The question now is: is this valid for any finite-dimension system?

To answer this question, we should define the next stochastic quantity, $\Sigma:=-\beta(\nu_{n}^{(1)}-\nu_{l}^{(1)})$. This latter is simply the $\textit{stochastic entropy production} $. When taking its average, we obtain $\langle \Sigma\rangle=-\beta\langle Q_{C}\rangle$ for the undephased engine and $\llangle \Sigma\rrangle=-\beta\llangle Q_{C}\rrangle$ for the dephased engine. For the monitored engine, one can show that:
\begin{equation}
\llangle e^{-\Sigma}\rrangle:=\sum_{n,m,k,l}e^{\beta(\nu_{n}^{(1)}-\nu_{l}^{(1)})} p_{n,m,k,l}=1. \label{HHT}
\end{equation}
Since $p_{n,m,k,l}$ is always positive, we can apply Jensen inequality to obtain $e^{-\llangle \Sigma\rrangle}\leq \llangle e^{-\Sigma}\rrangle=1$, from which one obtains that $\llangle \Sigma\rrangle\geq0$. From the fact that $\llangle \Sigma\rrangle=-\beta \llangle Q_ {C}\rrangle$, and for $\beta\geq0$ one has $\llangle Q_{C}\rrangle\leq0$. Thus, we proved that regardless of the dimension, whenever the channel is unital, the system cannot work as a refrigerator. When $\beta<0$, we have $\llangle Q_{C}\rrangle\geq0$. But note that this is not in contradiction with the second law since this is not a refrigerator but a heat engine with unit efficiency; see Refs. \cite{Struchtrup,Warren}.

Similarly to the dephased engine, for the undephased engine we have,
\begin{equation}
\langle e^{-\Sigma}\rangle:=\sum_{n,m,k,l}e^{\beta(\nu_{n}^{(1)}-\nu_{l}^{(1)})} p_{n,m,k,l}^{UDE}=1.
\label{HHHT}
\end{equation}
But note that in this case, the Jensen inequality cannot be applied directly to $p_{n,m,k,l}^{UDE}$ since it is not positive and even it can be a complex number. That is, care should be taken here. But one can show that there is a beautiful way to prove that $\langle \Sigma\rangle\geq0$ even for the unmonitored engine. We have,
\begin{equation}
\langle e^{-\Sigma}\rangle=\sum_{n,m,k,l}e^{\beta(\nu_{n}^{(1)}-\nu_{l}^{(1)})} p_{n,m,k,l}^{UDE}=\sum_{n,l}e^{\beta(\nu_{n}^{(1)}-\nu_{l}^{(1)})} \sum_{m,k}p_{n,m,k,l}^{UDE}=1.
\end{equation}
Here $\sum_{m,k}p_{n,m,k,l}^{UDE}=\mathrm{Tr} \left[ \Pi_{l}^{(1)}V \Phi( U\Pi_{n}^{(1)}\rho_{1}U^{\dagger})V^{\dagger} \right]$. One can prove that the latter is a true probability since it is always positive. Therefore, in this case, we can apply Jensen inequality safely, and we obtain $\langle \Sigma\rangle\geq0$, from which we have $\langle Q_{C}\rangle\leq0$. This proves that the heat exchanged with the cold bath is always $\leq0$ in accordance with the second law. This generalizes the results of Refs. \cite{Abdelkader,Shanhe}, where the proof was only limited to qubit systems. In contrast, this proof is valid for an arbitrary finite-dimensional working medium.

Now let's return to the interpretation of $\sum_{m,k}p_{n,m,k,l}^{UDE}$. The latter's expression is given by:
\begin{equation}
\sum_{m,k}p_{n,m,k,l}^{UDE}=\mathrm{Tr} \left[ \Pi_{l}^{(1)}V \Phi( U\Pi_{n}^{(1)}\rho_{1}U^{\dagger})V^{\dagger} \right]=\mathrm{Tr} \left[ \Pi_{l}^{(1)}V \Phi( U\Pi_{n}^{(1)}\rho_{1}\Pi_{n}^{(1)}U^{\dagger})V^{\dagger} \right].\label{inter}
\end{equation}
Here we used the fact that $\Pi_{n}^{(1)}\rho_{1}=\rho_{1}\Pi_{n}^{(1)}$, since $\rho_{1}$ is incoherent in the eigenbasis of $H_{1}$. The interpretation of Eq. (\ref{inter}) is nothing but: first we projectively measure the working medium at $\mathbf{A}$ in the cycle (see Fig. \ref{Otto}), then evolve the projected state by $V\Phi(U(.)U^{\dagger})V^{\dagger}$ and finally projectively measure the system again at $\mathbf{D}$ (see Fig. \ref{Otto}). In this case, even though the system is projectively measured at the beginning and at the end of the cycle, these two measurements do not influence the cumulants of $Q_{C}$. This is because the cumulants are given with respect to the undephased states and not the dephased ones. This shows that it is the projective measurement at $\mathbf{B}$ and $\mathbf{C}$ that influences the cumulants of work and heats. This is because they kill the coherence created in the eigenbasis of the Hamiltonian $H_{2}$.

For further details about entropy production, see Appendix \ref{PROD}.

\section{A qubit as a working medium}\label{casestu}

The above discussion was for an arbitrary finite-dimensional working substance. In what follows, we limit ourselves to a single qubit. However, even with a single qubit, we will see that the results are not trivial. 

For a two-level system, $|e\rangle$ and $|g\rangle$ are, respectively, the excited and ground states. The first Hamiltonian in the Otto cycle, see Fig. \ref{Otto}, is given by $H_{1}:=\nu_{1}\sigma_{z}$, where $\sigma_{z}$ is the $\textit{z}$-Pauli operator. The expression of $H_{1}$ in the computational basis $\{|e\rangle,|g\rangle\}$ is,
\begin{equation}
H_{1}=\nu_{1}(|e\rangle\langle e|-|g\rangle\langle g|).\label{sigmamz}
\end{equation}
We see that the gap between the excited and ground states is $2\nu_{1}$. For the second Hamiltonian, i.e., $H_{2}$, we choose it to be $H_{2}:=\nu_{2}\sigma_{x}$, where $\sigma_{x}$ is the $\textit{x}$-Pauli operator. In what follows, we take $\nu_{2}\geq\nu_{1}$. $H_{2}$ can be written as follows:
\begin{equation}
H_{2}=\nu_{2}(|e\rangle\langle g|+|g\rangle\langle e|).\label{sigmamx}
\end{equation}
Its eignestates are $|+\rangle:=(|e\rangle+|g\rangle)/\sqrt{2}$ and $|-\rangle:=(|e\rangle-|g\rangle)/\sqrt{2}$, with their corresponding eigenvalues being $\nu_{2}$ and $-\nu_{2}$, respectively. Following Refs. \cite{Zener1,Zener2,Zener3,JuzarSong,Denzler}, the unitary operator $U$ that governs the first adiabatic transformation, i.e., $\mathbf{A}\rightarrow \mathbf{B}$ in Fig. \ref{Otto}, is given by $U:=\sqrt{1-\delta}\left(e^{-i\phi}|+\rangle\langle e|+e^{i\phi}|-\rangle\langle g|\right)+\sqrt{\delta}\left(|-\rangle\langle e|-|+\rangle\langle g|\right)$. And written in the basis $\{|e\rangle,|g\rangle\}$ we have,
\begin{equation}
\begin{split}
U=\begin{pmatrix}
\sqrt{1-\delta}e^{-i\phi}+\sqrt{\delta} & \sqrt{1-\delta}e^{i\phi}-\sqrt{\delta} \\
\sqrt{1-\delta}e^{-i\phi}-\sqrt{\delta} & -\sqrt{1-\delta}e^{i\phi}-\sqrt{\delta} \\
\end{pmatrix}/\sqrt{2}.
\end{split}\label{Unit}
\end{equation}
$\delta$ $\in[0,1]$ here is the degree of the non-adiabaticity, and $\phi$ $\in[0,2\pi]$ is a phase. To see that $\delta$ is a transition probability between the eigenstates of $H_{1}$ and $H_{2}$ one can show that
\begin{equation}
|\langle +| U |g \rangle|^{2}=|\langle -| U |e \rangle|^{2}=\delta,
\end{equation}
and,
\begin{equation}
|\langle -| U |g \rangle|^{2}=|\langle +| U |e \rangle|^{2}=1-\delta.
\end{equation}
When $\delta=0$, then in this case we are in the adiabatic regime. In this situation, the unitary operator reduces to $e^{-i\phi}|+\rangle\langle e|+e^{i\phi}|-\rangle\langle g|$. We see that after $U$ acts on $\rho_{1}$, only it changes the eigenstates without changing the populations of the ground state and the excited state since the initial state is a thermal state. On the other hand, when $\delta=1$, in this case, $U$ becomes a swap operator since it reduces to $|-\rangle\langle e|-|+\rangle\langle g|$. From these two cases, we see that the phase $\phi$ would be relevant only when $0<\delta<1$.

 Limiting ourselves to the case when the cycle is time-reversal symmetric, the unitary operator $V$ characterizing the second adiabatic stroke, i.e., the stroke $\mathbf{C}\rightarrow \mathbf{D}$, is defined by, $V:= C^{\ast}U^{\dagger}C =\sqrt{1-\delta}\left(e^{-i\phi}|e\rangle\langle +|+e^{i\phi}|g\rangle\langle-|\right)+\sqrt{\delta}\left(|e\rangle\langle -|-|g\rangle\langle+|\right)$, where $C$ here is the complex conjugation operator. The expression of $V$ in the computational basis $\{|e\rangle,|g\rangle\}$ is given by,
\begin{equation}
\begin{split}
V=\begin{pmatrix}
\sqrt{1-\delta}e^{-i\phi}+\sqrt{\delta} & \sqrt{1-\delta}e^{-i\phi}-\sqrt{\delta} \\
\sqrt{1-\delta}e^{i\phi}-\sqrt{\delta} & -\sqrt{1-\delta}e^{i\phi}-\sqrt{\delta} \\
\end{pmatrix}/\sqrt{2}.
\end{split}\label{Unit2}
\end{equation}
Note that for the purpose of the paper, we do not need to specify the expressions of $\delta$ and $\phi$ as a function of the parameters. Further note that the operators $U$ and $V$ are elements of the special uniray group $SU(2)$.

\begin{figure}[hbtp]
\centering
\includegraphics[scale=0.48]{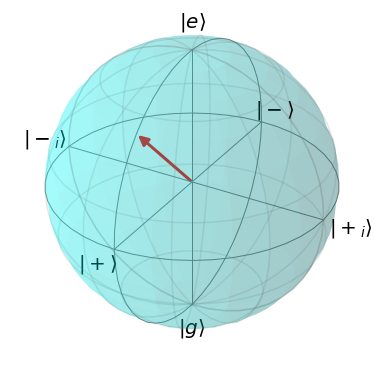}
\caption{Scheme of the Bloch sphere. The vector drawn here is $|\pi_{1}\rangle$ for $\chi=0$ and $\phi=\pi/4$. $|+_{i}\rangle=(|e\rangle+i|g\rangle)/\sqrt{2}$ and $|-_{i}\rangle=(|e\rangle-i|g\rangle)/\sqrt{2}$ are the eigenstates of $\sigma_{y}$, i.e., the \textit{y}-Pauli operator. The states $|+\rangle$, $[-\rangle$, $|+_{i}\rangle$, and $|-_{i}\rangle$ are some of the states that belong to the $\textit{xy}$-plane. Similarly, the states $|+\rangle$, $[-\rangle$, $|e\rangle$, and $|g\rangle$ are one of the states that belong to the $\textit{xz}$-plane, and the states $|+_{i}\rangle$, $[-_{i}\rangle$, $|e\rangle$, and $|g\rangle$ are one of the states that belong to the $\textit{yz}$-plane.}\label{blochh}
\end{figure}

It is necessary to mention that the choice of the Hamiltonians Eqs. (\ref{sigmamz})-(\ref{sigmamx}) is motivated by three reasons: (i) because they were already considered in the experimental implementation of the quantum Otto cycles and in assessing their statistics \cite{Shanhe,JPeterson,Camati,DenzlerR} (ii) $H_{1}$ and $H_{2}$ do not commute, thus coherence would be generated, which is our purpose (of course choosing $H_{2}=\nu_{2}\sigma_{y}$ would also work),  (iii) the eigenstates of $H_{1}$ and $H_{2}$ are independent of the parameters — one can also consider $H_{1}$ and $H_{2}$ to be a linear combination of all Pauli matrices, (cf. Eq. (\ref{qubitHH}) below), but in this case, their eigenstates have a more complicated expression and thus also $U$ and $V$. Thus, the choice of $H_{1}$ and $H_{2}$ is for the ease of the numerical analysis.

For the unital channel fueling the engine, we consider the quantum projective measurement channel $\Phi(\cdot)=\sum_{j=1}^{2} \pi_{j}(\cdot)\pi_{j}$ with $\pi_{1}(:=|\pi_{1}\rangle\langle \pi_{1}|)$, and $\pi_{2}(:=|\pi_{2}\rangle\langle \pi_{2}|)$ are projective operators. By projective, we mean that they verify: $\pi_{1}+\pi_{2}=\mathbb{1}_{2}$, $\pi_{1}\pi_{1}=\pi_{1}$, $\pi_{2}\pi_{2}=\pi_{2}$, and $\pi_{1}\pi_{2}=\mathbb{0}_{2}$, where $\mathbb{0}_{2}$ is the zero operator. The expressions of $|\pi_{1}\rangle$ and $|\pi_{2}\rangle$ are given as follows:
\begin{equation}
|\pi_{1}\rangle=\cos\left(\alpha/2\right)|e\rangle+e^{i\chi}\sin\left(\alpha/2\right)|g\rangle,\label{Pi1}
\end{equation}
and
\begin{equation}
|\pi_{2}\rangle=e^{-i\chi}\sin\left(\alpha/2\right)|e\rangle-\cos\left(\alpha/2\right)|g\rangle,\label{Pi2}
\end{equation}
with $0\leq \alpha \leq\pi$ and $0\leq \chi < 2\pi$. Written in the computational basis, we have 
\begin{equation}
\pi_{1}=\begin{pmatrix}
\cos^{2}(\alpha/2) & \sin(\alpha/2)\cos(\alpha/2)e^{-i\chi} \\
\sin(\alpha/2)\cos(\alpha/2)e^{i\chi} & \sin^{2}(\alpha/2) \\
\end{pmatrix},\label{pi1}
\end{equation}
and,
\begin{equation}
\pi_{2}=\begin{pmatrix}
\sin^{2}(\alpha/2) & -\sin(\alpha/2)\cos(\alpha/2)e^{-i\chi} \\
-\sin(\alpha/2)\cos(\alpha/2)e^{i\chi} & \cos^{2}(\alpha/2) \\
\end{pmatrix}.\label{pi2}
\end{equation}
We note that because the operators in Eqs. (\ref{pi1})-(\ref{pi2}) satisfy $\pi_{1}(\alpha,\chi)=\pi_{2}((-\alpha+\pi,\chi+\pi))$ and $\pi_{2}(\alpha,\chi)=\pi_{1}((-\alpha+\pi,\chi+\pi))$ then all cumulants satisfy this symmetry relation.

In Ref. \cite{Kim} it was shown that a projective measurement channel can fuel the engine with heat. In Ref. \cite{Buffoni} the authors have shown that a quantum projective can make a two-stroke Otto heat engine work in four different modes i.e. engine, refrigerator, accelerator and heater. For further details about quantum thermal machines fueled by quantum projective measurement one is recommended to see \cite{Das,Anka,Prasanna1,Chand,Shanhe, Yamamoto,BiswasP,SantosJ,BhandariJ} and references therein. It is these works that were our motivation in choosing the unital channel as a quantum projective measurement.

At the next values of $\chi$ and $\alpha$, we have the next bases and planes (see Fig. \ref{blochh}): 
\begin{enumerate}
\item $\chi=0$ and $\alpha=\pi/2$ corresponds to the \textit{x}-basis,
\item $\chi=\pi/2$ and $\alpha=\pi/2$ corresponds to the \textit{y}-basis,
\item $\alpha=0$ (and independently of $\chi$) corresponds to the \textit{z}-basis,
\item for arbitrary $\alpha$ and for $\chi=0$ corresponds to the \textit{xz}-plane,
\item for arbitrary $\alpha$ and for $\chi=\pi/2$ corresponds to the $\textit{yz}$-plane,
\item and finally for arbitrary $\chi$ and for $\alpha=\pi/2$ corresponds to the $\textit{xy}$-plane.
\end{enumerate}
This means is that when set, e.g., $\chi=\pi/2$ and $\alpha=\pi/2$, we are measuring the qubit in the $\textit{y}$-basis.

\begin{remark}
Please note that the projective measurement defined by Eqs. (\ref{Unit})-(\ref{Unit2}) is used to fuel the engine with heat. When we were talking about projective measurements in Sec. \ref{1}, we meant the quantum projective measurements applied between the strokes to assess the fluctuations of the thermodynamic quantities.
\end{remark}

\subsection{Main quantities of interest}
After showing how one can derive the cumulants of the dephased and undephased engines, presenting detailed information about the working substance that we apply our analytical results to, and fixing the notation, let's now give in detail the main quantities that we are going to focus on.

In this paper, after deriving our analytical results, our main interest is on the influence of the parameters: $\beta$, $\nu_{1}$, $\nu_{2}$, $\phi$, $\delta$, $\chi$, and $\alpha$, on the next three quantities of the dephased and undephased engines: work averages $\langle W\rangle$ and $\llangle W\rrangle$, efficiencies $\langle \eta\rangle$ and $\llangle \eta\rrangle$ (defined respectively, by $\langle \eta\rangle:=\langle W\rangle/\langle Q_{M}\rangle$ and $\llangle \eta\rrangle:=\llangle W\rrangle/\llangle Q_{M}\rrangle$) and work reliabilities $R_{WUD}(:=\langle W\rangle/\sqrt{\mathrm{Re}\left[\langle W^{2}\rangle_{c}\right]})$ and $R_{WD}(:=\llangle W\rrangle/\sqrt{\llangle W^{2}\rrangle_{c}})$. Here, $\mathrm{Re}$ refers to the fact that we are considering only the real part of the second cumulant of work that follows from Eq. (\ref{chiFF}). Furthermore, note that by $\langle \eta\rangle$ and $\llangle \eta\rrangle$, we do not mean the average of the stochastic efficiency $W/Q_{M}$, since the latter can diverge, see Refs. \cite{Fei,Denzler}. That is, $\langle W/Q_{M}\rangle(\llangle W/Q_{M}\rrangle)$ is different from $\langle W\rangle/\langle Q_{M}\rangle$ ($\llangle W\rrangle/\llangle Q_{M}\rrangle$) that we consider in our paper. In general, we want to know on which basis in Fig. \ref{blochh} we should measure the qubit such that our three quantities can achieve their best possible values.

Finally, one can show from the first law and the second of thermodynamics, i.e., $\llangle \Sigma\rrangle\geq0$ and $\langle \Sigma\rangle\geq0$ (see Subsection \ref{Entropy}), that $\llangle \eta\rrangle\leq1$ and $\langle \eta\rangle\leq1$. In Ref. \cite{Abdelkader}, we proved that for a single qubit, the efficiency of the monitored engine is always $\leq1-\nu_{1}/\nu_{2}$; however, when increasing the dimension of the working medium, efficiency can exceed the Otto bound even without the presence of coherence. For example, in Ref. \cite{Anka}, efficiency attained 1, even without any coherence or correlations. The point is that, while in the single-qubit case, coherence is necessary for efficiency to attain 1, it is not necessary for higher-dimensional working systems.

\section{Dephased engine}\label{dede}
In this section, we focus on the dephased engine.
\subsection{First and second cumulants}

Now let's forget about the unitaries and the unital channel in Sec. \ref{casestu}. Let's consider an arbitrary initial qubit Hamiltonian $H_{1}$ (Fig. \ref{Otto}) given by,
\begin{equation}
H_{1}:=\nu_{1}(|+\rangle_{11}\langle+|-|-\rangle_{11}\langle-|),\label{H11}
\end{equation}
where $|+\rangle_{1}(|-\rangle_{1})$ is the excited state (ground state) of $H_{1}$. Similarly to Eq. (\ref{H11}), the second Hamiltonian in Fig. \ref{Otto} is given by,
\begin{equation}
H_{2}:=\nu_{2}(|+\rangle_{22}\langle+|-|-\rangle_{22}\langle-|)\label{H22}.
\end{equation}
In Appendix \ref{zedine}, we explain what we mean by arbitrary $H_{1(2)}$.

We define the next transition probabilities:
\begin{equation}
\delta':=|_{2}\langle +|U|-\rangle_{1}|^{2},\label{deltap}
\end{equation}
\begin{equation}
\theta:={}_{2}\langle -|\Phi\left(|+\rangle_{22}\langle+|\right)|-\rangle_{2},\label{tethaaaa}
\end{equation}
and
\begin{equation}
\zeta:=|_{1}\langle +|V|-\rangle_{2}|^{2}.\label{zetaa}
\end{equation}
Here $\delta'$ is the transition probability from the ground state of $H_{1}$ into the excited state of $H_{2}$ after applying $U$. $\theta$ and $\zeta$ can be defined similarly. After long but straightforward algebra, the exact analytical expression of the forward CF (see Ref. \cite{Abdelkader2}), Eq. (\ref{Chi}), is given by:
\begin{widetext}
\begin{equation}
\begin{split}
\chi(\gamma_{1},\gamma_{2},\gamma_{3},\gamma_{4})  = & ((1-\delta')(1-\theta)(1-\zeta)\cos((\gamma_{1}+\gamma_{4}+i\beta)\nu_{1}+(\gamma_{3}+\gamma_{2})\nu_{2})
\\ &
+(1-\delta')(1-\theta)\zeta \cos((\gamma_{1}-\gamma_{4}+i\beta)\nu_{1}+(\gamma_{3}+\gamma_{2})\nu_{2})
\\ &
+(1-\delta')\theta \zeta \cos((\gamma_{1}+\gamma_{4}+i\beta)\nu_{1}+(-\gamma_{3}+\gamma_{2})\nu_{2})
\\ &
+(1-\delta')\theta (1-\zeta) \cos((\gamma_{1}-\gamma_{4}+i\beta)\nu_{1}+(-\gamma_{3}+\gamma_{2})\nu_{2})
\\ &
+\delta'\theta (1-\zeta) \cos((\gamma_{1}+\gamma_{4}+i\beta)\nu_{1}+(\gamma_{3}-\gamma_{2})\nu_{2})
\\ &
+\delta'\theta \zeta \cos((\gamma_{1}-\gamma_{4}+i\beta)\nu_{1}+(\gamma_{3}-\gamma_{2})\nu_{2})
\\ &
+\delta'(1-\theta) \zeta \cos((\gamma_{1}+\gamma_{4}+i\beta)\nu_{1}-(\gamma_{3}+\gamma_{2})\nu_{2})
\\ &
+\delta'(1-\theta) (1-\zeta) \cos((\gamma_{1}-\gamma_{4}+i\beta)\nu_{1}-(\gamma_{3}+\gamma_{2})\nu_{2}))/\cos(i\beta\nu_{1}).\label{chiga}
\end{split}
\end{equation}
\end{widetext}
From this equation, one can derive the averages and variances of $Q_{M}$, $Q_{C}$, and $W$. We have:
\begin{center}
\begin{equation}
\begin{split}
& \llangle Q_{M}\rrangle=2(1-2\delta')\theta\nu_{2}\tanh(\beta\nu_{1}),
\end{split}\label{qma}
\end{equation}
\begin{equation}
\llangle Q_{M}^{2}\rrangle_{c}=4\theta\nu_{2}^{2}-\llangle Q_{M}\rrangle^{2},\label{qmv}
\end{equation}
\begin{equation}
\llangle Q_{C}\rrangle=-2(\theta+(1-2\theta)(\delta'+\zeta-2\delta'\zeta))\nu_{1}\tanh(\beta\nu_{1}),\label{qca}
\end{equation}
\begin{equation}
\llangle Q_{C}^{2}\rrangle_{c}=-\llangle Q_{C}\rrangle(2\nu_{1}\coth(\beta\nu_{1})+\llangle Q_{c}\rrangle),\label{qcv}
\end{equation}
\begin{equation}
\begin{split}
\llangle W\rrangle=\llangle Q_{M}\rrangle+\llangle Q_{C}\rrangle & =2((1-2\delta')\theta\nu_{2}-(\theta+(1-2\theta)(\delta'+\zeta-2\delta'\zeta))\nu_{1})\tanh(\beta\nu_{1}),
\end{split}\label{wa}
\end{equation}
and,
\begin{equation}
\begin{split}
\llangle W^{2}\rrangle_{c} & = -2\nu_{1}\coth(\beta\nu_{1})\llangle Q_{C}\rrangle+8\theta(\delta'+\zeta-1)\nu_{1}\nu_{2}+4\theta\nu_{2}^{2}-\llangle W\rrangle^{2}.
\end{split}\label{wv}
\end{equation}
\end{center}

When $\zeta=\delta'$, in Ref \cite{Abdelkader2}, we proved that, $ \llangle W^{2}\rrangle_{c}/\llangle W\rrangle^{2}\geq\llangle Q_{C}^{2}\rrangle_{c}/\llangle Q_{C}\rrangle^{2}\geq2/\llangle \Sigma\rrangle-1 $ and $\llangle Q_{M}^{2}\rrangle_{c}/\llangle Q_{M}\rrangle^{2}\geq\llangle Q_{C}^{2}\rrangle_{c}/\llangle Q_{C}\rrangle^{2}\geq2/\llangle \Sigma\rrangle-1$ independently of the operation regime, i.e., this is valid when the system is working as a heat engine, heater, and accelerator. In the heat engine region, we have shown that
\begin{equation}
\frac{\llangle W^{2}\rrangle_{c}}{\llangle W\rrangle^{2}}\geq\frac{\llangle Q_{M}^{2}\rrangle_{c}}{\llangle Q_{M}\rrangle^{2}}\geq\frac{\llangle Q_{C}^{2}\rrangle_{c}}{\llangle Q_{C}\rrangle^{2}}\geq\frac{2}{\llangle \Sigma\rrangle}-1,\label{mqcqhw1}
\end{equation}
and ,
\begin{equation}
\llangle\eta\rrangle^{2}=\frac{\llangle W\rrangle^{2}}{\llangle Q_{M}\rrangle^{2}}\leq\frac{\llangle W^{2}\rrangle_{c}}{\llangle Q_{M}^{2}\rrangle_{c}}<1.\label{UppLo}
\end{equation}
The fact that the ratio of fluctuations is always less than 1 can be seen from the fact that:
$\llangle Q_{M}^{2}\rrangle_{c}-\llangle W^{2}\rrangle_{c}=(\llangle W\rrangle+\llangle Q_{M}\rrangle)(2\nu_{1}+\llangle Q_{C}\rrangle\tanh(\beta\nu_{1}))\coth(\beta\nu_{1})$. Thus, we see from the heat engine conditions that $\llangle Q_{M}^{2}\rrangle_{c}>\llangle W^{2}\rrangle_{c}$. 
\subsection{The influence of the parameters on work and its reliability and efficiency}
Consider now the case when $\zeta=\delta'$, and let's now analyze the effect of the parameters on our main quantities. From Eqs. (\ref{qma})-(\ref{qmv})-(\ref{qca})-(\ref{qcv})-(\ref{wa})-(\ref{wv}), we have the next conclusions:

\begin{enumerate}
\item \textit{ Influence of  $\beta:$} Independently of the other parameters, we see that increasing  $\beta$ increases work average $\llangle W\rrangle$ and heat absorbed $\llangle Q_{M}\rrangle$, but note that efficiency $\llangle \eta\rrangle$ is always independent of the inverse temperature. The latter is because the inverse temperature has the same effect on work and heat averages; thus, when taking their ratio, the influence cancels out. 

For fluctuations of work and heat, we see that increasing $\beta$ decreases them. The latter can be explained by the fact that while the average work is temperature-dependent, the second moment is not. Fluctuations of a given stochastic quantity are defined by its second moment, $\textit{minus}$ the square of its average. So when increasing $\beta$, we increase the averages without affecting the second moments, thus decreasing fluctuations. This is also in agreement with the intuition that when we lower temperature, i.e., increase $\beta$, the outcomes of the thermodynamic quantities become less random. To make this clear, consider, e.g., the fluctuations of $Q_{M}$. From Eq. (\ref{qmv}), we see that the second moment (i.e., $4\theta\nu_{2}^{2}$) is inverse temperature independent, while the inverse temperature dependency only comes from the square of the average of $Q_{M}$.

From the fact that increasing $\beta$, increases work and decreases its fluctuations, we conclude that increasing $\beta$ would increase the reliability of work, which is \textit{desirable}.

\item \textit{ Influence of  $\delta'$:} In figure \ref{XXCC}, we plot the work average, efficiency, fluctuations of work, and work reliability as a function of $\delta'$ for five values of $\theta$. We see that increasing $\delta'$ decreases work, decreases efficiency, increases work fluctuations, and decreases the reliability of $W$. More specifically, we see that the highest values of work, efficiency, and work reliability are achieved only when $\delta'=0$, and this is independent of the parameters. This shows that non-adiabatic transitions are detrimental to the important thermodynamic quantities of the engine. Already, we proved in Ref. \cite{Abdelkader} that the highest possible efficiency is that of the Otto, and it is achieved in the adiabatic regime. One can prove the same for work and reliability.

\item $\textit{ Influence of } \phi:$ Now consider the unitaries (\ref{Unit}) and (\ref{Unit2}) considered in Sec. \ref{casestu}. First note that since $\delta'$ (Eq. (\ref{deltap})) is independent of $\phi$, the CF (Eq. (\ref{chiga})) is also independent of it. This means that the cumulants are also independent of the phase $\phi$. This is because $\phi$ has an influence only on the off-diagonal elements of the state $\rho_{2}$. And because the latter state is dephased, the phase $\phi$ will not influence the next states, thus having no influence on the cumulants of work and heats.

\item $\textit{ Influence of } \theta:$ From figure \ref{XXCC}, we see that when we increase $\theta$ towards 1/2, it has a positive influence on our work and its reliability and efficiency. We only consider $\theta\leq1/2$ since for the quantum projective measurement channel in Sec. \ref{casestu} we have $0\leq \theta\leq1/2$. Further, the difference between work at $\theta=1/2$ and $0\leq\theta<1/2$ is given by:
\begin{equation}
\llangle W\rrangle_{\theta=1/2}-\llangle W\rrangle=(1-\delta')(1-2\theta)(\nu_{2}-\nu_{1}+2\nu_{1}\delta')\tanh(\beta\nu_{1}).
\end{equation}
From $0\leq\theta\leq1/2$, $\nu_{2}\geq\nu_{1}$, and the condition $\delta'\leq1/2$ for the system to work as a heat engine, we see that $\llangle W\rrangle_{\theta=1/2}\geq\llangle W\rrangle$. From the influence of $\delta'$ and $\theta$, one can show that the maximal amount of the extracted work is achieved when $\theta=1/2$ and $\delta'=0$ and is given by:
\begin{equation}
\llangle W\rrangle=(\nu_{2}-\nu_{1})\tanh(\beta\nu_{1}).\label{maxw}
\end{equation}
For efficiency we have,
\begin{equation}
\llangle \eta\rrangle_{\theta=1/2}-\llangle \eta\rrangle=\frac{4(1-\delta')\delta'(1-2\theta)\nu_{1}}{(1-2\delta')\theta\nu_{2}}\geq0.
\end{equation}
From this equation, we see that in the heat engine region, we have $\llangle \eta\rrangle_{\theta=1/2}\geq\llangle \eta\rrangle$. 

\item $\textit{ Influence of } \nu_{2}:$ Fixing $\nu_{1}$ and independently of the other parameters, one can see that increasing $\nu_{2}$ has a positive influence on work and efficiency. This is because when we fix $\nu_{1}$ and increase $\nu_{2}$, we increase the heat absorbed without affecting the heat released, thus enhancing work and efficiency. On the other hand, increasing $\nu_{2}$ also increases the fluctuations. However, numerically, one can show that $\nu_{2}$ has a positive influence on the reliability of work; that is, even though it increases fluctuations, it also increases the average, such that reliability increases as we increase $\nu_{2}$. Furthermore, from Eqs. (\ref{qma})-(\ref{qmv})-(\ref{qca})-(\ref{qcv}), note that the RFs of $Q_{M}$ and $Q_{C}$ are independent of $\nu_{2}$, thus increasing $\nu_{2}$ does not influence the reliability of the heats. 

\end{enumerate}

\begin{figure}[hbtp]
\centering
\includegraphics[scale=0.9]{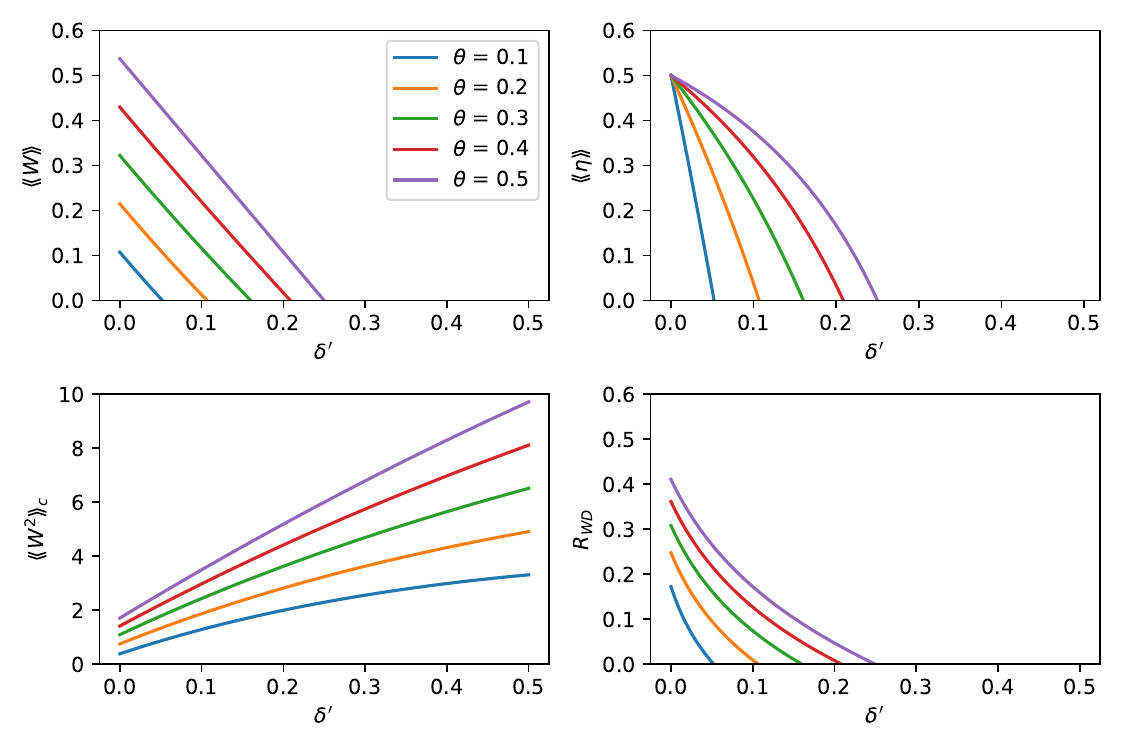}
\caption{Plot of work $\llangle W\rrangle$, efficiency $\llangle \eta\rrangle$, work fluctuations $\llangle W^{2}\rrangle_{c}$, and work reliability $R_{WD}$, as a function of $\delta'$ for five values of $\theta$. The other parameters are: $\beta=0.6$, $\nu_{1}=1$, and $\nu_{2}=2$.}
\label{XXCC}
\end{figure}

\subsection{Upper bounds on work reliability}
In Ref. \cite{Abdelkader2}, we proved that the relative fluctuations of work and heats verify Eq. (\ref{mqcqhw1}). However, one can easily see that the bound $2/\llangle \Sigma\rrangle-1$ becomes useless when $\llangle \Sigma\rrangle$ exceeds 2, since this lower bound $2/\llangle \Sigma\rrangle-1$ becomes negative. More precisely, note that in the lower temperature regime, i.e., when $\beta\rightarrow +\infty$, the average entropy production $\llangle \Sigma\rrangle$ diverges, thus the lower bound in Eq. (\ref{mqcqhw1}) goes to -1, i.e., becomes informationless, since the RFs of work are already $\geq$0 by definition. By defining the reliability of the stochastic heat $Q_{M}$ to be $R_{Q_{M}D}:=\llangle Q_{M}\rrangle/\sqrt{\llangle Q_{M}^{2}\rrangle_{c}}$, one can prove that work reliability is bounded as follows:
\begin{theorem}
For arbitrary $U$, $V$, and $\Phi$ that satisfy $\zeta=\delta'$ and $0\leq\theta\leq1/2$, one can show that work reliability is bounded as follows: 
\begin{equation} 
R_{WD}\leq R_{Q_{M}D}\leq 1. 
\end{equation} 
And the upper bound 1 is attained in the adiabatic limit.
\end{theorem} \label{theorem1E}

\begin{proof} We have:
\begin{equation}
\llangle Q_{C}^{2}\rrangle_{c}-\llangle Q_{C}\rrangle^{2}=4\nu_{1}^{2}(1-\theta+2\delta'(1-\delta')(1-2\theta)\tanh^{2}(\beta\nu_{1}))(\theta+2\delta'(1-\delta')(1-2\theta))\geq0.
\end{equation}
Where the inequality follows from the fact that $\theta\leq1/2$. From $\llangle Q_{C}^{2}\rrangle_{c}-\llangle Q_{C}\rrangle^{2}\geq0$ we have $\llangle Q_{C}^{2}\rrangle_{c}/\llangle Q_{C}\rrangle^{2}\geq1$. Using this proved fact and plugging it into Eq. (\ref{mqcqhw1}), one obatins
\begin{equation}
\frac{\llangle W\rrangle^{2}}{\llangle W^{2}\rrangle_{c}}\leq\frac{\llangle Q_{M}\rrangle^{2}}{\llangle Q_{M}^{2}\rrangle_{c}}\leq\frac{\llangle Q_{C}\rrangle^{2}}{\llangle Q_{C}^{2}\rrangle_{c}}\leq1
\end{equation}
Therefore in the heat engine region, i.e., $\llangle W\rrangle>0$, the reliability of work $R_{WD}$ is upper-bounded as follows:
\begin{equation}
 R_{WD}\leq R_{Q_{M}D}\leq 1. \label{wrqm}
\end{equation}
In the adiabatic limit, one can show that we have,
\begin{equation}
\llangle W^{2}\rrangle_{c}-\llangle W\rrangle^{2}=4\theta(\nu_{2}-\nu_{1})^{2}(1-2\theta\tanh^{2}(\beta\nu_{1}))\geq0.
\end{equation}
Therefore, one can easily show that the highest bound of the work reliability, i.e., 1, can be reached when $\theta=1/2$ and $\beta=+\infty$. In this case, we have $R_{WD}=R_{Q_{M}D}=1$.
\end{proof}

Because the unitaries and the unital channel of Sec. \ref{casestu} satisfy $\zeta=\delta'$ and $0\leq\theta\leq1/2$,
from this theorem, we see that the work reliability is upper-bounded by 1. From Eq. (\ref{mqcqhw1}) one can also derive an upper bound on work reliability using $2/\llangle \Sigma\rrangle-1$. However, this bound has the drawback that it becomes undefined when the average entropy production exceeds 2. On the other hand, our upper new bound does not suffer from this issue.

\subsection{The main features of the dephased engine}\label{mfde}
Let's conclude this section by stating all the main features of the dephased engine:

\begin{enumerate} 
\item For $0\leq\delta'\leq1/2$, the system can either work as a heat engine or an accelerator. For $\delta'\geq1/2$, only a heater is possible.

\item The average work is $\leq0$ when $\nu_{2}=\nu_{1}$. Thus, we need $\nu_{2}>\nu_{1}$ for a positive work condition. See our previous work in Ref. \cite{Abdelkader}.

\item The maximum of work, efficiency, and work reliability are achieved when $\delta'=0$, $\theta=1/2$, and $\beta=+\infty$. The maximum of work, efficiency, and work reliability are, respectively: $\llangle W\rrangle=\nu_{2}-\nu_{1}$ (see Eq. (\ref{maxw})), $\llangle\eta\rrangle=1-\nu_{1}/\nu_{2}$, and $R_{WD}=1$. What this shows is that our three important quantities are strictly monotonically decreasing as we increase $\delta'$ towards 1/2; see, e.g., Fig. \ref{XXCC}.

\item From Eqs. (\ref{pi1})-(\ref{pi2}), the expression of $\theta$ (i.e., Eq. (\ref{tethaaaa})) is given by, 
\begin{equation} \theta=(1-\cos^{2}(\chi)\sin^{2}(\alpha))/2(\leq1/2).
\label{thetta} 
\end{equation}
When $\theta=1/2$, this corresponds to the case when we set $\chi=\pi/2$ and let $\alpha$ be arbitrary. This is nothing but the $\textit{yz}$-plane in the Bloch sphere; see figure \ref{blochh}. In this case, we see that any quantum projective measurement in the $\textit{yz}$-plane is better.

\item In general, we see that lowering $\delta'$ towards 0, increasing $\theta$ towards 1/2, increasing $\beta$ (i.e., $\textit{low-temperature regime}$), and increasing $\nu_{2}$ have a positive influence on work and its reliability and efficiency. On the other hand, the phase $\phi$ does not influence the cumulants because of the measurement between the strokes. 
\end{enumerate}

\section{undephased engine}\label{undeunde}
In the previous section, we showed in detail which parameters should be increased and which should not, for a positive influence on our three main quantities. In this section, we give our analytical and numerical results for the undephased engine, and we show in detail which features are present, like in the case of the dephased engine, and what advantage quantum coherence can provide.
\subsection{Expression of the average work and heat and the contributions coming from coherence}

Before we compute the expressions of work and heat for arbitrary unitaries and unital channels, considering the unitaries and the unital channel in Sec. \ref{casestu}, the expression of the heat absorbed $\langle Q_{M}\rangle$ for the undephased engine is given by
\begin{equation}
\begin{split}
 \langle Q_{M}\rangle& =(1-2\delta)(1-\cos^{2}(\chi)\sin^{2}(\alpha))\nu_{2}\tanh(\beta\nu_{1}) 
\\ &
+\nu_{2}\sqrt{\delta(1-\delta)}\sin(\alpha)(\sin(\alpha)\sin(\phi)\sin(2\chi)-2\cos(\alpha)\cos(\phi)\cos(\chi))\tanh(\beta\nu_{1}) .\label{QMA}
\end{split}
\end{equation}
By plugging the expression of $\theta$ (Eq. (\ref{thetta})) into Eq. (\ref{qma}), we see that the first term in Eq. (\ref{QMA}) is the heat absorbed by the dephased engine. Let's now show the origin of the second term in Eq. (\ref{QMA}). 

Consider the general case, i.e., $U$, $V$, and $\Phi$ are arbitrary. The expression of $\rho_{2}$ when written in the eigenbasis of $H_{2}$ is given by
\begin{equation}
\rho_{2}:=\mathbb{1}_{2}\rho_{2}\mathbb{1}_{2}=(|+\rangle_{22}\langle+|+|-\rangle_{22}\langle-|)\rho_{2}((|+\rangle_{22}\langle+|+|-\rangle_{22}\langle-|))=\Delta_{2}(\rho_{2})+Off_{2}(\rho_{2}).\label{decomposi}
\end{equation}
$\Delta_{2}(\rho_{2})={}_{2}\langle +|\rho_{2}|+\rangle_{2}|+\rangle_{22}\langle+|+{}_{2}\langle -|\rho_{2}|-\rangle_{2}|-\rangle_{22}\langle-|$ and $Off_{2}(\rho_{2})={}_{2}\langle +|\rho_{2}|-\rangle_{2}|+\rangle_{22}\langle-|+{}_{2}\langle -|\rho_{2}|+\rangle_{2}|-\rangle_{22}\langle+|$, are respectively, the diagonal and off-diagonal elements of $\rho_{2}$ in the eigenbasis of $H_{2}$. Now let's return to the average $\langle Q_{M}\rangle$. Mathematically, we have
\begin{equation}
\begin{split}
\langle Q_{M}\rangle  & =\mathrm{Tr}\left[(\rho_{3}-\rho_{2})H_{2}\right] 
\\ &
 =\mathrm{Tr}\left[(\Phi(\rho_{2})-\rho_{2})H_{2}\right]
\\ &=\mathrm{Tr}\left[(\Phi(\Delta_{2}(\rho_{2})+Off_{2}(\rho_{2}))-\rho_{2})H_{2}\right]
\\ &=\mathrm{Tr}\left[(\Phi (\Delta_{2}(\rho_{2})) -\rho_{2})H_{2}\right]+\mathrm{Tr}\left[\Phi (Off_{2}(\rho_{2}))H_{2}\right]
\\ &= \llangle Q_{M}\rrangle+\mathrm{Tr}\left[\Phi (Off_{2}(\rho_{2}))H_{2}\right].
\end{split}\label{QMA2}
\end{equation}
Where in the second line we use $\rho_{3}=\Phi(\rho_{2})$, in the third line we employ Eq. (\ref{decomposi}), in the fourth line we use the linearity property of the trace and of $\Phi$, and in the last line we use the first equation in Eq. (\ref{FFF}). Now it is clear that the second term in Eq. (\ref{QMA}) comes from $\mathrm{Tr}\left[\Phi (Off_{2}(\rho_{2}))H_{2}\right]$. Further, note that this latter term contributes to heat only when $\delta'\neq0$, i.e., in the non-adiabatic regime. This is because when $\delta'=0$, then in this case $\rho_{2}$ would be diagonal in the eigenbasis of $H_{2}$.

Similarly, to $\langle Q_{M}\rangle$ for $\langle Q_{C}\rangle$ we have,
\begin{equation}
\begin{split}
\langle Q_{C}\rangle  & =\mathrm{Tr}\left[(V\Phi(\rho_{2})V^{\dagger}-\rho_{2})H_{2}\right]
\\ &
=\mathrm{Tr}\left[(V(\Delta_{2}(\Phi(\rho_{2}))+Off_{2}(\Phi(\rho_{2})))V^{\dagger}-\rho_{2})H_{2}\right]
\\ &=\mathrm{Tr}\left[(V(\Delta_{2}(\Phi(\Delta_{2}(\rho_{2})))+Off_{2}(\Phi(\Delta_{2}(\rho_{2})))+\Delta_{2}(\Phi(Off_{2}(\rho_{2})))+Off_{2}(\Phi(Off_{2}(\rho_{2}))))V^{\dagger}-\rho_{2})H_{2}\right]
\\ &=\mathrm{Tr}\left[(V(\Delta_{2}(\Phi(\Delta_{2}(\rho_{2}))))V^{\dagger}-\rho_{2})H_{2}\right]+\mathrm{Tr}\left[ V(Off_{2}(\Phi(\Delta_{2}(\rho_{2})))+\Delta_{2}(\Phi(Off_{2}(\rho_{2})))+Off_{2}(\Phi(Off_{2}(\rho_{2}))))V^{\dagger}H_{2}\right]
\\ &= \llangle Q_{C}\rrangle+\mathrm{Tr}\left[ V(Off_{2}(\Phi(\Delta_{2}(\rho_{2})))+\Delta_{2}(\Phi(Off_{2}(\rho_{2})))+Off_{2}(\Phi(Off_{2}(\rho_{2}))))V^{\dagger}H_{2}\right].
\end{split}\label{QCA}
\end{equation}
In the second line, we decompose the state $\Phi(\rho_{2})$ in the eigenbasis of $H_{2}$ as we did for $\rho_{2}$. In the third line, we use the decomposition of the state $\rho_{2}$ in the eigenbasis of $H_{2}$, and in the last line, we use the second equation in Eq. (\ref{FFF}). From equations (\ref{QCA}) and (\ref{QMA}), the average work is given by,
\begin{equation}
\langle W\rangle=\llangle W\rrangle+\mathrm{Tr}\left[\Phi (Off_{2}(\rho_{2}))H_{2}\right]+\mathrm{Tr}\left[ V(Off_{2}(\Phi(\Delta_{2}(\rho_{2})))+\Delta_{2}(\Phi(Off_{2}(\rho_{2})))+Off_{2}(\Phi(Off_{2}(\rho_{2}))))V^{\dagger}H_{2}\right].\label{opm}
\end{equation}
Please note that the expressions of work and heats in Eqs. (\ref{QMA2})-(\ref{QCA})-(\ref{opm}) are valid for any finite $d$, i.e., not only for qubit systems, since the decomposition (\ref{decomposi}) can be generalized to arbitrary finite dimension $d$. These expressions show that when we are monitoring the system between the strokes, the coherence generated in the energy eigenbasis is removed.

Now let's compute Eqs. (\ref{QMA2})-(\ref{QCA})-(\ref{opm}) for qubit systems in terms of the parameters. First, in addition to $\delta'$, $\theta$, and $\zeta$, we define the next two transition probabilities, $\theta_{c}$ and $\zeta^{c}$, which are given by,
\begin{equation}
\theta_{c}:={}_{2}\langle -|\Phi(U|+\rangle_{11}\langle +|U^{\dagger})|-\rangle_{2},\label{teth}
\end{equation}
and,
\begin{equation}
\zeta^{c}:={}_{1}\langle -|V\Phi(U|+\rangle_{11}\langle +|U^{\dagger})V^{\dagger}|-\rangle_{1}.\label{zet}
\end{equation}
Their interpretation is similar to $\delta'$, $\theta$, and $\zeta$. Using them, one can compress the work and heat averages into simpler expressions, given by
\begin{equation}
\langle Q_{C}\rangle=-2\zeta^{c}\nu_{1}\tanh(\beta\nu_{1}),\label{qcun1}
\end{equation}
\begin{equation}
\langle Q_{M}\rangle=2(\theta_{c}-\delta')\nu_{2}\tanh(\beta\nu_{1}),\label{qmun1}
\end{equation}
and,
\begin{equation}
\langle W\rangle=2((\theta_{c}-\delta')\nu_{2}-\zeta^{c}\nu_{1})\tanh(\beta\nu_{1}).\label{wun1}
\end{equation}
See Appendix \ref{app1r} for details about the derivation. Please note that Eqs. (\ref{teth})-(\ref{zet}) are not independent of $\delta'$, $\theta$, and $\zeta$. We define them to compress the expressions of the averages into simpler formulas. Their role would become even more important when we consider the second cumulants.
\subsection{Variance of $W$, $Q_{M}$, and $Q_{C}$}
In Appendix \ref{app2r}, we give in detail the derivation of the variances. In addition to $\delta'$, $\theta$, $\zeta$, $\theta_{c}$, and $\zeta^{c}$, we define another transition probability given as follows:
\begin{equation}
\zeta_{c}:={}_{1}\langle -|V(\Phi(|+\rangle_{22}\langle +|))V^{\dagger}|-\rangle_{1}.\label{zetac}
\end{equation}
Again, note that $\zeta_{c}$ can be written in terms of $\theta$ and $\zeta$. Now the expressions of the variances are given as follows:
\begin{equation}
\langle Q_{C}^{2}\rangle_{c}=4\zeta^{c}\nu_{1}^{2}-\langle Q_{C}\rangle^{2},\label{resu11}
\end{equation}
\begin{equation}
\mathrm{Re}[\langle Q_{M}^{2}\rangle_{c}]=4\theta\nu_{2}^{2}-\langle Q_{M}\rangle^{2}=4\nu_{2}^{2}(\theta-((\theta_{c}-\delta')\tanh(\beta\nu_{1}))^{2}),\label{resu22}
\end{equation}
and,
\begin{equation}
\mathrm{Re}[\langle W^{2}\rangle_{c}]=4\nu_{1}\nu_{2}(\delta'+\zeta-\theta_{c}-\zeta_{c})+4(\nu_{1}^{2}\zeta^{c}+\nu_{2}^{2}\theta)-\langle W\rangle^{2}.\label{resu33}
\end{equation}

In the previous sections, we have mentioned that the second cumulants of work and heat $Q_{M}$ following from Eq. (\ref{pnmkl2}) have a real and imaginary part. One can show that because of the next averages: $\langle E_{2}E_{3}\rangle$, $\langle E_{2}E_{4}\rangle$, and $\langle E_{3}E_{4}\rangle$, the second cumulants of work $W$ and heat $Q_{M}$, have real and imaginary parts. On the other hand, the next three averages, $\langle E_{1}E_{2}\rangle$, $\langle E_{1}E_{3}\rangle$, and $\langle E_{1}E_{4}\rangle$, have only a real part, thus the second cumulant of $Q_{C}$ would be real. And because the second cumulant of work has an imaginary part, we defined the reliability of the undephased engine only with the real part. On the other hand, numerically, we found that the imaginary part can be negative. 

Concerning the fluctuations of the heat released to the cold bath denoted by $Q_{C}'$, see Appendix \ref{mpmpm}. Please, note that we do not need to redefine $Q_{C}'$, since the result $-$ that the $l$th cumulant of $Q_{C}'$ is the sum of the $l$th cumulants of $E_{1}$ and $E_{4}$ $-$ proven in Ref. \cite{Abdelkader2} is valid here. Thanks to the assumption that the bath maps every state of the system to the equilibrium state.
 
Finally, one can prove that $\langle Q_{C}^{2}\rangle_{c}/\langle Q_{C}\rangle^{2}\geq 2/\langle \Sigma\rangle-1$, see Appendix \ref{lowerrr}. In the same appendix, we prove that $\langle Q_{C}^{2}\rangle_{c}/\langle Q_{C}\rangle^{2}\geq1$ when $0\leq\zeta^{c} \leq1/2$. Now let's look at the influence of the parameters on our main quantities.

\subsection{Effect of the parameters on the undephased engine}\label{3azwa}

\begin{enumerate}

\item $\textit{Influence of}$ $\beta:$ The same as for the dephased engine. Thus, we see that the temperature of the bath does not give us a difference between the dephased and undephased engines. The important thing is that lowering the temperature of the bath has a good influence on average work and its reliability. That is, when we initialize the qubit in the ground state, we are getting rid of thermal fluctuations. In this case, if we consider the model in Sec. \ref{casestu}, the engine is purely driven by $\textit{quantum fluctuations}$, i.e., the unitaries and the projective measurement.
Of course, initializing the qubit in the ground state is not that simple assumption, since in this case, the third law comes into play, which would prevent us from doing this using a finite amount of resources, such as time and energy. Nevertheless, from a computational point of view, there is no problem assuming zero temperature.

\item $\textit{Influence of}$ $\nu_{2}:$ Concerning $\nu_{2}$, one can see from Eqs. (\ref{qcun1})-(\ref{qmun1})-(\ref{wun1}) that still only the heat absorbed and work are dependent on $\nu_{2}$. This dependence is linear. Thus, when increasing $\nu_{2}$, we increase work and efficiency. 

Now consider the model of Sec. \ref{casestu} as a case study. Numerically, one can show that increasing $\nu_{2}$ has a positive influence on work reliability. Further, note that the condition $\nu_{2}>\nu_{1}$ is no longer necessary for the system to work as a heat engine. See Refs. \cite{Abdelkader,Shanhe,JordanNM}, where a heat engine is possible also when $\nu_{2}=\nu_{1}$. Finally, note that the RFs of $Q_{M}$ and $Q_{C}$ are also independent of $\nu_{2}$.

\begin{figure}[hbtp]
\centering
\includegraphics[scale=0.45]{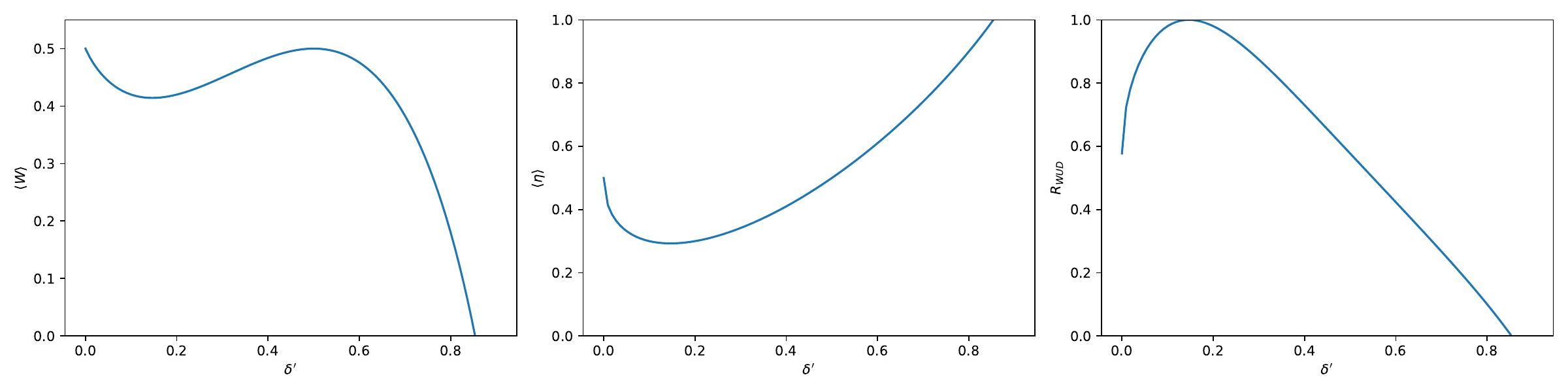}
\caption{Plot of the average work $\langle W\rangle$, efficiency $\langle \eta\rangle$, and work relibality $R_{WUD}$ as a function of $\delta'$ for $\nu_{1}=1$, $\nu_{2}=2$, $\chi=\phi=0$, $\beta=10$, and $\alpha=3\pi/4$. We see that they can increase or decrease as we increase $\delta'$. Note that in this plot we are considering the qubit model of Sec. \ref{casestu}. In this case, we have $\delta'=\zeta=\delta$.}\label{therm}
\end{figure}

\begin{figure}[hbtp]
\centering
\includegraphics[scale=0.354]{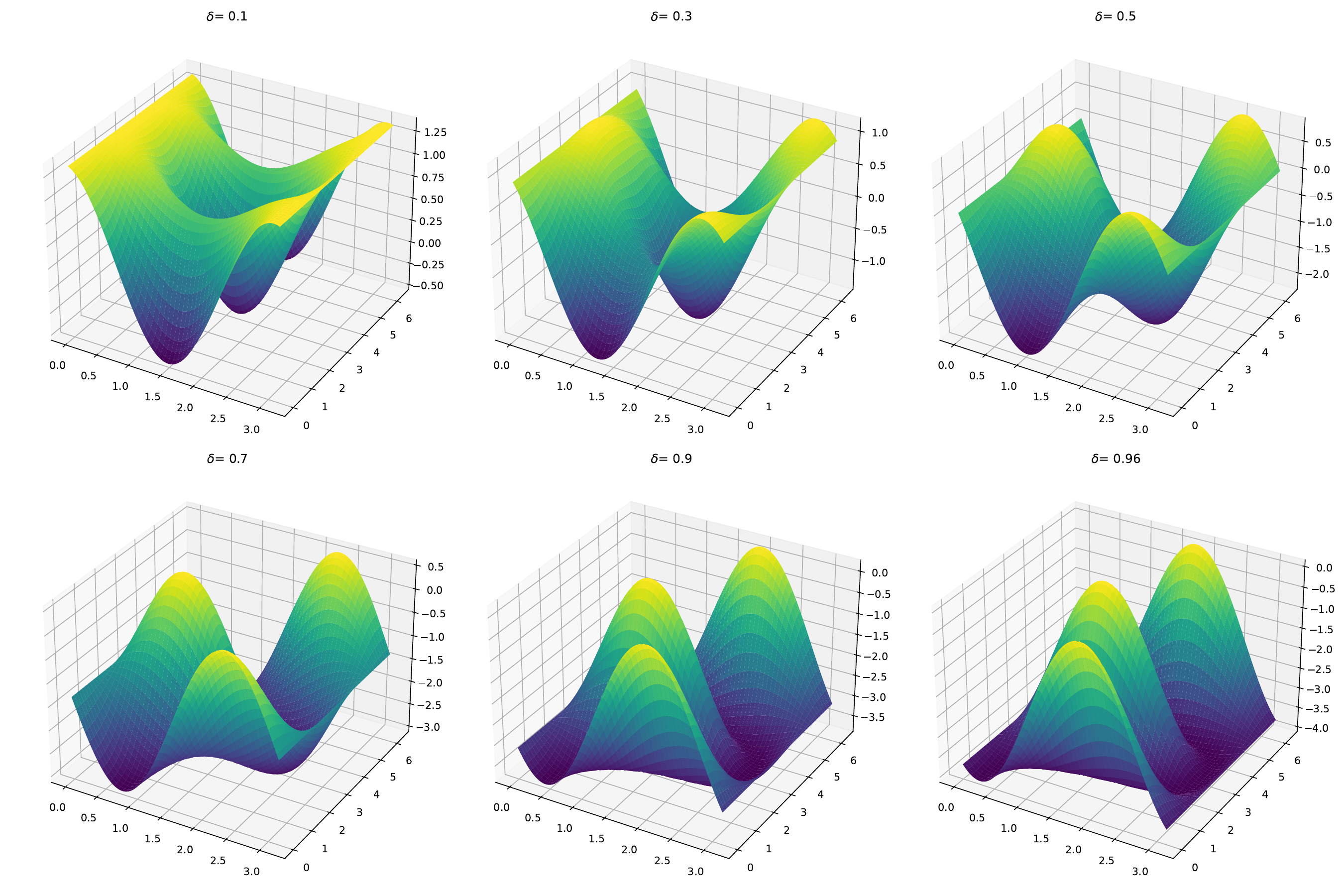}
\caption{Plot the average work $\langle W\rangle$ as a function of $\alpha\in[0,\pi]$ and $\chi\in[0,2\pi]$ for $\beta=+\infty$, $\nu_{1}=1$, $\nu_{2}=2$, $\phi=0$, $\delta=0.1,0.3,0.5,0.7,0.9$, and $0.96$.}\label{ggg}
\end{figure}

\item $\textit{Influence of}$ $\phi :$ Let's comment on the effect of this phase. From extensive numerical analysis, not necessarily to be presented here, we found that taking $\phi=0$ has a better influence on the higher values of work, its reliability, and efficiency. Of course, sometimes $\phi\neq0$ has a better influence on our main thermodynamic quantities. But our point is that the highest possible values of work, efficiency, and work reliability are only achieved when we set $\phi=0$. Please note that the highest value of efficiency and work reliability is 1. Already, we proved this for efficiency in Section \ref{casestu}. Below, we comment on the upper bound of work reliability. 

\item $\textit{Influence of}$ $\delta':$ In the case of the dephased engine, we found that heat, work, efficiency, and work reliability are monotonically decreasing functions as we increase $\delta'$ towards 1/2. For the undephased engine, this is no longer the case; see Fig. \ref{therm}. For example, from the average of the heat absorbed, Eq. (\ref{QMA}), we see that while the first term is monotonically decreasing as we increase $\delta'$ towards $1/2$, the second can be increasing depending on the three angles $\phi$, $\chi$, and $\alpha$. That is, while the highest value of the first term in Eq. (\ref{QMA}) is at $\delta'=0$, the second term in Eq. (\ref{QMA}) is attained when $\delta'=1/2$. Furthermore, from Fig. \ref{therm}, note that a heat engine is possible when $\delta'\geq1/2$ due to the term $\mathrm{Tr}\left[\Phi (Off_{2}(\rho_{2}))H_{2}\right]$ in Eq. (\ref{QMA}).

\item $\textit{Influence of}$ $\chi$ $\textit{and}$ $\alpha :$ In figure \ref{ggg}, we plot the average work as a function of $\alpha$ and $\chi$ for $\beta=+\infty$, $\nu_{1}=1$, $\nu_{2}=2$, $\phi=0$, $\delta=0.1,0.3,0.5,0.7,0.9$ and $0.96$. From this figure, we have the next two observations: 1) The highest possible value of work is achieved when we project the qubit in some basis in the \textit{xz}-plane. For example, for $\delta=0.1$, this value is approximately 0.97. 2) We see that the more we increase $\delta$ towards 1, the more it is better to project the qubit close to the \textit{x}-basis than to the \textit{z}-basis. By close, we mean that the angle between the best basis in the \textit{xz}-plane (i.e., the basis where the average work is maximal)  and the \textit{x}-basis is less than the angle between this best basis in the \textit{xz}-plane and the \textit{z}-basis.

\end{enumerate}

We have repeated the same plots as those of Fig. \ref{ggg} for average work for $\phi=\pi/2$, and we found that in this case, the best plane is the \textit{xy}-plane. However, note that a positive work condition was only verified for $\delta=0.1$ and not for $\delta$=0.3, 0.5, 0.7, 0.9, and 0.96. Further, for $\delta=0.1$, the highest possible value of work was found to be 0.6, which is smaller than the best achieved in the case when $\phi=0$, which is 0.97.

We have also plotted the efficiency and reliability of work as a function of $\alpha$ and $\chi$, and we found that when we set $\phi=0$, then efficiency and work reliability can achieve their best values also in the \textit{xz}-plane. But note that the maximum of work, its reliability, and efficiency are not necessarily achieved on the same basis in the \textit{xz}-plane. On the other hand, when setting $\phi=\pi/2$, it is better to measure the qubit in the \textit{xy}-plane.

In our work in Ref. \cite{Abdelkader}, we found numerically that a heat engine is possible even when $\delta'\geq1/2$, in contrast to the case when the Otto cycle is based on two completely thermalizing baths. Nevertheless, we did not give the reason behind this. In the next subsection, we explain this. 
\subsection{Why is work extraction possible for $\delta'\geq1/2$ for the undephased engine?}

When $\delta'\geq1/2$, the probability to find the qubit in the excited state of the Hamiltonian $H_{2}$ at $\mathbf{B}$ (in the cycle; see Fig. \ref{Otto}) is given by, $p_{e\mathbf{B}}:=(\langle E_{2}\rangle/\nu_{2}+1)/2=(1-(1-2\delta')\tanh(\beta\nu_{1}))/2$. Thus when $\delta'\geq1/2$, then we have $p_{e\mathbf{B}}\geq1/2$. This means that the excited level of $H_{2}$ at $\mathbf{C}$ must be more populated than at $\mathbf{B}$ in the cycle. This is the minimum condition to ensure that the system will absorb heat. This is not possible using a hot thermal bath with a positive inverse temperature since it can't populate higher levels with higher probabilities than lower levels. Thanks to quantum measurement, we could fuel the system even when $\delta'\geq1/2$. 

On the other hand, note that the dephased engine cannot work for $\delta'\geq1/2$. This is because the projective measurement between the strokes kills coherence. Thus, we see that coherence can be advantageous when present. However, we should mention that even though the dephased engine cannot work for $\delta'\geq1/2$, the dephased engine is different from the case of thermal baths. For example, while the highest probability of occupation of the excited state at $\mathbf{C}$ for the Otto cycle when it is based on two thermal baths is 1/2, the probability of occupation of the higher level of the dephased engine at $\mathbf{C}$ can exceed this. More precisely, at $\mathbf{C}$ in the cycle, the occupation probability of the excited state of the dephased engine is given by,
\begin{equation}
p_{e\mathbf{C}}:=(\llangle E_{3}\rrangle/\nu_{2}+1)/2=(1-(1-2\delta')(1-2\theta)\tanh(\beta\nu_{1}))/2.
\end{equation}
Thus, we see that when $\delta'$ exceeds 1/2, so does $p_{eC}$. However, we have $p_{e\mathbf{C}}\leq p_{e\mathbf{B}}$ for $\delta'\geq1/2$, thus a heat engine is not possible in this regime. Thus, to resume, we see that when the coherence is not erased, a heat engine is possible even in the usually \textit{not allowed regime in the literature}.
\subsection{Comparison between the work extracted in the \textit{x, y, and z} bases of the dephased and undephased engines}
For the qubit model of Sec. \ref{casestu}, we have the next results:

\textit{x}-basis: In this basis, both the dephased and undephased engines cannot extract work since the system cannot absorb heat. Thus, all the work consumed by the engine is transformed into useless heat dumped into the cold bath. In general, we have
\begin{equation}
\langle W\rangle=\llangle W\rrangle=-4\delta(1-\delta)\nu_{1}\tanh(\beta\nu_{1})\leq0.
\end{equation}

\textit{y}-basis: In this basis we have,
\begin{equation}
\langle W\rangle-\llangle W\rrangle=-4\delta(1-\delta)\sin^{2}(\phi)\nu_{1}\tanh(\beta\nu_{1})\leq0.\label{y}
\end{equation}
We see that the dephased engine can enhance the undephased when $\phi\neq0$ or $\pi$, i.e., $\langle W\rangle\leq\llangle W\rrangle$. On the other hand, when $\phi=0$ or $\pi$, we have $\langle W\rangle=\llangle W\rrangle$. When $\nu_{2}=\nu_{1}$ we have,
\begin{equation}
\langle W\rangle=-2 \delta \nu_{1} (1+2(1-\delta) \sin^{2}(\phi))\tanh(\beta\nu_{1})\leq0.
\end{equation}
This shows that when measuring the qubit in the \textit{y}-basis, work cannot be extracted when $\nu_{2}=\nu_{1}$.

\textit{z}-basis: In this basis we have,
\begin{equation}
\langle W\rangle-\llangle W\rrangle=4\delta(1-\delta)\cos^{2}(\phi)\nu_{1}\tanh(\beta\nu_{1})\geq0,
\end{equation}
Thus, $\langle W\rangle\geq\llangle W\rrangle$. Note that they become equal for $\phi=\pi/2$ and $\phi=3\pi/2$. However, differently from the \textit{y}-basis, a heat engine can be possible in the \textit{z}-basis when $\nu_{2}=\nu_{1}$.

\subsection{Not all the bases in the \textit{yz}-plane are equivalent for the undephased engine}
The \textit{yz}-plane corresponds to $\chi=\pi/2$ with $\alpha$ being arbitrary. For the dephased engine, we have seen that all the bases in the \textit{yz}-plane are equivalent, since when $\chi=\pi/2$, $\theta$ is equal to 1/2, independently of the value of $\alpha$. However, this is not the case for the undephased engine. More precisely, we have
\begin{equation}
\langle W\rangle_{\textit{yz-}plane}-\langle W\rangle_{\textit{y-}basis}=4\delta(1-\delta)\nu_{1}\cos^{2}(\alpha)\tanh(\beta\nu_{1}) \geq0,
\end{equation}
and,
\begin{equation}
\langle W\rangle_{\textit{z-}basis}-\langle W\rangle_{\textit{yz-}plane}=4\delta(1-\delta)\nu_{1}\sin^{2}(\alpha)\tanh(\beta\nu_{1}) \geq0.
\end{equation}
Thus, the maximal amount of extracted work in the \textit{yz-plane} is when the qubit is measured in the \textit{z}-basis. Therefore, we have
\begin{equation}
\langle W\rangle_{\textit{z-}basis}\geq\langle W\rangle_{\textit{yz-}plane}\geq\langle W\rangle_{\textit{y-}basis}.\label{xyz}
\end{equation}
From this equation, we see that measuring the qubit close to the \textit{z}-basis is better than close to the \textit{y}-basis. By close, we mean that the angle between the basis on which the qubit is projected and the \textit{z}-basis is small compared to the angle between the considered basis and the \textit{y}-basis. For the heat absorbed, we have
\begin{equation}
\langle Q_{M}\rangle_{\textit{z-}basis}=\langle Q_{M}\rangle_{\textit{yz-}plane}=\langle Q_{M}\rangle_{\textit{y-}basis}=(1-2\delta)\nu_{2}\tanh(\beta\nu_{1}).
\end{equation}
Since the heat absorbed in all bases in the \textit{yz}-plane is the same, and from Eq. (\ref{xyz}) we have,
\begin{equation}
\langle \eta\rangle_{\textit{z-}basis}\geq\langle \eta\rangle_{\textit{yz-}plane}\geq\langle \eta\rangle_{\textit{y-}basis}.\label{sbah}
\end{equation}
From this equation and Eq. (\ref{xyz}), we see that the best basis in the \textit{yz}-plane is the \textit{z}-basis.
\subsection{High values of work and efficiency}
Let's now look at the previous results carefully. From Eqs. (\ref{qmun1})-(\ref{qcun1})-(\ref{wun1}), the efficiency expression is given by,
\begin{equation}
\langle \eta\rangle=\frac{\langle W\rangle}{\langle Q_{M}\rangle}=1+\frac{\langle Q_{C}\rangle}{\langle Q_{M}\rangle}=1-\frac{\nu_{1}}{\nu_{2}}\frac{\zeta^{c}}{\theta_{c}-\delta'}.
\end{equation}
From the expression of $\langle Q_{C}\rangle$ we see the heat released to the cold bath is minimal when $\zeta^{c}\rightarrow0$. While from the expression of $\langle Q_{M}\rangle$ we see the heat absorbed is maximal when $\theta_{c}-\delta'$ is maximal.

From Eq. (\ref{wun1}), we see that for a given $\delta'\nu_{2}$, the more we increase $\theta_{c}\nu_{2}-\zeta^{c}\nu_{1}$, the more work can be extracted. Numerically, for the model of Sec. \ref{casestu}, we found that the best basis that verifies this condition is a part of the \textit{xz}-plane. Further, from $\nu_{2}\geq\nu_{1}$, one can see that work is lower-bounded as follows:
\begin{equation}
\langle W\rangle\geq2((\theta_{c}-\zeta^{c})-\delta')\nu_{2}\tanh(\beta\nu_{1})\geq2((\theta_{c}-\zeta^{c})-\delta')\nu_{1}\tanh(\beta\nu_{1}).\label{iwn}
\end{equation}
For a given $\delta'$, we see that for the lower bound $((\theta_{c}-\zeta^{c})-\delta')\nu_{1}\tanh(\beta\nu_{1})$ to be as higher as possible, the difference $\theta_{c}-\zeta^{c}$ should be increased. When considering the qubit model of Sec. \ref{casestu}, please note that the basis at which work is maximum when $\nu_{2}=\nu_{1}$ is not the same when $\nu_{2}>\nu_{1}$.

Now let's look at the maximum of the efficiency. We already showed that it is upper-bounded by 1. From its expression, we see that for $\langle \eta\rangle$ to be higher, one has the following \textit{two} possibilities:
\begin{enumerate}
\item For $\frac{\zeta^{c}}{\theta_{c}-\delta'}>0$ but still $\frac{\zeta^{c}}{\theta_{c}-\delta'}$ finite, we have to make $\nu_{2}\gg\nu_{1}$.

\item For $\nu_{1}/\nu_{2}>0$, we have to make $\frac{\zeta^{c}}{\theta_{c}-\delta'}\rightarrow 0$.
\end{enumerate}
The first possibility may be challenging experimentally since we need $\nu_{1}/\nu_{2}\rightarrow 0$, \textit{i.e.}, we should increase $\nu_{2}$ to very high values with respect to $\nu_{1}$. Thus, reaching 1 efficiency using this possibility may not be feasible experimentally.

For the second possibility, we have two sub-possibilities: a) For $\theta_{c}-\delta'>0$ (\textit{already satisfied in the heat engine region}), we need to make $\zeta^{c}\rightarrow0$; b) Both $\theta_{c}-\delta'$ and $\zeta^{c}$ go to zero as efficiency $\rightarrow1$. The first sub-possibility means that we can reach 1 efficiency with finite work. In this case, all the heat gets converted into work, and thus no heat is released. For the second sub-possibility, it means that as efficiency converges to 1, work also converges to 0. \textit{Numerically} we found that only the second subpossibility is possible, i.e., the greatest possible efficiency is achieved only in the case when all the heat and work averges converge to 0. This is reminiscent of the case when, e.g., a heat engine reaches the Carnot bound and all the currents converge to zero. In our case, it is 1 that plays the role of \textit{Carnot efficiency}. Furthermore, one can actually show that because $\nu_{2}\geq\nu_{1}$ efficiency is lower-bounded, as follows:
\begin{equation}
\langle \eta\rangle\geq 1-\frac{\zeta^{c}}{\theta_{c}-\delta'}.\label{ietan}
\end{equation}
When considering the qubit model of Sec. \ref{casestu}, contrary to the average work, the basis that maximizes the lower bound of efficiency to reach 1 is the same independently of $\nu_{2}=\nu_{1}$ or $\nu_{2}>\nu_{1}$. This is because when efficiency converges to 1, the ratio $\zeta^{c}/(\theta_{c}-\delta')$ goes to zero independently of whether $\nu_{2}=\nu_{1}$ or $\nu_{2}>\nu_{1}$.

\begin{remark}
Of course, one can achieve unit efficiency with a non-zero value of work. In this case, we should take $\nu_{2}\gg\nu_{1}$, but as we pointed out before, this may be challenging experimentally. 
\end{remark}

\subsection{Do the Eqs. (\ref{mqcqhw1})-(\ref{UppLo}) derived for the dephased engine hold for the undephased engine?}
Let's ask the question about the validity of Eqs. (\ref{mqcqhw1})-(\ref{UppLo}) for the undephased engine.  Let's limit ourselves to the case when $V=U^{\dagger}$. Consider the unital channel in Sec. \ref{casestu}. In this case, one can prove that $\zeta_{c}=\theta_{c}$ and $\delta'=\zeta$; see Appendix \ref{rwqm}. Using these two facts, one can arrive at the next theorem (see Appendix \ref{EERF}).
\begin{theorem}
Consider the unital channel in Sec. \ref{casestu}. When $\zeta_{c}=\theta_{c}$ and $\delta'=\zeta$, one can prove that the ratio of work $W$ and heat $Q_{M}$ fluctuations in the heat engine region satisfy,
\begin{equation}
\frac{\mathrm{Re}[\langle W^{2}\rangle_{c}]}{\mathrm{Re}[\langle Q_{M}^{2}\rangle_{c}]}\leq1.\label{UppLoE}
\end{equation}
\end{theorem}
This theorem shows that even when coherence is present, the fluctuations of heat $Q_{M}$ give an upper bound to the fluctuations of work. However, while in the case of the dephased engine, we have $\llangle W^{2}\rrangle_{c}/\llangle Q_{M}^{2}\rrangle_{c}<1$, in the case of the undephased engine we have $\mathrm{Re}[\langle W^{2}\rangle_{c}]/\mathrm{Re}[\langle Q_{M}^{2}\rangle_{c}]\leq1$ even in the heat engine region. The fluctuations become equal when efficiency goes to 1. To show this, from Appendix \ref{EERF} we have
\begin{equation}
\begin{split}
\mathrm{Re}[\langle Q_{M}^{2}\rangle_{c}]-\mathrm{Re}[\langle W^{2}\rangle_{c}] =2\nu_{1}\coth(\beta\nu_{1})(\langle W\rangle+\langle Q_{M}\rangle)(1-\zeta^{c}\tanh^{2}(\beta\nu_{1})).
\end{split}\label{gth}
\end{equation}
Thus, we see that when work and heat $Q_{M}$ go to zero, their fluctuations become equal. Please note that similarly to the depahsed engine, Eq. (\ref{gth}) can still be $\geq0$ even in the accelerator region. On the other hand, in the heater region we have $\mathrm{Re}[\langle Q_{M}^{2}\rangle_{c}]-\mathrm{Re}[\langle W^{2}\rangle_{c}]\leq0$, since both $\langle W\rangle$ and $\langle Q_{M}\rangle$ are $\leq0$.

Now let's go back to the difference between the RFs. When $V=U^{\dagger}$, in Appendix \ref{diffref} we compute the difference between the RFs of $W$ and $Q_{M}$ and between the RFs of $W$ and $Q_{C}$. Their expressions are
\begin{equation}
\begin{split}
\frac{\mathrm{Re}[\langle W^{2}\rangle_{c}]}{\langle W\rangle^{2}}-\frac{\mathrm{Re}[\langle Q_{M}^{2}\rangle_{c}]}{\langle Q_{M}\rangle^{2}}=\frac{8\nu_{1}\nu_{2}^{2}\tanh(\beta\nu_{1})(\langle Q_{M}\rangle+\langle W\rangle)(\theta\zeta^{c}-(\theta_{c}-\delta')^{2})}{\langle W\rangle^{2}\langle Q_{M}\rangle^{2}},
\end{split}\label{osam}
\end{equation}
and,
\begin{equation}
\frac{\mathrm{Re}[\langle W^{2}\rangle_{c}]}{\langle W\rangle^{2}}-\frac{\langle Q_{C}^{2}\rangle_{c}}{\langle Q_{C}\rangle^{2}}=\frac{(4\nu_{1}\nu_{2}\tanh(\beta\nu_{1}))^{2}(\theta\zeta^{c}-(\theta_{c}-\delta')^{2})}{\langle W\rangle^{2}\langle Q_{C}\rangle^{2}}.\label{osmr}
\end{equation}
Note that both differences are proportional to $(\theta\zeta^{c}-(\theta_{c}-\delta')^{2})$. Numerically, for the model of Sec. \ref{casestu}, we always find that it is $\geq0$. However, we could not prove it. This is because $\theta$, $\zeta^{c}$, $\theta_{c}$, and $\delta'$ are all linked to each other. Further, note that when $(\langle Q_{M}\rangle+\langle W\rangle)\geq0$, the differences in Eqs. (\ref{osam})-(\ref{osmr}) have the same sign. Furthermore, even though we could not prove that,
\begin{equation}
\frac{\mathrm{Re}[\langle W^{2}\rangle_{c}]}{\langle W\rangle^{2}}\geq\frac{\mathrm{Re}[\langle Q_{M}^{2}\rangle_{c}]}{\langle Q_{M}\rangle^{2}}\geq \frac{\langle Q_{C}^{2}\rangle_{c}}{\langle Q_{C}\rangle^{2}},
\end{equation}
we believe it holds for the undephased engine. However, while the inequalities in Eq. (\ref{mqcqhw1}) becomes equalities only in the adiabatic regime \cite{Abdelkader,Abdelkader2}. For the undephased engine, the RFs can be equal even when $\delta'\neq0$. To show this, let's consider the model of Sec. \ref{casestu}. We already pointed out that the highest values of work, efficiency, and work reliability are achieved when $\phi=\chi=0$. Under the later conditions we have,
\begin{equation}
\theta\zeta^{c}=(\theta_{c}-\delta')^{2}(\neq0).\label{pmol}
\end{equation}
Plugging this into Eqs. (\ref{osam})-(\ref{osmr}), we have,
\begin{equation}
\frac{\mathrm{Re}[\langle W^{2}\rangle_{c}]}{\langle W\rangle^{2}}=\frac{\mathrm{Re}[\langle Q_{M}^{2}\rangle_{c}]}{\langle Q_{M}\rangle^{2}}= \frac{\langle Q_{C}^{2}\rangle_{c}}{\langle Q_{C}\rangle^{2}}.
\end{equation}
From the first equality in this equation, we see that
\begin{equation}
\langle \eta\rangle^{2}=\left(\frac{\langle W\rangle}{\langle Q_{M}\rangle}\right)^{2}=\frac{\mathrm{Re}[\langle W^{2}\rangle_{c} ]}{\mathrm{Re}[\langle Q_{M}^{2}\rangle_{c}]}.
\end{equation}
Again, while the lower bound on the ratio of work and heat $Q_{M}$ fluctuations in Eq. (\ref{UppLo}) is only achieved in the adibatic regime, in the presence of coherence, this is not the case.

Finally we have,
\begin{equation}
\mathrm{Re}[\langle Q_{M}^{2}\rangle_{c}]-\langle Q_{M}\rangle^{2}=4\nu_{2}^{2}(\theta-2((\theta_{c}-\delta')\tanh(\beta\nu_{1}))^{2}).
\end{equation}
When $\phi=\chi=0$ this equation becomes,
\begin{equation}
\mathrm{Re}[ \langle Q_{M}^{2}\rangle_{c}]-\langle Q_{M}\rangle^{2}=4\theta\nu_{2}^{2}(1-2\zeta^{c}\tanh^{2}(\beta\nu_{1})).\label{uji}
\end{equation}
Where we use Eq. (\ref{pmol}). Further note that in Sec. \ref{casestu}, $V=U^{\dagger}$, when $\phi=0$. In Appendix \ref{lowerrr}, we proved that for arbitray $U$ and for $V=U^{\dagger}$ we have $0\leq\zeta^{c}\leq1/2$. In this case, we have $\mathrm{Re}[\langle Q_{M}^{2}\rangle_{c}]\geq\langle Q_{M}\rangle^{2}$. From this, we have
\begin{equation}
R_{WUD}\leq1.\label{rwm}
\end{equation}
This shows that work reliability is still bounded by 1. Furthermore, note that while the reliability of the dephased engine needs both $\delta'=0$ and $\beta=+\infty$ to reach 1, the reliability of the undephased engine needs only $\beta=+\infty$. That is, as we see from Eq. (\ref{uji}), when $\beta=+\infty$ and $\zeta^{c}=1/2$, even when $\delta'\neq0$, we have $R_{WUD}=1$. In the future, we try to prove Eq. (\ref{rwm}) for arbitrary $\phi$ and $\chi$. In Fig. \ref{therm2}, we give the plot of reliability and $2/\langle\Sigma\rangle-1$.

\begin{figure}[hbtp]
\centering
\includegraphics[scale=0.35]{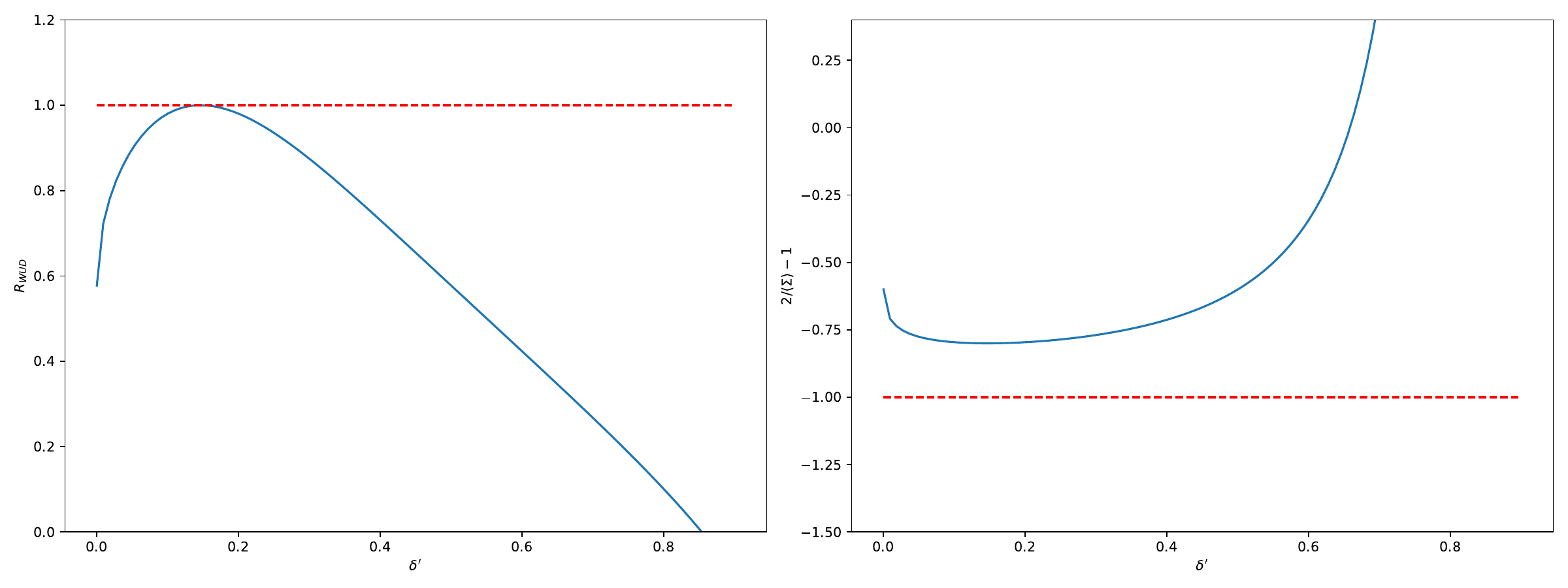}
\caption{Plot of the relibality $R_{WUD}$ and the bound $2/\langle\Sigma\rangle-1$ as a function of $\delta'$ for $\nu_{1}=1$, $\nu_{2}=2$, $\chi=\phi=0$, $\beta=10$, and $\alpha=3\pi/4$. We see from the plot of $R_{WUD}$ that the work reliability is $\leq1$, thus respecting our bound, in theorem (\ref{theorem1E}) and Eq. (\ref{rwm}). On the other hand, we see that $2/\langle\Sigma\rangle-1$ becomes negative even in the heat engine (see Fig. (\ref{therm})), showing that we can't use it to bound the reliability of work.}\label{therm2}
\end{figure}

\subsection{Main featues of the undephased engine }
Here, let's resume all the features of the undephased engine. We have,
\begin{enumerate}
\item $\langle Q_{M}\rangle$ can still be positive even when $\delta'\geq1/2$. Thus, in the presence of coherence, a heat engine or an accelerator can still be possible even when $\delta'\geq1/2$, which is not the case for the dephased engine. A heater also becomes possible when $\delta'\leq1/2$. Thanks to the second term of the last equation in Eqs. (\ref{QMA2}).

\item The average work does not need the condition $\nu_{2}>\nu_{1}$ for a positive work condition. Even when $\nu_{2}=\nu_{1}$, we still have a heat engine \cite{Abdelkader,Shanhe,JordanNM}. Further, we showed that the work extracted by the dephased heat engine is bounded by $\nu_{2}-\nu_{1}$. The latter becomes useless when $\nu_{2}=\nu_{1}$ since it predicts that a heat engine is not possible.

\item The efficiency is bounded by 1, not by that of the Otto.

\item Consider the qubit model of Sec. \ref{casestu}. When $\delta>0$, the best plane for a high amount of work, high reliability of work, and high efficiency is the \textit{xz}-plane. However, note that when $\delta=0$, both the dephased and undephased engines become identical in terms of the averages and, in general, the cumulants. 

\item The monotonic behavior of the cumulants of the dephased engine is no longer valid for the undephased engine. That is we can see work, efficiency and work reliability increases when we increase $\delta$.

\item Equality between the RFs does not need the adiabatic regime.
\end{enumerate}

\section{Conclusions}\label{conc}

In this work, we have extended our previous one, Refs. \cite{Abdelkader,Abdelkader2}, more thoroughly by considering also quantum coherence. We have shown how one can derive the cumulants of the undephased engine. We found that for coherence to be included, we should use \textit{Kirkwood-Dirac quasiprobability}. Then we explained in detail the influence of the parameters on average work, efficiency, and work reliability on the monitored engine. For this latter engine, we found that the highest values of the main quantities are achieved only in the adiabatic regime. For the undephased engine, we first showed how using Eqs. (\ref{deltap})-(\ref{tethaaaa})-(\ref{zetaa})-(\ref{teth})-(\ref{zet})-(\ref{zetac}), one can obtain all the averages and variances and compress them into simpler expressions; see Eqs. (\ref{qcun1})-(\ref{qmun1})-(\ref{wun1}) and (\ref{resu11})-(\ref{resu22})-(\ref{resu33}) for arbitrary qubit unitaries and unital channels. Then, considering the qubit model of Sec. \ref{casestu}, we have shown in which plane we should fuel the engine for the best average work, efficiency, and work reliability. Our study explains in detail which parameters should be increased and which should not for an enhancement of work, efficiency, and work reliability.

In addition to our analytical results, we showed that non-adiabatic transitions are not always detrimental to thermodynamic quantities; see Refs. \cite{Alecce,Feldmann,KosloffR}, where the negative role of non-adiabatic transitions was pointed out and explained. Our work shows that we can take advantage of them to increase average work, efficiency, and work reliability. This advantage would not be possible when the hot bath is completely thermalizing or when the working medium is monitored. Thanks to quantum measurement \textit{that fuels the engine} and \textit{having a positive influence on the engine's coherence created in the first unitary stroke $\mathbf{A}\rightarrow\mathbf{B}$} (cf. Fig. \ref{Otto}), non-adiabatic transitions become useful. Furthermore, we found that a heat engine becomes possible in the usually not-allowed regime, i.e., $\langle W\rangle>0$ even when $\delta\geq1/2$; see Fig. \ref{therm} and Fig. \ref{ggg}. Numerical plots also showed that an accelerator (a heater) becomes possible when $\delta\geq(\leq)1/2$.

We proved that the ratio of the fluctuations of work $W$ and heat $Q_{M}$ is still bounded by 1 in the heat engine region, even for the undephased engine. Further, we explained in detail the relationships between the RFs, i.e., Eqs. (\ref{mqcqhw1})-(\ref{UppLo}), for the undephased engine. We hope this work sheds more light on \textit{quantum unital Otto heat engines} \cite{Abdelkader,Abdelkader2}. Our study has the advantage that the majority of the results are proven analytically. And we believe that they can be pushed further. For example, one can look at these results for higher-dimensional working mediums such as coupled spins \cite{Abdelkader3}. One can also look at the implementation of this engine experimentally and test the validity of these results and the previous ones in Refs. \cite{Abdelkader,Abdelkader2}, when the system is subjected to external noises, i.e., when, e.g., the unitaries and the unital channel are not exact. Finally, one can relax the assumption that the cold bath is completely thermalizing.

We must emphasize that people use full-counting statistics to assess the fluctuations of quantum Otto heat engines \cite{THTH,THTH2}. Now, the question is: what is the difference between our work and the previous ones? Actually, after some calculations, we found that the second cumulants that follow from Eq.(\ref{chiFF}) contain more terms, which we believe would enhance our understanding of how to engineer reliable quantum thermal machines. The fact that our approach gives more terms is because Eq.(\ref{chiFF}) also includes the correlations between the stochastic quantities $E_1$, $E_2$, $E_3$, and $E_4$. To be more precise, the Kirkwood-Dirac quasi-probability takes into account the correlations between the random quantities that describe work during the adiabatic strokes. It would be important to compare our approach of using the Kirkwood-Dirac quasi-probability with full-counting statistics.


\begin{thebibliography}{88}%
\makeatletter
\providecommand \@ifxundefined [1]{%
 \@ifx{#1\undefined}
}%
\providecommand \@ifnum [1]{%
 \ifnum #1\expandafter \@firstoftwo
 \else \expandafter \@secondoftwo
 \fi
}%
\providecommand \@ifx [1]{%
 \ifx #1\expandafter \@firstoftwo
 \else \expandafter \@secondoftwo
 \fi
}%
\providecommand \natexlab [1]{#1}%
\providecommand \enquote  [1]{``#1''}%
\providecommand \bibnamefont  [1]{#1}%
\providecommand \bibfnamefont [1]{#1}%
\providecommand \citenamefont [1]{#1}%
\providecommand \href@noop [0]{\@secondoftwo}%
\providecommand \href [0]{\begingroup \@sanitize@url \@href}%
\providecommand \@href[1]{\@@startlink{#1}\@@href}%
\providecommand \@@href[1]{\endgroup#1\@@endlink}%
\providecommand \@sanitize@url [0]{\catcode `\\12\catcode `\$12\catcode
  `\&12\catcode `\#12\catcode `\^12\catcode `\_12\catcode `\%12\relax}%
\providecommand \@@startlink[1]{}%
\providecommand \@@endlink[0]{}%
\providecommand \url  [0]{\begingroup\@sanitize@url \@url }%
\providecommand \@url [1]{\endgroup\@href {#1}{\urlprefix }}%
\providecommand \urlprefix  [0]{URL }%
\providecommand \Eprint [0]{\href }%
\providecommand \doibase [0]{http://dx.doi.org/}%
\providecommand \selectlanguage [0]{\@gobble}%
\providecommand \bibinfo  [0]{\@secondoftwo}%
\providecommand \bibfield  [0]{\@secondoftwo}%
\providecommand \translation [1]{[#1]}%
\providecommand \BibitemOpen [0]{}%
\providecommand \bibitemStop [0]{}%
\providecommand \bibitemNoStop [0]{.\EOS\space}%
\providecommand \EOS [0]{\spacefactor3000\relax}%
\providecommand \BibitemShut  [1]{\csname bibitem#1\endcsname}%
\let\auto@bib@innerbib\@empty







\bibitem{Scully} H. E. D. Scovil and E. O. Schulz-Dubois, Three-level masers as heat engines, \href {\doibase
https://doi.org/10.1103/physrevlett.2.262} {\bibfield {journal} {\bibinfo{journal}
		{Phys. Rev. Lett.}\ }\textbf {\bibinfo {volume} {2}},\ \bibinfo {pages}
	{262} (\bibinfo {year} {1959})}.
	


\bibitem{Anders} S. Vinjanampathy and J. Anders, Quantum thermodynamics, \href {\doibase
https://doi.org/10.1080/00107514.2016.1201896} {\bibfield {journal} {\bibinfo{journal}
		{Contemp. Phys.}\ }\textbf {\bibinfo {volume} {57}},\ \bibinfo {pages}
	{545} (\bibinfo {year} {2016})}.

\bibitem{Huber} J. Goold, M. Huber, A. Riera, L. del Rio and P. Skrzypczyk, The role of quantum information in thermodynamics—a topical review, \href {\doibase
https://doi.org/10.1088/1751-8113/49/14/143001} {\bibfield {journal} {\bibinfo{journal}
		{J. Phys. A: Math. Theor.}\ }\textbf {\bibinfo {volume} {49}},\ \bibinfo {pages}
	{143001} (\bibinfo {year} {2016})}.

\bibitem{Binder} F. Binder, L. A. Correa, C. Gogolin, J. Anders and G. Adesso, Thermodynamics in the quantum regime, Fundamental Theories of Physics {\bf195}, 1-2 (2018) .

\bibitem{Deffner} S. Deffner and S. Campbell, Quantum Thermodynamics: An introduction to the thermodynamics of quantum information, (Morgan \& Claypool Publishers, 2019) .

\bibitem{Myers1} N. M. Myers, O. Abah and S. Deffner, Quantum thermodynamic devices: from theoretical proposals to experimental reality,  \href {\doibase
https://doi.org/10.1116/5.0083192} {\bibfield {journal} {\bibinfo{journal}
		{AVS quantum science}\ }\textbf {\bibinfo {volume} {4}},\ \bibinfo {pages}
	{027101} (\bibinfo {year} {2022})}.

\bibitem{Nicole} N. Yunger Halpern, Toward physical realizations of thermodynamic resource theories, in Information and Interaction: Eddington, Wheeler, and the Limits of Knowledge, edited by I. T. Durham and D. Rickles (Springer International Publishing, Cham, 2017) pp. 135–166.



\bibitem{Matteo} M. Lostaglio, An introductory review of the resource theory approach to thermodynamics, \href {\doibase
https://doi.org/10.1088/1361-6633/ab46e5} {\bibfield {journal} {\bibinfo{journal}
		{Rep. Prog. Phys}\ }\textbf {\bibinfo {volume} {82}},\ \bibinfo {pages}
	{114001} (\bibinfo {year} {2019})}.

\bibitem{Maxwell} H. S. Leff and A. F. Rex, Maxwell’s Demon: Entropy, Information, Computation, Computing (Princeton University Press, Princeton, 1990).



\bibitem{Kim} J. Yi, P. Talkner and Y. W. Kim, Single-temperature quantum engine without feedback control, \href {\doibase
https://doi.org/10.1103/PhysRevE.96.022108} {\bibfield {journal} {\bibinfo{journal}
		{Phys. Rev. E}\ }\textbf {\bibinfo {volume} {96}},\ \bibinfo {pages}
	{022108} (\bibinfo {year} {2017})}.

\bibitem{Das} A. Das and S. Ghosh, Measurement based quantum heat engine with coupled working medium, \href {\doibase
https://doi.org/10.3390/e21111131} {\bibfield {journal} {\bibinfo{journal}
		{Entropy}\ }\textbf {\bibinfo {volume} {21}},\ \bibinfo {pages}
	{1131} (\bibinfo {year} {2019})}.

\bibitem{Elouard} C. Elouard, D. Herrera-Martí, B. Huard and A. Auffèves, Extracting Work from Quantum Measurement in Maxwell’s Demon Engines, \href {\doibase
https://doi.org/10.1103/PhysRevLett.118.260603} {\bibfield {journal} {\bibinfo{journal}
		{Phys. Rev. Lett}\ }\textbf {\bibinfo {volume} {118}},\ \bibinfo {pages}
	{260603} (\bibinfo {year} {2017})}.

\bibitem{Buffoni} L. Buffoni, A. Solfanelli, P. Verrucchi, A. Cuccoli and M. Campisi, Quantum measurement cooling, \href {\doibase
https://doi.org/10.1103/PhysRevLett.122.070603} {\bibfield {journal} {\bibinfo{journal}
		{Phys. Rev. Lett}\ }\textbf {\bibinfo {volume} {22}},\ \bibinfo {pages}
	{070603} (\bibinfo {year} {2019})}.

\bibitem{Anka} M. F. Anka, T. R. de Oliveira and D. Jonathan, Measurement-based quantum heat engine in a multilevel system, \href {\doibase
https://doi.org/10.1103/PhysRevE.104.054128} {\bibfield {journal} {\bibinfo{journal}
		{Phys. Rev. E}\ }\textbf {\bibinfo {volume} {105}},\ \bibinfo {pages}
	{054128} (\bibinfo {year} {2021})}.


\bibitem{Prasanna1} M. Sahnawaz Alam and B. Prasanna Venkatesh, Two-stroke Quantum Measurement Heat Engine, (2022), \href {\doibase
https://doi.org/10.48550/arXiv.2201.06303} {\bibfield {journal} {\bibinfo{journal}
		{arXiv:2201.06303 [quant-ph]}\ }}.


\bibitem{KimT} X. Ding, J. Yi, Y. W. Kim and P. Talkner, Measurement-driven single temperature engine, \href {\doibase
https://doi.org/10.1103/PhysRevE.98.042122} {\bibfield {journal} {\bibinfo{journal}
		{Phys. Rev. E}\ }\textbf {\bibinfo {volume} {98}},\ \bibinfo {pages}
	{042122} (\bibinfo {year} {2018})}.
  

\bibitem{Sagawa} J. J. Park, K.-H. Kim, T. Sagawa and S. W. Kim, Heat engine driven by purely quantum information, \href {\doibase
https://doi.org/10.1103/PhysRevLett.111.230402} {\bibfield {journal} {\bibinfo{journal}
		{Phys. Rev. Lett.}\ }\textbf {\bibinfo {volume} {111}},\ \bibinfo {pages}
	{230402} (\bibinfo {year} {2013})}.

\bibitem{Jordan} C. Elouard and A. N. Jordan, Efficient quantum measurement engines, \href {\doibase
https://doi.org/10.1103/PhysRevLett.120.260601} {\bibfield {journal} {\bibinfo{journal}
		{Phys. Rev. Lett.}\ }\textbf {\bibinfo {volume} {120}},\ \bibinfo {pages}
	{260601} (\bibinfo {year} {2018})}.

\bibitem{Chand} S. Chand and A. Biswas,  Measurement-induced operation of two-ion quantum heat machines, \href {\doibase
https://doi.org/10.1103/PhysRevE.95.032111} {\bibfield {journal} {\bibinfo{journal}
		{Physical Review E}\ }\textbf {\bibinfo {volume} {95}},\ \bibinfo {pages}
	{032111} (\bibinfo {year} {2017})}.


\bibitem{Lisboa} V. F. Lisboa, P. R. Dieguez,  J. R. Guimaraes, J. F. G. Santos and R. M. Serra, Experimental investigation of a quantum heat engine powered by generalized measurements, \href {\doibase
https://doi.org/10.1103/PhysRevA.106.022436} {\bibfield {journal} {\bibinfo{journal}
		{Phys. Rev. A}\ }\textbf {\bibinfo {volume} {106}},\ \bibinfo {pages}
	{022436} (\bibinfo {year} {2022})}.

\bibitem{Shanhe} Z. Lin, S. Su, J. Chen, J. Chen and J. F. G. Santos, Suppressing coherence effects in quantum-measurement based engines, \href {\doibase
https://doi.org/10.1103/PhysRevA.104.062210} {\bibfield {journal} {\bibinfo{journal}
		{Phys. Rev. A}\ }\textbf {\bibinfo {volume} {104}},\ \bibinfo {pages}
	{062210} (\bibinfo {year} {2021})}.


\bibitem{Bresque} L. Bresque, P. A. Camati, S. Rogers, K. Murch, A. N. Jordan and A. Auffèves, Two-qubit engine fueled by entanglement and local measurements, \href {\doibase
https://doi.org/10.1103/PhysRevLett.126.120605} {\bibfield {journal} {\bibinfo{journal}
		{Phys. Rev. Lett.}\ }\textbf {\bibinfo {volume} {126}},\ \bibinfo {pages}
	{120605} (\bibinfo {year} {2021})}.

\bibitem{Behzadi} N. Behzadi, Quantum engine based on general measurements, \href {\doibase
https://doi.org/10.1088/1751-8121/abca74} {\bibfield {journal} {\bibinfo{journal}
		{J. Phys. A: Math. Theor}\ }\textbf {\bibinfo {volume} {54}},\ \bibinfo {pages}
	{015304} (\bibinfo {year} {2021})}.

\bibitem{JuzarSong} J. Son, P. Talkner and J. Thingna, Monitoring quantum Otto engines, \href {\doibase
https://doi.org/10.1103/PRXQuantum.2.040328} {\bibfield {journal} {\bibinfo{journal}
		{Phys. Rev. X Quantum}\ }\textbf {\bibinfo {volume} {2}},\ \bibinfo {pages}
	{040328} (\bibinfo {year} {2021})}.

\bibitem{JordanNM} A. N. Jordan, C. Elouard and A. Auffèves, Quantum measurement engines and their relevance for quantum interpretations,  https://doi.org/10.48550/arXiv.1911.06838.	

\bibitem{Yamamoto} T. Yamamoto and Y. Tokura, Heat flow from a measurement apparatus monitoring a dissipative qubit, \href {\doibase
https://doi.org/10.1103/PhysRevResearch.6.013300} {\bibfield {journal} {\bibinfo{journal}
		{Phys. Rev. Research}\ }\textbf {\bibinfo {volume} {6}},\ \bibinfo {pages}
	{013300} (\bibinfo {year} {2024})}.

\bibitem{BiswasP} C. Purkait and A. Biswas, Measurement-based quantum Otto engine with a two-spin system coupled
by anisotropic interaction: enhanced efficiency at finite times, \href {\doibase
https://doi.org/10.1103/PhysRevE.107.054110} {\bibfield {journal} {\bibinfo{journal}
		{Phys. Rev. E}\ }\textbf {\bibinfo {volume} {107}},\ \bibinfo {pages}
	{054110} (\bibinfo {year} {2023})}.

\bibitem{SantosJ} J. F. G. Santos and P. Chattopadhyay, PT -symmetric effects in measurement-based quantum thermal
machines, \href {\doibase
https://doi.org/10.1016/j.physa.2023.129342} {\bibfield {journal} {\bibinfo{journal}
		{Physica A}\ }\textbf {\bibinfo {volume} {632}},\ \bibinfo {pages}
	{129342} (\bibinfo {year} {2023})}.

\bibitem{BhandariJ} B. Bhandari, R. Czupryniak, P.A. Erdman and A.N. Jordan, Measurement-Based Quantum Thermal Machines
with Feedback Control, \href {\doibase
https://doi.org/10.3390/e25020204} {\bibfield {journal} {\bibinfo{journal}
		{Entropy}\ }\textbf {\bibinfo {volume} {25}},\ \bibinfo {pages}
	{204} (\bibinfo {year} {2023})}.

\bibitem{Perna} G. Perna and E. Calzetta, Limits on quantum measurement engines, \href {\doibase
https://doi.org/10.1103/PhysRevE.109.044102} {\bibfield {journal} {\bibinfo{journal}
		{Phys. Rev. E}\ }\textbf {\bibinfo {volume} {109}},\ \bibinfo {pages}
	{044102} (\bibinfo {year} {2024})}.

\bibitem{Robert} R. Raussendorf and H. J. Briegel, Quantum computing via measurements only, https://doi.org/10.48550/arXiv.quant-ph/0010033.

\bibitem{Abdelkader} A. El Makouri, A. Slaoui, and R. Ahl Laamara, Monitored nonadiabatic and coherent-controlled quantum unital Otto heat engines: First four cumulants, \href {\doibase
https://doi.org/10.1103/physreve.108.044114} {\bibfield {journal} {\bibinfo{journal}
		{Phys. Rev. E}\ }\textbf {\bibinfo {volume} {108}},\ \bibinfo {pages}
	{044114} (\bibinfo {year} {2023})}.
	



\bibitem{Gerry1} S. Saryal, M. Gerry, I. Khait, D. Segal and B. K. Agarwalla,
Universal bounds on fluctuations in continuous
thermal machines, \href {\doibase
https://doi.org/10.1103/Phys Rev Lett.127.190603} {\bibfield {journal} {\bibinfo{journal}
		{Phys. Rev. Lett.}\ }\textbf {\bibinfo {volume} {127}},\ \bibinfo {pages}
	{190603} (\bibinfo {year} {2021})}.


\bibitem{Gerry2} M. Gerry, N. Kalantar and D. Segal, Bounds on fluctuations for ensembles of quantum thermal machines, \href {\doibase
https://doi.org/10.1088/1751-8121/ac4c10} {\bibfield {journal} {\bibinfo{journal}
		{J. Phys. A: Math. Theor.}\ }\textbf {\bibinfo {volume} {55}},\ \bibinfo {pages}
	{104005} (\bibinfo {year} {2022})}.

\bibitem{Watanabe} K. Ito, C. Jiang and G. Watanabe, Universal Bounds for Fluctuations in Small Heat Engines, (2019),
\href {\doibase
https://doi.org/10.48550/arXiv.1910.08096
} {\bibfield {journal} {\bibinfo{journal}
		{arXiv:1910.08096 [cond-mat.stat-mech].}\ }}

\bibitem{Saryal} S. Saryal and B. K. Agarwalla, Bounds on fluctuations for finite-time quantum otto cycle, \href {\doibase
https://doi.org/10.1103/PhysRevE.103.L060103} {\bibfield {journal} {\bibinfo{journal}
		{Phys. Rev. E}\ }\textbf {\bibinfo {volume} {103}},\ \bibinfo {pages}
	{L060103} (\bibinfo {year} {2021})}.


\bibitem{Sandipan} S. Mohanta, S. Saryal and B. K. Agarwalla, Universal bounds on cooling power and cooling efficiency for autonomous absorption refrigerators, \href {\doibase
https://doi.org/10.1103/PhysRevE.105.034127} {\bibfield {journal} {\bibinfo{journal}
		{Physical Review E}\ }\textbf {\bibinfo {volume} {105}},\ \bibinfo {pages}
	{034127} (\bibinfo {year} {2022})}.

\bibitem{Gerry} M. Gerry, N. Kalantar and D. Segal, Bounds on fluctuations for ensembles of quantum thermal machines, \href {\doibase
https://doi.org/10.1088/1751-8121/ac4c10} {\bibfield {journal} {\bibinfo{journal}
		{J. Phys. A : Math. Theor.}\ }\textbf {\bibinfo {volume} {55}},\ \bibinfo {pages}
	{104005} (\bibinfo {year} {2022})}.


\bibitem{Watanabe2} G.-H. Xu, C. Jiang, Y. Minami and G. Watanabe, Relation between fluctuations and efficiency at maximum power for small heat engines, \href {\doibase
https://doi.org/10.1103/PhysRevResearch.4.043139} {\bibfield {journal} {\bibinfo{journal}
		{Phys. Rev. Res.}\ }\textbf {\bibinfo {volume} {4}},\ \bibinfo {pages}
	{043139} (\bibinfo {year} {2022})}.

\bibitem{Mohanta} S. Mohanta, M. Saha, B. P. Venkatesh and B. K. Agarwalla, Study of bounds on non-equilibrium fuctuations for asymmetrically driven quantum Otto engine, \href {\doibase
https://doi.org/10.1103/PhysRevE.108.014118} {\bibfield {journal} {\bibinfo{journal}
		{Phys. Rev. E}\ }\textbf {\bibinfo {volume} {108}},\ \bibinfo {pages}
	{014118} (\bibinfo {year} {2023})}.

\bibitem{Abdelkader2} A. El Makouri, A. Slaoui and R. Ahl Laamara, Monitored quantum unital Otto heat engines: Fluctuations of the released heat and entropy production, \href {\doibase
https://doi.org/10.1016/j.physa.2025.130591} {\bibfield {journal} {\bibinfo{journal}
		{Physica A}\ }\textbf  (\bibinfo {year} {2025})}.



\bibitem{Sacchi} M. F. Sacchi, Thermodynamic uncertainty relations for bosonic Otto engines, \href {\doibase
https://doi.org/10.1103/PhysRevE.103.012111} {\bibfield {journal} {\bibinfo{journal}
		{Phys. Rev. E}\ }\textbf {\bibinfo {volume} {103}},\ \bibinfo {pages}
	{012111} (\bibinfo {year} {2021})}.

\bibitem{Barato} A. C. Barato and U. Seifert, Thermodynamic Uncertainty Relation for Biomolecular Processes, \href {\doibase
https://doi.org/10.1103/PhysRevLett.114.158101} {\bibfield {journal} {\bibinfo{journal}
		{Phys. Rev. Lett.}\ }\textbf {\bibinfo {volume} {114}},\ \bibinfo {pages}
	{158101} (\bibinfo {year} {2015})}.

\bibitem{Timpanaro} A. M. Timpanaro, G. Guarnieri, J. Goold and G. T. Landi, Thermodynamic uncertainty relations from exchange fluctuation theorems, \href {\doibase
https://doi.org/10.1103/PhysRevLett.123.090604} {\bibfield {journal} {\bibinfo{journal}
		{Phys. Rev. Lett.}\ }\textbf {\bibinfo {volume} {114}},\ \bibinfo {pages}
	{090604} (\bibinfo {year} {2019})}.

\bibitem{Horowitz} T. R. Gingrich, J. M. Horowitz, N. Perunov and J. L. England, Dissipation Bounds All Steady-State Current Fluctuations, \href {\doibase
https://doi.org/10.1103/PhysRevLett.116.120601} {\bibfield {journal} {\bibinfo{journal}
		{Phys. Rev. Lett}\ }\textbf {\bibinfo {volume} {116}},\ \bibinfo {pages}
	{120601} (\bibinfo {year} {2016})}.


\bibitem{Horowitz2} J. Horowitz and T. Gingrich, Thermodynamic uncertainty relations constrain non-equilibrium fluctuations, \href {\doibase
https://doi.org/10.1038/s41567-019-0702-6} {\bibfield {journal} {\bibinfo{journal}
		{Nat. Phys.}\ }\textbf {\bibinfo {volume} {16}},\ \bibinfo {pages}
	{15-20} (\bibinfo {year} {2020})}.

\bibitem{England} T. R. Gingrich, J. M. Horowitz, N. Perunov and J. L. England, Dissipation bounds all steady state current fuctuations, \href {\doibase
https://doi.org/10.1103/PhysRevLett.116.120601} {\bibfield {journal} {\bibinfo{journal}
		{Phys. Rev. Lett.}\ }\textbf {\bibinfo {volume} {116}},\ \bibinfo {pages}
	{120601} (\bibinfo {year} {2016})}.


\bibitem{Falasco} G. Falasco, M. Esposito, and J.-C. Delvenne, Unifying thermodynamic uncertainty relations, \href {\doibase
https://doi.org/10.1088/1367-2630/ab8679} {\bibfield {journal} {\bibinfo{journal}
		{New J. Phys.}\ }\textbf {\bibinfo {volume} {22}},\ \bibinfo {pages}
	{053046} (\bibinfo {year} {2020})}.


\bibitem{Brandner} K. Macieszczak, K. Brandner and J. P. Garrahan, Unified thermodynamic uncertainty relations in linear response, \href {\doibase
https://doi.org/10.1103/PhysRevLett.121.130601} {\bibfield {journal} {\bibinfo{journal}
		{Phys. Rev. Lett.}\ }\textbf {\bibinfo {volume} {121}},\ \bibinfo {pages}
	{130601} (\bibinfo {year} {2018})}.

\bibitem{Brandner2} K. Brandner, T. Hanazato and K. Saito, Thermodynamic bounds on precision in ballistic multiterminal transport, \href {\doibase
https://doi.org/10.1103/PhysRevLett.120.090601} {\bibfield {journal} {\bibinfo{journal}
		{Phys. Rev. Lett.}\ }\textbf {\bibinfo {volume} {120}},\ \bibinfo {pages}
	{090601} (\bibinfo {year} {2018})}.

\bibitem{LiuSegal} J. Liu and D. Segal, Thermodynamic uncertainty relation in quantum thermoelectric junctions, \href {\doibase
https://doi.org/10.1103/PhysRevE.99.062141} {\bibfield {journal} {\bibinfo{journal}
		{Phys. Rev. E}\ }\textbf {\bibinfo {volume} {99}},\ \bibinfo {pages}
	{062141} (\bibinfo {year} {2019})}.

\bibitem{Proesmans} K. Proesmans and J. M Horowitz, Hysteretic thermodynamic funcertainty relation for systems with broken timereversal symmetry, \href {\doibase
https://doi.org/10.1088/1742-5468/ab14da} {\bibfield {journal} {\bibinfo{journal}
		{J. Stat. Mech.}\ }\textbf \ (\bibinfo {year} {2019}) \bibinfo {pages}
	{054005} }.

\bibitem{Hasegawa} Y. Hasegawa and T. Van Vu, Fluctuation Theorem Uncertainty Relation, \href {\doibase
https://doi.org/10.1103/PhysRevLett.123.110602} {\bibfield {journal} {\bibinfo{journal}
		{Phys. Rev. Lett.}\ }\textbf {\bibinfo {volume} {123}},\ \bibinfo {pages}
	{110602} (\bibinfo {year} {2019})}.

\bibitem{Kosloff} R. Kosloff, A quantum mechanical open system as a model of a heat engine, \href {\doibase
https://doi.org/10.1063/1.446862} {\bibfield {journal} {\bibinfo{journal}
		{J. Chem. Phys.}\ }\textbf {\bibinfo {volume} {80}},\ \bibinfo {pages}
	{1625} (\bibinfo {year} {1984})}.

\bibitem{Kieu} T. D. Kieu, The second law, Maxwell’s demon, and work derivable from quantum heat engines, \href {\doibase
https://doi.org/10.1103/PhysRevLett.93.140403} {\bibfield {journal} {\bibinfo{journal}
		{Phys. Rev. Lett.}\ }\textbf {\bibinfo {volume} {93}},\ \bibinfo {pages}
	{140403} (\bibinfo {year} {2004})}.

\bibitem{Quan} H. T. Quan, Yi-xi Liu, C. P. Sun and F. Nori, Quantum thermodynamic cycles and quantum heat engines, \href {\doibase
https://doi.org/10.1103/PhysRevE.76.031105} {\bibfield {journal} {\bibinfo{journal}
		{Phys. Rev. E}\ }\textbf {\bibinfo {volume} {76}},\ \bibinfo {pages}
	{031105} (\bibinfo {year} {2007})}.

\bibitem{Quan2} H. T. Quan, Quantum thermodynamic cycles and quantum heat engines. II, \href {\doibase
https://doi.org/10.1103/PhysRevE.79.041129} {\bibfield {journal} {\bibinfo{journal}
		{Phys. Rev. E}\ }\textbf {\bibinfo {volume} {79}},\ \bibinfo {pages}
	{041129} (\bibinfo {year} {2009})}.


\bibitem{Rezek} R. Kosloff and Y. Rezek, The quantum harmonic Otto cycle, \href {\doibase
https://doi.org/10.3390/e19040136} {\bibfield {journal} {\bibinfo{journal}
		{Entropy}\ }\textbf {\bibinfo {volume} {19}},\ \bibinfo {pages}
	{136} (\bibinfo {year} {2017})}.


\bibitem{Johal} G. Thomas and R. S. Johal, A Coupled Quantum Otto Cycle, \href {\doibase
https://doi.org/10.1103/PhysRevE.83.031135} {\bibfield {journal} {\bibinfo{journal}
		{Phys. Rev. E}\ }\textbf {\bibinfo {volume} {76}},\ \bibinfo {pages}
	{031105} (\bibinfo {year} {2007})}.
	
	
	
	\bibitem{Bernardo} Bertulio de Lima Bernardo, Unraveling the role of coherence in the first law of quantum thermodynamics,
\href {\doibase
https://doi.org/10.1103/PhysRevE.102.062152} {\bibfield {journal} {\bibinfo{journal}
		{Phys. Rev. E}\ }\textbf {\bibinfo {volume} {102}},\ \bibinfo {pages}
	{062152} (\bibinfo {year} {2020})}.


\bibitem{Chen} S. Su1, J.-F. Chen, Y.-H. Ma1, J.-C. Chen, and C.-P. Sun, The heat and work of quantum thermodynamic processes with quantum coherence, \href {\doibase
https://doi.org/10.1088/1674-1056/27/6/060502} {\bibfield {journal} {\bibinfo{journal}
		{Chinese Phys. B}\ }\textbf {\bibinfo {volume} {27}},\ \bibinfo {pages}
	{060502} (\bibinfo {year} {2018})}.
	

\bibitem{Alipour} S. Alipour, A. T. Rezakhani, A. Chenu, A. del Campo, and T. Ala-Nissila, Entropy-based formulation of thermodynamics in arbitrary quantum evolution, \href {\doibase
https://doi.org/10.1103/PhysRevA.105.L040201} {\bibfield {journal} {\bibinfo{journal}
		{Phys. Rev. A}\ }\textbf {\bibinfo {volume} {105}},\ \bibinfo {pages}
	{L040201} (\bibinfo {year} {2022})}. 


\bibitem{Ahmadi} B. Ahmadi, S. Salimi, and A. S. Khorashad, On the contribution of work or heat in exchanged energy via interaction in open bipartite quantum systems, \href {\doibase
https://doi.org/10.1038/s41598-022-27156-0} {\bibfield {journal} {\bibinfo{journal}
		{Sci. Rep.}\ }\textbf {\bibinfo {volume} {13}},\ \bibinfo {pages}
	{160} (\bibinfo {year} {2023})}.
	

\bibitem{Tasaki}	H. Tasaki, Jarzynski Relations for Quantum Systems and
Some Applications, \href {\doibase
https://doi.org/10.48550/arXiv.cond-mat/0009244} {\bibfield {journal} {\bibinfo{journal}
		{arXiv:cond-mat/0009244(2000)}\ }}.
	

\bibitem{Campisi1} M. Campisi, P. Hänggi and P. Talkner, Colloquium: Quantum fluctuation relations: Foundations and applications, \href {\doibase
https://doi.org/10.1103/RevModPhys.83.771} {\bibfield {journal} {\bibinfo{journal}
		{Rev. Mod. Phys.}\ }\textbf {\bibinfo {volume} {83}},\ \bibinfo {pages}
	{771} (\bibinfo {year} {2011})}.


\bibitem{Esposito} M. Esposito, U. Harbola and S. Mukamel, Nonequilibrium fluctuations, fluctuation theorems, and counting statistics in quantum systems, \href {\doibase
https://doi.org/10.1103/RevModPhys.81.1665} {\bibfield {journal} {\bibinfo{journal}
		{Rev. Mod. Phys.}\ }\textbf {\bibinfo {volume} {81}},\ \bibinfo {pages}
	{1665} (\bibinfo {year} {2009})}.

\bibitem{MatteoLevy} A. Levy and M. Lostaglio, A quasiprobability distribution for heat fluctuations in the quantum regime, \href {\doibase
https://doi.org/10.1103/PRXQuantum.1.010309} {\bibfield {journal} {\bibinfo{journal}
		{PRX Quantum}\ }\textbf {\bibinfo {volume} {1}},\ \bibinfo {pages}
	{010309} (\bibinfo {year} {2020})}.

\bibitem{Kirkwood} J. G. Kirkwood, Quantum statistics of almost classical assemblies, \href {\doibase
https://doi.org/10.1103/PhysRev.44.31} {\bibfield {journal} {\bibinfo{journal}
		{Phys. Rev.}\ }\textbf {\bibinfo {volume} {44}},\ \bibinfo {pages}
	{31} (\bibinfo {year} {1933})}.

\bibitem{Dirac} P. A. M. Dirac, On the analogy between classical and quantum mechanics, \href {\doibase
https://doi.org/10.1103/RevModPhys.17.195} {\bibfield {journal} {\bibinfo{journal}
		{Rev. Mod. Phys.}\ }\textbf {\bibinfo {volume} {17}},\ \bibinfo {pages}
	{195} (\bibinfo {year} {1945})}.

\bibitem{Francica2} G. Francica, What is the most general class of quasiprobabilities of work? \href {\doibase
https://doi.org/10.1103/PhysRevE.106.054129} {\bibfield {journal} {\bibinfo{journal}
		{Phys. Rev. E}\ }\textbf {\bibinfo {volume} {106}},\ \bibinfo {pages}
	{054129} (\bibinfo {year} {2022})}.


\bibitem{Chiara} S. Gherardini and G. Chiara, Quasiprobabilities in quantum thermodynamics and many-body systems: A tutorial, \href {\doibase
https://doi.org/10.1103/PRXQuantum.5.030201} {\bibfield {journal} {\bibinfo{journal}
		{PRX Quantum}\ }\textbf {\bibinfo {volume} {5}},\ \bibinfo {pages}
	{030201} (\bibinfo {year} {2024})}.

\bibitem{NicoleY}  David R. M. Arvidsson-Shukur, William F. Braasch Jr., Stephan De Bievre, Justin Dressel, Andrew N. Jordan, Christopher Langrenez, Matteo Lostaglio, Jeff S. Lundeen and Nicole Yunger Halpern, Properties and Applications of the Kirkwood-Dirac Distribution, (2024), \href {\doibase
https://doi.org/10.48550/arXiv.2403.18899
} {\bibfield {journal} {\bibinfo{journal}
		{ arXiv:2403.18899 [quant-ph].}\ } }.

\bibitem{Bievre} S. De Bièvre, Complete Incompatibility, Support Uncertainty, and Kirkwood-Dirac Nonclassicality, \href {\doibase
https://doi.org/10.1103/PhysRevLett.127.190404} {\bibfield {journal} {\bibinfo{journal}
		{Phys. Rev. Lett.}\ }\textbf {\bibinfo {volume} {127}},\ \bibinfo {pages}
	{190404} (\bibinfo {year} {2021})}.

\bibitem{Belenchia} M. Lostaglio, A. Belenchia, A. Levy, S. Hernández-Gómez, N. Fabbri and S. Gherardini, Kirkwood-Dirac quasiprobability approach to the statistics of incompatible observables, \href {\doibase
	https://doi.org/10.22331/q-2023-10-09-1128} {\bibfield {journal} {\bibinfo{journal}
		{Quantum}\ }\textbf {\bibinfo {volume} {7}},\ \bibinfo {pages}
	{1128} (\bibinfo {year} {2023})}.

\bibitem{Santini} A. Santini, A. Solfanelli, S. Gherardini and M. Collura, Work statistics, quantum signatures, and enhanced work extraction in quadratic fermionic models, \href {\doibase
https://doi.org/10.1103/PhysRevB.108.104308} {\bibfield {journal} {\bibinfo{journal}
		{Phys. Rev. B}\ }\textbf {\bibinfo {volume} {108}},\ \bibinfo {pages}
	{104308} (\bibinfo {year} {2023})}.


\bibitem{Dipojono} A. Budiyono and H. K. Dipojono, Quantifying quantum coherence via Kirkwood-Dirac quasiprobability, \href {\doibase
https://doi.org/10.1103/PhysRevA.107.022408} {\bibfield {journal} {\bibinfo{journal}
		{Phys. Rev. A}\ }\textbf {\bibinfo {volume} {107}},\ \bibinfo {pages}
	{022408} (\bibinfo {year} {2023})}.


\bibitem{HTQuan} Ji-Hui Pei, Jin-Fu Chen and H. T. Quan, Exploring quasiprobability approaches to quantum work in the presence of initial coherence: Advantages of the Margenau-Hill distribution, \href {\doibase
https://doi.org/10.1103/PhysRevE.108.054109} {\bibfield {journal} {\bibinfo{journal}
		{Phys. Rev. E}\ }\textbf {\bibinfo {volume} { 108}},\ \bibinfo {pages}
	{054109} (\bibinfo {year} {2023})}.

\bibitem{Francica1} G. Francica and L. Dell’Anna, Quasiprobability distribution of work in the quantum Ising model, \href {\doibase
https://doi.org/10.1103/PhysRevE.108.014106} {\bibfield {journal} {\bibinfo{journal}
		{Phys. Rev. E}\ }\textbf {\bibinfo {volume} { 108}},\ \bibinfo {pages}
	{014106} (\bibinfo {year} {2023})}.

\bibitem{Hill} H. Margenau and R. N. Hill, Correlation between measurements
in quantum theory, \href {\doibase
https://doi.org/10.1143/PTP.26.722} {\bibfield {journal} {\bibinfo{journal}
		{Prog. Theor. Phys.}\ }\textbf {\bibinfo {volume} {26}},\ \bibinfo {pages}
	{722} (\bibinfo {year} {1961})}.

\bibitem{Struchtrup} H. Struchtrup, Work storage in states of apparent negative thermodynamic temperature, \href {\doibase
https://doi.org/10.1103/PhysRevLett.120.250602} {\bibfield {journal} {\bibinfo{journal}
		{Phys. Rev. Lett.}\ }\textbf {\bibinfo {volume} { 120}},\ \bibinfo {pages}
	{250602} (\bibinfo {year} {2018})}.

\bibitem{Warren} D. Frenkel and P. B. Warren, Gibbs, Boltzmann, and negative temperatures,\href {\doibase
https://doi.org/10.1119/1.4895828
} {\bibfield {journal} {\bibinfo{journal}
		{Am. J. Phys.}\ }\textbf {\bibinfo {volume} {83}},\ \bibinfo {pages}
	{163} (\bibinfo {year} {2015})}.


\bibitem{Zener1} S. N. Shevchenko, S. Ashhab and F. Nori, Landau–Zener–Stückelberg interferometry, \href {\doibase
https://doi.org/10.1016/j.physrep.2010.03.002} {\bibfield {journal} {\bibinfo{journal}
		{Phys. Rep.}\ }\textbf {\bibinfo {volume} {1}},\ \bibinfo {pages}
	{492} (\bibinfo {year} {2010})}.

\bibitem{Zener2} J. Thingna, F. Barra and M. Esposito, Kinetics and
thermodynamics of a driven open quantum system, \href {\doibase
https://doi.org/10.1103/PhysRevE.96.052132} {\bibfield {journal} {\bibinfo{journal}
		{Phys. Rev. E}\ }\textbf {\bibinfo {volume} {96}},\ \bibinfo {pages}
	{052132} (\bibinfo {year} {2017})}.

\bibitem{Zener3} J. Thingna, M. Esposito and F. Barra, Landau-Zener Lindblad equation and work extraction from coherences, \href {\doibase
https://doi.org/10.1103/PhysRevE.99.042142} {\bibfield {journal} {\bibinfo{journal}
		{Phys. Rev. E}\ }\textbf {\bibinfo {volume} {99}},\ \bibinfo {pages}
	{042142} (\bibinfo {year} {2019})}.

\bibitem{Denzler} T. Denzler and E. Lutz, Efficiency fluctuations of a quantum heat engine, \href {\doibase
https://doi.org/10.1103/PhysRevResearch.2.032062} {\bibfield {journal} {\bibinfo{journal}
		{Phys. Rev. Res.}\ }\textbf {\bibinfo {volume} {2}},\ \bibinfo {pages}
	{032062} (\bibinfo {year} {2020})}.
	
	\bibitem{JPeterson} John P. S. Peterson, Tiago B. Batalhão, Marcela Herrera, Alexandre M. Souza, Roberto S. Sarthour, Ivan S. Oliveira, and Roberto M. Serra, Experimental Characterization of a Spin Quantum Heat Engine, \href {\doibase
https://doi.org/10.1103/PhysRevLett.123.240601} {\bibfield {journal} {\bibinfo{journal}
		{Phys. Rev. Lett.}\ }\textbf {\bibinfo {volume} {123}},\ \bibinfo {pages}
	{240601} (\bibinfo {year} {2019})}.


\bibitem{Camati}	P. A. Camati, J. F. G. Santos, and R. M. Serra, Coherence effects in the performance of the quantum Otto heat engine, \href {\doibase
https://doi.org/10.1103/PhysRevA.99.062103} {\bibfield {journal} {\bibinfo{journal}
		{Phys. Rev. A}\ }\textbf {\bibinfo {volume} {99}},\ \bibinfo {pages}
	{062103} (\bibinfo {year} {2019})}.

\bibitem{DenzlerR} Tobias Denzler, Jonas F. G. Santos, Eric Lutz, and Roberto M. Serra,  Nonequilibrium fuctuations of a quantum heat enginee, \href {\doibase
https://doi.org/10.1088/2058-9565/ad6287} {\bibfield {journal} {\bibinfo{journal}
		{Quantum Sci. Technol.}\ }\textbf {\bibinfo {volume} {9}},\ \bibinfo {pages}
	{045017} (\bibinfo {year} {2024})}.

\bibitem{Fei} Zhaoyu Fei, Jin-Fu Chen and Yu-Han Ma, Efficiency statistics of a quantum Otto cycle, \href {\doibase
https://doi.org/10.1103/PhysRevA.105.022609} {\bibfield {journal} {\bibinfo{journal}
		{Phys. Rev. A}\ }\textbf {\bibinfo {volume} {105}},\ \bibinfo {pages}
	{022609} (\bibinfo {year} {2020})}.


\bibitem{KosloffR} R. Kosloff and T. Feldmann, Discrete four-stroke quantum
heat engine exploring the origin of friction, \href {\doibase
https://doi.org/10.1103/PhysRevE.65.055102} {\bibfield {journal} {\bibinfo{journal}
		{Phys. Rev. E}\ }\textbf {\bibinfo {volume} {65}},\ \bibinfo {pages}
	{055102(R)} (\bibinfo {year} {2002})}.


\bibitem{Feldmann} T. Feldmann and R. Kosloff, Quantum four-stroke heat engine: Thermodynamic observables in a model with intrinsic
friction, \href {\doibase
https://doi.org/10.1103/PhysRevE.68.016101} {\bibfield {journal} {\bibinfo{journal}
		{Phys. Rev. E}\ }\textbf {\bibinfo {volume} {68}},\ \bibinfo {pages}
	{016101} (\bibinfo {year} {2003})}.

\bibitem{Alecce} F. Plastina, A. Alecce, T. J. G. Apollaro, G. Falcone, G. Francica, F. Galve, N. Lo Gullo and R. Zambrini, Irreversible
Work and Inner Friction in Quantum Thermodynamic Processes, \href {\doibase
https://doi.org/10.1103/PhysRevLett.113.260601} {\bibfield {journal} {\bibinfo{journal}
		{Phys. Rev. Lett.}\ }\textbf {\bibinfo {volume} {113}},\ \bibinfo {pages}
	{260601} (\bibinfo {year} {2014})}.

\bibitem{Abdelkader3} A. E. Makouri, A. Slaoui and M. Daoud, Enhancing the performance of coupled quantum Otto thermal machines without entanglement and quantum correlations, \href {\doibase
https://doi.org/10.1088/1361-6455/acc36d} {\bibfield {journal} {\bibinfo{journal}
		{J. Phys. B : At. Mol. Opt. Phys.}\ }\textbf {\bibinfo {volume} {56}},\ \bibinfo {pages}
	{085501} (\bibinfo {year} {2023})}.
	
	
	\bibitem{THTH} Xiao, Yang and Liu, Dehua and He, Jizhou and Zhuang, Lin and Liu, Wu-Ming and Yan, L.-L and Wang, Jianhui, Thermodynamics and fluctuations in finite-time quantum heat engines under reservoir squeezing, \href {\doibase
10.1103/PhysRevResearch.5.043185.} {\bibfield {journal} {\bibinfo{journal}
		{Phys. Rev. Res.}\ }\textbf {\bibinfo {volume} {5}},\ \bibinfo {pages}
	{043185} (\bibinfo {year} {2023})}.
	
	


	\bibitem{THTH2} Liu, Dehua and Hong, Yingying and Luo, Shaolin and He, Xian and Wu, Zhaoqi and Wang, Jianhui, Feedback-driven quantum Otto-like engines, \href {\doibase
10.1103/10.1103/PhysRevA.111.012203} {\bibfield {journal} {\bibinfo{journal}
		{Phys. Rev. A}\ }\textbf {\bibinfo {volume} {111}},\ \bibinfo {pages}
	{012203} (\bibinfo {year} {2025})}.
	
	
	
	
	

	

\bibitem{Baumer} Martí Perarnau-Llobet, Elisa Bäumer, Karen V. Hovhannisyan, Marcus Huber, and Antonio
Acin, No-Go Theorem
for the Characterization of Work Fluctuations in Coherent Quantum Systems, \href {\doibase
https://link.aps.org/doi/10.1103/PhysRevLett.118.070601.} {\bibfield {journal} {\bibinfo{journal}
		{Phys Rev Lett.}\ }\textbf {\bibinfo {volume} {118}},\ \bibinfo {pages}
	{070601} (\bibinfo {year} {2017})}.


\bibitem{GomeZ} Santiago Hernández-Gómez, Takuya Isogawa, Alessio Belenchia, Amikam Levy, Nicole Fabbri, Stefano Gherardini, and Paola Cappellaro, Interferometry of quantum correlation functions to access quasiprobability
distribution of work, (2024), \href {\doibase
https://doi.org/10.48550/arXiv.2405.21041} {\bibfield {journal} {\bibinfo{journal}
		{arXiv.2405.21041 [quant-ph]}\ }}.





	
\bibitem{Andrea} S. Gherardini, M. M. Muller, A. Trombettoni, S. Ruo, and F. Caruso, Reconstructing quantum entropy production to
probe irreversibility and correlations, \href {\doibase
https://doi.org/10.1088/2058-9565/aac7e1} {\bibfield {journal} {\bibinfo{journal}
		{Quantum Sci. Technol.}\ }\textbf {\bibinfo {volume} {3}},\ \bibinfo {pages}
	{035013} (\bibinfo {year} {2018})}.
	
	
	\bibitem{Zanardi} T. Albash, D. A. Lidar, M. Marvian, and P. Zanardi, Fluctuation theorems for quantum processes, \href {\doibase
https://doi/10.1103/PhysRevE.88.032146} {\bibfield {journal} {\bibinfo{journal}
		{Phys. Rev. E}\ }\textbf {\bibinfo {volume} {88}},\ \bibinfo {pages}
	{032146} (\bibinfo {year} {2013})}.
	



	
	\bibitem{Landii} G. T. Landi and M. Paternostro, Irreversible entropy production: From classical to quantum, \href{https://link.aps.org/doi/10.1103/RevModPhys.93.035008}{Rev. Mod. Phys.} \textbf{93}, 035008 (2021).

	
	
	
\bibitem{DeffnerT} S. Deffner and E. Lutz, Generalized Clausius Inequality for Nonequilibrium Quantum Processes, \href{https://link.aps.org/doi/10.1103/PhysRevLett.105.170402}{Phys. Rev. Lett.} \textbf{105}, 170402 (2010).




\bibitem{DeffnerT2} S. Deffner and E. Lutz, Nonequilibrium Entropy Production for Open Quantum Systems, \href{https://link.aps.org/doi/10.1103/PhysRevLett.107.140404}{Phys. Rev. Lett.} \textbf{107}, 140404 (2011).



\bibitem{Nielsen} M. A. Nielsen and I. L. Chuang, Quantum Computation and Quantum Information, Cambridge University Press Cambridge. (2000).	

	
	



\end{thebibliography}


\appendix

\section{On why eliminating the dephasing channel from Eq. (\ref{chiF}) gives the cumulants of undephased engine}\label{hjkjkl}

Before we give our arguments to convince the reader about the usefulness of Eq. (\ref{chiFF}) we should recall that: \textit{suppose we have a quantum system described by a state $\rho$ and a Hamiltonian $H$. From quantum statistical mechanics — thanks to von Neumann — we know that the average energy of this system denoted by $\langle E\rangle$ is given by,
\begin{equation}
\langle E\rangle=\mathrm{Tr}[\rho H].\label{EEE}
\end{equation}
On the other hand, the variance of the energy of the system denoted by $var(E)$ is given by,
\begin{equation}
var(E)=\mathrm{Tr}[\rho H^{2}]-\mathrm{Tr}[\rho H]^{2}.\label{VVV}
\end{equation}}
Please, note that these two equations can be used for other operators and not only for the Hamiltonian operator.

Let's explain why the CF given by Eq. (\ref{chiFF}) is a good candidate to assess the cumulants of all the thermodynamic quantities of the undephased engine. We have already reported in section \ref{ude} the expression of the average energies for the unmonitored engine. If one is also interested in the fluctuations of these energies then according to the TPM scheme \cite{Campisi1, Esposito,Tasaki} (i.e. Eq (\ref{chiF})) we have,
\begin{equation}
\begin{split}
&
\llangle E_{1}^{2}\rrangle_{c}=\mathrm{Tr}[\rho_{1}H_{1}^{2}]-\mathrm{Tr}[\rho_{1}H_{1}]^{2},
\\ &
\llangle E_{2}^{2}\rrangle_{c}=\mathrm{Tr}[\rho_{2}H_{2}^{2}]-\mathrm{Tr}[\rho_{2}H_{2}]^{2},
\\ &
\llangle E_{3}^{2}\rrangle_{c}=\mathrm{Tr}[\Phi(\Delta_{2}(\rho_{2}))H_{2}^{2}]-\mathrm{Tr}[\Phi(\Delta_{2}(\rho_{2})) H_{2}]^{2},
\\ &
\llangle E_{4}^{2}\rrangle_{c}=\mathrm{Tr}[V\Delta_{2}(\Phi(\Delta_{2}(\rho_{2})))V^{\dagger}H_{1}^{2}]-\mathrm{Tr}[V\Delta_{2}(\Phi(\Delta_{2}(\rho_{2})))V^{\dagger}H_{1}]^{2}.
\end{split}\label{iopp}
\end{equation} 
On the other hand from Eq. (\ref{chiFF}) we have,
\begin{equation}
\begin{split}
&
\langle E_{1}^{2}\rangle_{c}=\mathrm{Tr}[\rho_{1}H_{1}^{2}]-\mathrm{Tr}[\rho_{1}H_{1}]^{2}=\llangle E_{1}^{2}\rrangle_{c},
\\ &
\langle E_{2}^{2}\rangle_{c}=\mathrm{Tr}[\rho_{2}H_{2}^{2}]-\mathrm{Tr}[\rho_{2}H_{2}]^{2}=\llangle E_{2}^{2}\rrangle_{c},
\\ &
\langle E_{3}^{2}\rangle_{c}=\mathrm{Tr}[\rho_{3}H_{2}^{2}]-\mathrm{Tr}[\rho_{3} H_{2}]^{2}\neq \llangle E_{3}^{2}\rrangle_{c},
\\ &
\langle E_{4}^{2}\rangle_{c}=\mathrm{Tr}[\rho_{4} H_{1}^{2}]-\mathrm{Tr}[\rho_{4} H_{1}]^{2}\neq\llangle E_{4}^{2}\rrangle_{c}.\label{thyu}
\end{split}
\end{equation} 
From Eqs. (\ref{EEE})-(\ref{VVV}) we see that the fluctuations of the energies of the engine should be computed with respect to the states and not with respect to the evolved dephased state, something that is satisfied for all variances of the energies that follow from Eq. \ref{chiFF} — contrary to Eq. (\ref{chiF}) where only the fluctuations $E_{1}$ and $E_{2}$ are computed with respect to the states. This shows that from Eq. (\ref{chiFF}) not only one can get the right average energies but also their true fluctuations.

Consider now the case when we are only interested in the fluctuations of $Q_{C}$}. In this case, according to the TPM scheme \cite{Campisi1, Esposito,Tasaki} we should measure the energy  of the engine at the point $\mathbf{A}$ (see Fig. \ref{Otto}) and then evolve it by $V(\Phi
(U(.)U^{\dagger})V^{\dagger}$ and finally we measure the system's energy at the point $\mathbf{D}$ and then construct the probability distribution of $Q_{C}$ — (since $\rho_{1}$ commutes with $H_{1}$ the TPM is sufficient to construct the distribution of $Q_{C})$. In this scenario the variance of $Q_{C}$ is given by,
\begin{equation}
\begin{split}
var(Q_{C}) & =\mathrm{Tr}\left[H_{1}^{2}V(\Phi(U\rho_{1}U^{\dagger}))V^{\dagger}\right]+\mathrm{Tr}\left[H_{1}^{2}\rho_{1}\right]-2\mathrm{Tr}\left[H_{1}V(\Phi(UH_{1}\rho_{1}U^{\dagger}))V^{\dagger}\right]-\langle Q_{C}\rangle^{2}.\label{thyui}
\end{split}
\end{equation}
On the other hand from Eq. (\ref{chiF}) the fluctuations of $Q_{C}$ are given by,
\begin{equation}
\begin{split}
\llangle Q_{C}^{2}\rrangle_{c}  =\mathrm{Tr}\left[H_{1}^{2}V\Delta_{2}(\Phi(\Delta_{2}(U\rho_{1}U^{\dagger})))V^{\dagger}\right]+\mathrm{Tr}\left[H_{1}^{2}\rho_{1}\right]-2\mathrm{Tr}\left[H_{1}V\Delta_{2}(\Phi(\Delta_{2}(UH_{1}\rho_{1}U^{\dagger})))V^{\dagger}\right]-\llangle Q_{C}\rrangle^{2}.
\label{bMp}
\end{split}
\end{equation}
We see that the fluctuations of $Q_{C}$ derived Eq. (\ref{chiFF}) (i.e. the first equation in Eq. (\ref{b6})) agree with those that follow from measuring the system at the start and the end of the cycle i.e. Eq. (\ref{thyui}). On the other hand, we see that the fluctuations of $Q_{C}$ derived from Eq. (\ref{chiF}) do not agree with Eq. (\ref{thyui}) since the variance is computed with respect to the evolved dephased states.

Furthermore, let's consider the fluctuations of the heat released to the cold bath, i.e., $Q_{C}'$. Without going into further details, in Ref. \cite{Abdelkader2}, we proved that the variance, or more generally the $l$th cumulant, of $Q_{C}'$ is the sum of the $l$th cumulant of $E_{1}$ and that of $E_{4}$. Thus, the variance of $Q_{C}'$, denoted by $var(Q_{C}')$, is given by
\begin{equation}
var(Q_{C}')=\langle E_{1}^{2}\rangle_{c}+\langle E_{4}^{2}\rangle_{c}
\end{equation}
If we follow the TPM scheme \cite{Campisi1, Esposito, Tasaki}, then the variance would be $\llangle E_{1}^{2}\rrangle_{c}+\llangle E_{4}^{2}\rrangle_{c}$, and thus coherence would not contribute.

In the field of quantum thermodynamics, one of the important questions is the distinction between work and heat and the assessment of their statistics, since fluctuations are important in the quantum regime. For assessing work in the quantum regime people use the TPM scheme \cite{Campisi1, Esposito,Tasaki}. The latter means that we measure the energy of the system before and after the evolution and then we define work to be the difference between the outcomes. This method works when the initial state of the system is incoherent in the energy eigenbasis. Actually, this scheme has two advantages: (i) the average work derived from TPM agrees with the change in energy for incoherent states and (ii) the validity of the fluctuations theorems. The problem starts to be interesting when the initial state of the system has non-zero coherence. In this case, the first measurement kills all the coherence and (possibly quantum correlations when the system is entangled with e.g. an environment). In Ref. \cite{Baumer} a no-go theorem shows that to study the statistics of work in the presence of coherence one must relax the assumption of positive probabilities. This is the reason that Kirkwood-Dirac quasiprobability started gaining more interest in the field, see Ref. \cite{,MatteoLevy,NicoleY,Chiara,Francica2,Francica1,HTQuan,Santini,Belenchia} for further details. This quasiprobability can be measured experimentally \cite{NicoleY,Chiara} by using e.g. an ancillary system in which the statistics of work would be encoded. The work \cite{HTQuan} has studied the real part of the  Kirkwood-Dirac quasiprobability and shows that it satisfies properties that other proposed schemes in literature do not satisfy. Finally, in Ref. \cite{Francica1} the author has studied the statistics of work in the quantum Ising model and showed that the second cumulants of work have a real and imaginary part.

Days after our paper, the authors of \cite{GomeZ} have shown how to experimentally reconstruct the real and imaginary parts of the Kirkwood-Dirac quasiprobability of work via an interferometric scheme adapted to the
nitrogen-vacancy (NV) center in diamond. The authors have also shown that the second cumulant of work has a real and imaginary part where the latter is linked to non-commutativity. It was found that coherence can not only be helpful in the extraction of work but also in reducing its fluctuations.

\begin{result}
In the main text and this section, we showed that Eq. (\ref{chiFF}) gives the correct average energies (and thus average work and heats), their fluctuations, and the statistics of $Q_{C}$ and $Q_{C}'$. On the other hand Eq. (\ref{chiF}) fails for this task — the reason is quantum coherence in the energy eigenbasis. Therefore, we believe that Eq. (\ref{chiFF}) could also be used to assess the fluctuations of the stochastic work and the stochastic heat absorbed. The fact that the cumulants are computed with respect to the true state can be beneficial to understanding the role of coherence on the engine performance i.e. in terms of enhancing work, efficiency and work reliability as we show in the main text.
\end{result}

\section{Entropy production}\label{PROD}

In Ref. \cite{Abdelkader2}, we showed that the average entropy production of the dephased engine can be written as follows:
\begin{equation}
\llangle \Sigma\rrangle=S(\rho_{4}^{deph})-S(\rho_{1})+S(\rho_{4}^{deph}||\rho_{1}).\label{Entropy}
\end{equation}
Here $\rho_{4}^{deph}:=V \Delta_{2}(\Phi(\Delta_{2}(U\rho_{1}U^{\dagger})))V^{\dagger}$, $S(\rho):=-\mathrm{Tr}\left[\rho\log \rho\right]$ is the von Neumann entropy, and $S(\rho||\sigma):=\mathrm{Tr}\left[\rho(\log \rho-\log \sigma)\right]$ is the Kullback-Leibler divergence. One can show similarly that the $\langle \Sigma\rangle$ satisfies the next set of equalities,
\begin{equation}
\begin{split}
\langle\Sigma\rangle & 
= S(\rho_{4})-S(\rho_{1})+S(\rho_{4}||\rho_{1})
\\&
= S(\rho_{3})-S(\rho_{2})+S(\rho_{4}||\rho_{1}).
\end{split}
\label{Entropy2}
\end{equation}
The states $\rho_{1}$, $\rho_{2}$, $\rho_{3}$, and $\rho_{4}$ are defined in section \ref{ude}. From the first to the second line in Eq. (\ref{Entropy2}), we use the fact that $U$ and $V$ are unitaries; hence, they do not change the entropy. Equation (\ref{Entropy2}) contains two sources of irreversibility: (i) the irreversibility caused by the unital channel $\Phi$ \cite{Andrea,Landii,Zanardi}, i.e., $S(\rho_{3})-S(\rho_{2})$, and (ii) the irreversibility due to the equilibration of the working medium as quantified by $S(\rho_{4}||\rho_{1})$, \cite{DeffnerT,DeffnerT2}.

\section{Arbitrary qubit Hamiltonian $H_{1(2)}$}\label{zedine}
Let's now explain what we mean by arbitrary $H_{1}$ and $H_{2}$. It's known that arbitray qubit Hamiltonians denoted $H_{qubit}$ can be written as follows \cite{Nielsen};
\begin{equation}
H_{qubit}:=\frac{1}{2}\left(\Gamma \mathbb{1}_{2\times 2}+\Sigma \sigma_{x}+\Delta\sigma_{y}+\omega \sigma_{z}\right).\label{qubitH}
\end{equation}
Here $\Gamma$ is a constant that causes a shift in the energy levels of the system, $\omega$ is the energy difference between the ground and the excited states of the qubit, and $\Sigma$ and $\Delta$ characterize the perturbations in the $x-$ and $y-$directions, e.g. a magnetic field in the $x-$ and $y-$directions. Furthermore, $\Gamma$, $\omega$, $\Sigma$ and $\Delta$ are real coefficients. If we diagonalize $H_{qubit}$ we obtain:
\begin{equation}
H_{qubit}=\nu|+\rangle\langle+|+\mu |-\rangle\langle-|,
\end{equation}
where $\nu=(\Sigma+\sqrt{\Gamma^{2}+\Delta^{2}+\omega^{2}})/2$, $\mu=(\Sigma-\sqrt{\Gamma^{2}+\Delta^{2}+\omega^{2}})/2$, and $|+\rangle$ and $|-\rangle$ are the eigenvectors of $H_{qubit}$. Note that $\Gamma$ has the same effect on the eigenvalues, i.e. it only causes a shift in them and thus can be eliminated. Thus we obtain,
\begin{equation}
H_{qubit}=\frac{1}{2}\left(\Sigma \sigma_{x}+\Delta\sigma_{y}+\omega \sigma_{z}\right).\label{qubitHH}
\end{equation}
Now let's now write $H_{1}$ and $H_{2}$ in this form. We have,
\begin{equation}
H_{1(2)}=\frac{1}{2}\left(\Sigma_{1(2)} \sigma_{x}+\Delta_{1(2)}\sigma_{y}+\omega_{1(2)} \sigma_{z}\right).\label{qubit12}
\end{equation} 
When we say that $H_{1}$ and $H_{2}$ are arbitrary we mean that the eigenvalues and their eigenvectors are arbitrary, i.e. $\Sigma$, $\Delta$ and $\omega$ can take any real values thus the Hamiltonian $H_{1}$ and $H_{2}$ are arbitrary since they are a linear combination of $\sigma_{x}$, $\sigma_{y}$ and $\sigma_{z}$.

\section{Expression of the average energies for the undephased engine in terms of the transition probabilities: arbitrary Hamiltonians, unitaries, and unital channels}\label{app1r}
The Hamiltonian $H_{1}$, Eq. (\ref{H11}), can also be written as
\begin{equation}
H_{1}=2\nu_{1}|+\rangle_{11}\langle+|-\nu_{1}\mathbb{1}_{2}\label{H1}.
\end{equation}
Since at the beginning of the cycle, the system is assumed to start in thermal equilibrium, we have
\begin{equation}
\rho_{1}:=p_{g}|-\rangle_{11}\langle-|+p_{e}|+\rangle_{11}\langle+|=-\tanh(\beta\nu_{1})|+\rangle_{11}\langle+|+p_{g}
   \mathbb{1}_{2} \label{rhoo},
\end{equation} 
where $p_{e}=e^{-\beta \nu_{1}}/Z (p_{g}=e^{\beta \nu_{1}}/Z)$ is the excited (ground) state occupation. The second equality in Eq. (\ref{rhoo}) can be proven easily after simple algebra. Note that the partition function $Z$ is given by $Z:=\cosh(\beta \nu_{1})$. Similarly to Eq. (\ref{H1}), $H_{2}$ is given by
\begin{equation}
H_{2}=2\nu_{2}|+\rangle_{22}\langle+|-\nu_{2}\mathbb{1}_{2}\label{H2}.
\end{equation}

Let's now compute the average energies. We start with the average energy  $\langle E_{1}\rangle$. We have,
\begin{equation}
\begin{split}
\langle E_{1}\rangle & =\mathrm{Tr}\left[\rho_{1}H_{1}\right]
\\ &
=\mathrm{Tr}\left[(-\tanh(\beta\nu_{1})|+\rangle_{11}\langle+|+p_{g}
   \mathbb{1}_{2} )(2\nu_{1}|+\rangle_{11}\langle+|-\nu_{1}\mathbb{1}_{2})\right]
   \\ &
=-2\nu_{1}\tanh(\beta\nu_{1})\mathrm{Tr}\left[|+\rangle_{11}\langle+|\right]+\nu_{1}\tanh(\beta\nu_{1})\mathrm{Tr}\left[|+\rangle_{11}\langle+|\right]
+2\nu_{1}p_{g}\mathrm{Tr}\left[|+\rangle_{11}\langle+|\right]
-\nu_{1}p_{g}\mathrm{Tr}\left[\mathbb{1}_{2}\right]
   \\ &
=-2\nu_{1}\tanh(\beta\nu_{1})+\nu_{1}\tanh(\beta\nu_{1})+2\nu_{1}p_{g}-2\nu_{1}p_{g}
\\ &
=-\nu_{1}\tanh(\beta\nu_{1}).
\end{split}\label{e1}
\end{equation}
In the second line, we replace $H_{1}$ and $\rho_{1}$ by their expressions, i.e., Eq. (\ref{H1}) and Eq. (\ref{rhoo}). In the fourth line, we use the fact that $\mathrm{Tr}\left[|+\rangle_{11}\langle+|\right]=1$ and $\mathrm{Tr}\left[\mathbb{1}_{2}\right]=2$. 

For the average energy $\langle E_{2}\rangle$ we have,
\begin{equation}
\begin{split}
\langle E_{2}\rangle & =\mathrm{Tr}\left[\rho_{2}H_{2}\right]
\\ &
=\mathrm{Tr}\left[(-\tanh(\beta\nu_{1})U|+\rangle_{11}\langle+|U^{\dagger}+p_{g}
   \mathbb{1}_{2} )(2\nu_{2}|+\rangle_{22}\langle+|-\nu_{2}\mathbb{1}_{2})\right]
   \\ &
=-2\nu_{2}\tanh(\beta\nu_{1})\mathrm{Tr}\left[U|+\rangle_{11}\langle+|U^{\dagger}|+\rangle_{22}\langle+|\right]+\nu_{2}\tanh(\beta\nu_{1})\mathrm{Tr}\left[U|+\rangle_{11}\langle+|U^{\dagger}\right]
+2\nu_{2}p_{g}\mathrm{Tr}\left[|+\rangle_{22}\langle+|\right]
-\nu_{2}p_{g}\mathrm{Tr}\left[\mathbb{1}_{2}\right]
   \\ &
=-2\nu_{2}\tanh(\beta\nu_{1})(1-\delta')+\nu_{2}\tanh(\beta\nu_{1})+2\nu_{2}p_{g}-2\nu_{2}p_{g}
\\ &
=-\nu_{2}(1-2\delta')\tanh(\beta\nu_{1}).
\end{split}\label{e2}
\end{equation}
In the fourth line, we use Eq. (\ref{deltap}) and the fact that $\delta'=|_{2}\langle +|U|-\rangle_{1}|^{2}=|_{2}\langle -|U|+\rangle_{1}|^{2}$ and $\delta'+|_{2}\langle +|U|+\rangle_{1}|^{2}=1$. The latter equality is nothing but probability conservation. 

The third average enegy is given by,
\begin{equation}
\begin{split}
\langle E_{3}\rangle & =\mathrm{Tr}\left[\rho_{3}H_{2}\right]
\\ &
=\mathrm{Tr}\left[(-\tanh(\beta\nu_{1})\Phi(U|+\rangle_{11}\langle+|U^{\dagger})+p_{g}
   \mathbb{1}_{2} )(2\nu_{2}|+\rangle_{22}\langle+|-\nu_{2}\mathbb{1}_{2})\right]
   \\ &
=-2\nu_{2}\tanh(\beta\nu_{1})\mathrm{Tr}\left[\Phi(U|+\rangle_{11}\langle+|U^{\dagger})|+\rangle_{22}\langle+|\right]+\nu_{2}\tanh(\beta\nu_{1})\mathrm{Tr}\left[\Phi(U|+\rangle_{11}\langle+|U^{\dagger})\right]
+2\nu_{2}p_{g}\mathrm{Tr}\left[|+\rangle_{22}\langle+|\right]
-\nu_{2}p_{g}\mathrm{Tr}\left[\mathbb{1}_{2}\right]
   \\ &
=-2\nu_{2}\tanh(\beta\nu_{1})(1-\theta_{c})+\nu_{2}\tanh(\beta\nu_{1})+2\nu_{2}p_{g}-2\nu_{2}p_{g}
\\ &
=-\nu_{2}(1-2\theta_{c})\tanh(\beta\nu_{1}).
\end{split}\label{e3}
\end{equation}
In the fourth line, we use Eq. (\ref{teth}) and the fact that $\theta_{c}+{}_{2}\langle +|\Phi(U|+\rangle_{11}\langle +|U^{\dagger})|+\rangle_{2}={}_{2}\langle -|\Phi(U|+\rangle_{11}\langle +|U^{\dagger})|-\rangle_{2}+{}_{2}\langle +|\Phi(U|+\rangle_{11}\langle +|U^{\dagger})|+\rangle_{2}=1$. Finally, we have
\begin{equation}
\begin{split}
\langle E_{4}\rangle & =\mathrm{Tr}\left[\rho_{4}H_{1}\right]
\\ &
=\mathrm{Tr}\left[(-\tanh(\beta\nu_{1})V(\Phi(U|+\rangle_{11}\langle+|U^{\dagger}))V^{\dagger}+p_{g}
   \mathbb{1}_2 )(2\nu_{1}|+\rangle_{11}\langle+|-\nu_{1}\mathbb{1}_2)\right]
   \\ &
=-2\nu_{1}\tanh(\beta\nu_{1})\mathrm{Tr}\left[V(\Phi(U|+\rangle_{11}\langle+|U^{\dagger}))V^{\dagger}|+\rangle_{22}\langle+|\right]+\nu_{1}\tanh(\beta\nu_{1})\mathrm{Tr}\left[V(\Phi(U|+\rangle_{11}\langle+|U^{\dagger}))V^{\dagger}\right]
+2\nu_{1}p_{g}\mathrm{Tr}\left[|+\rangle_{22}\langle+|\right]
\\ &
-\nu_{1}p_{g}\mathrm{Tr}\left[\mathbb{1}_2\right]
   \\ &
=-2\nu_{1}\tanh(\beta\nu_{1})(1-\zeta^{c})+\nu_{1}\tanh(\beta\nu_{1})+2\nu_{1}p_{g}-2\nu_{1}p_{g}
\\ &
=-\nu_{1}(1-2\zeta^{c})\tanh(\beta\nu_{1}).
\end{split}\label{e4}
\end{equation}
In the fourth line, we use Eq. (\ref{zet}) and the fact that $\zeta^{c}+{}_{1}\langle +|V(\Phi(U|+\rangle_{11}\langle +|U^{\dagger}))V^{\dagger}|+\rangle_{1}={}_{1}\langle -|V(\Phi(U|+\rangle_{11}\langle +|U^{\dagger}))V^{\dagger}|-\rangle_{1}+{}_{1}\langle +|V(\Phi(U|+\rangle_{11}\langle +|U^{\dagger}))V^{\dagger}|+\rangle_{1}=1$. From Eqs. (\ref{e1})-(\ref{e2})-(\ref{e3})-(\ref{e4}), follow the averages of work and heat, which are given by,
\begin{equation}
\langle Q_{C}\rangle=\langle E_{1}\rangle-\langle E_{4}\rangle=-2\zeta^{c}\nu_{1}\tanh(\beta\nu_{1}),\label{qcun}
\end{equation}
\begin{equation}
\langle Q_{M}\rangle=\langle E_{3}\rangle-\langle E_{2}\rangle=2(\theta_{c}-\delta')\nu_{2}\tanh(\beta\nu_{1}),\label{qmun}
\end{equation}
and,
\begin{equation}
\langle W\rangle=\langle Q_{M}\rangle+\langle Q_{C}\rangle=2((\theta_{c}-\delta')\nu_{2}-\zeta^{c}\nu_{1})\tanh(\beta\nu_{1}).\label{wun}
\end{equation}
Please note that $\langle Q_{C}\rangle\leq0$ as expected, since $\zeta^{c}\geq0$.
\section{Expression of the second cumulants for the undepahed engine in terms of the transition probabilities: arbitrary Hamiltonians, unitaries, and unital channels}\label{app2r}
Let's now start with the variance of $Q_{C}$. From definition \ref{defin1} and Eq. (\ref{chiFF}), it follows that
\begin{equation}
\begin{split}
\langle Q_{C}^{2}\rangle_{c}:&=\frac{\partial^{2} \mathrm{log} (\chi_{UDE}(\gamma_{1}=\gamma_{C},\gamma_{2}=0,\gamma_{3}=0,\gamma_{4}=-\gamma_{C}))}{\partial (i\gamma_{C})^{2}}\bigg\rvert_{\gamma_{C}=0}
.
\end{split}
\end{equation}
Let's now make $a:=\chi_{UDE}(\gamma_{1}=\gamma_{C},\gamma_{2}=0,\gamma_{3}=0,\gamma_{4}=-\gamma_{C})$. One can arrive at:
\begin{equation}
\begin{split}
\langle Q_{C}^{2}\rangle_{c} &=\left(\frac{1}{a}\frac{\partial^{2}  a}{\partial (i\gamma_{C})^{2}}\right)\bigg\rvert_{\gamma_{C}=0}-\left(\frac{1}{a}\frac{\partial  a}{\partial (i\gamma_{C})}\right)^{2}\bigg\rvert_{\gamma_{C}=0}
\\ &
=\left(\frac{1}{a}\frac{\partial^{2}  a}{\partial (i\gamma_{C})^{2}}\right)\bigg\rvert_{\gamma_{C}=0}-\langle Q_{C}\rangle^{2}.
\end{split}
\end{equation}
Using the fact that when $\gamma_{C}=0$ we obtain $a=1$ then the variance of $Q_{C}$ is given by,
\begin{equation}
\begin{split}
\langle Q_{C}^{2}\rangle_{c} & =\mathrm{Tr}\left[H_{1}^{2}V(\Phi(U\rho_{1}U^{\dagger}))V^{\dagger}\right]+\mathrm{Tr}\left[H_{1}^{2}\rho_{1}\right]-2\mathrm{Tr}\left[H_{1}V(\Phi(UH_{1}\rho_{1}U^{\dagger}))V^{\dagger}\right]-\langle Q_{C}\rangle^{2}
\\ &
=2\nu_{1}^{2}-2\mathrm{Tr}\left[H_{1}V(\Phi(UH_{1}\rho_{1}U^{\dagger}))V^{\dagger}\right]-\langle Q_{C}\rangle^{2}.\label{b6}
\end{split}
\end{equation}
Note that in the second line, we use the fact that $H_{1}^{2}=\nu_{1}^{2}\mathbb{1_{2}}$. Furthermore, since $\rho_{1}H_{1}=H_{1}\rho_{1}$, one can show that $\mathrm{Im} [\langle Q_{C}^{2}\rangle_{c}] = 0$, hence using $\langle Q_{C}^{2}\rangle_{c}$ for the variance of $Q_{C}$ instead of $\mathrm{Re}[\langle Q_{C}^{2}\rangle_{c}]$.

 Now let's compute $\mathrm{Tr}\left[H_{1}V(\Phi(UH_{1}\rho_{1}U^{\dagger}))V^{\dagger}\right]$. We have,
\begin{equation}
\begin{split}
\mathrm{Tr}\left[H_{1}V(\Phi(UH_{1}\rho_{1}U^{\dagger}))V^{\dagger}\right] & =\mathrm{Tr}\left[(\nu_{1}|+\rangle_{11}\langle+|-\nu_{1}|-\rangle_{11}\langle-|)V(\Phi(U(\nu_{1}p_{e}|+\rangle_{11}\langle+|-\nu_{1}p_{g}|-\rangle_{11}\langle-|)U^{\dagger}))V^{\dagger}\right]
\\ &
=\nu_{1}^{2}p_{e}(1-\zeta^{c})-\nu_{1}^{2}p_{g}\zeta^{c}-\nu_{1}^{2}p_{e}\zeta^{c}+\nu_{1}^{2}p_{g}(1-\zeta^{c})
\\ &
=\nu_{1}^{2}(1-2\zeta^{c}).
\end{split}\label{b7}
\end{equation}
In the first line, we replace $H_{1}$ and $\rho_{1}$ by their expressions, and in the second line, we use Eq. (\ref{zet}). Putting all things together, i.e., Eqs. (\ref{b6})-(\ref{b7}), one arrives at,
\begin{equation}
\langle Q_{C}^{2}\rangle_{c}=4\zeta^{c}\nu_{1}^{2}-\langle Q_{C}\rangle^{2}.\label{resu1}
\end{equation}
Using Eq. (\ref{qcun}), we have, $\langle Q_{C}^{2}\rangle_{c}=4\zeta^{c}\nu_{1}^{2}(1-\zeta^{c}\tanh^{2}(\beta\nu_{1}))$. Furthermore, from the fact that $0\leq\zeta^{c}\leq1$, one can see that $\langle Q_{C}^{2}\rangle_{c}\geq0$, in agreement with the fact that this is a variance.

Now let's compute the fluctuations of $Q_{M}$. The latter are given as follows:
\begin{equation}
\mathrm{Re}[\langle Q_{M}^{2}\rangle_{c}]:=\mathrm{Re}\bigg\{\frac{\partial^{2} \mathrm{log} (\chi_{UDE}(\gamma_{1}=0,\gamma_{2}=-\gamma_{M},\gamma_{3}=\gamma_{M},\gamma_{4}=0))}{\partial (i\gamma_{M})^{2}}\bigg\rvert_{\gamma_{M}=0}
\bigg\}.
\end{equation}
$\mathrm{Re}$ here refers to the real part of the second cumulant of $Q_{M}$. That is, in contrast to the second cumulant of $Q_{C}$, the second cumulant of $Q_{M}$ has a real and imaginary part. After simple algebra (as we did for the variance of $Q_{C}$), one can arrive at,
\begin{equation}
\begin{split}
\mathrm{Re}[ \langle Q_{M}^{2}\rangle_{c}] & =\mathrm{Tr}\left[H_{2}^{2}\Phi(\rho_{2})\right]+\mathrm{Tr}\left[H_{2}^{2}\rho_{2}\right]-2\mathrm{Re}\{\mathrm{Tr}\left[H_{2}\Phi(H_{2}\rho_{2})\right]\}-\langle Q_{M}\rangle^{2}
\\ &
=2\nu_{2}^{2}-2\mathrm{Re}\{\mathrm{Tr}\left[H_{2}\Phi(H_{2}\rho_{2})\right]\}-\langle Q_{M}\rangle^{2}
\\ &
=2\nu_{2}^{2}-\mathrm{Tr}\left[H_{2}\Phi(H_{2}\rho_{2}+\rho_{2}H_{2})\right]-\langle Q_{M}\rangle^{2}.
\end{split}
\end{equation}
In the second line, we use the fact that $H_{2}^{2}=\nu_{2}^{2}\mathbb{1}_{2}$. In the last line, we use the fact that $2\mathrm{Re}\{\mathrm{Tr}\left[H_{2}\Phi(H_{2}\rho_{2})\right]\}=2(\mathrm{Tr}\left[H_{2}\Phi(H_{2}\rho_{2})\right]+\mathrm{Tr}\left[H_{2}\Phi(H_{2}\rho_{2})\right]^{\ast})/2=\mathrm{Tr}\left[H_{2}\Phi(H_{2}\rho_{2}+\rho_{2}H_{2})\right]$, where $\mathrm{Tr}\left[H_{2}\Phi(H_{2}\rho_{2})\right]^{\ast}(:=\mathrm{Tr}\left[\left(H_{2}\Phi(H_{2}\rho_{2})\right)^{\dagger}\right])$ is the complex conjugate of $\mathrm{Tr}\left[H_{2}\Phi(H_{2}\rho_{2})\right]$. Now let's simplify the latter term using an important trick. One can prove the next result:
\begin{lemma}
The anticommutator $\{H_{2},\rho_{2}\}:=H_{2}\rho_{2}+\rho_{2}H_{2}$ can be simplified as follows:
\begin{equation}
\{H_{2},\rho_{2}\}=H_{2}\rho_{2}+\rho_{2}H_{2}=2\nu_{2}({}_{2}\langle +|\rho_{2}|+\rangle_{2} |+\rangle_{22}\langle+|-\langle -|\rho_{2}|-\rangle |-\rangle_{22}\langle-|)=2H_{2}\Delta_{2}(\rho_{2}).
\end{equation}\label{lemma}
\end{lemma}
\begin{proof}
Using the completeness relation in the eigenbasis of $H_{2}$, i.e., $|+\rangle_{22}\langle+|+|-\rangle_{22}\langle-|=\mathbb{1}_{2}$, we have
\begin{equation}
\begin{split}
H_{2}\rho_{2}+\rho_{2}H_{2}&=H_{2}\mathbb{1}_{2}\rho_{2}\mathbb{1}_{2}+\mathbb{1}_{2}\rho_{2}\mathbb{1}_{2}H_{2}
\\ &
=\nu_{2}(|+\rangle_{22}\langle+|-|-\rangle_{22}\langle-|)({}_{2}\langle +|\rho_{2}|+\rangle_{2} |+\rangle_{22}\langle+|+{}_{2}\langle +|\rho_{2}|-\rangle_{2} |+\rangle_{22}\langle-|+{}_{2}\langle -|\rho_{2}|+\rangle_{2} |-\rangle_{22}\langle+|+{}_{2}\langle -|\rho_{2}|-\rangle_{2} |-\rangle_{22}\langle-|)
\\ &
+({}_{2}\langle +|\rho_{2}|+\rangle_{2} |+\rangle_{22}\langle+|+{}_{2}\langle +|\rho_{2}|-\rangle_{2} |+\rangle_{22}\langle-|+{}_{2}\langle -|\rho_{2}|+\rangle_{2} |-\rangle_{22}\langle+|+{}_{2}\langle -|\rho_{2}|-\rangle_{2} |-\rangle_{22}\langle-|)\nu_{2}(|+\rangle_{22}\langle+|-|-\rangle_{22}\langle-|)
\\ &
=\nu_{2}({}_{2}\langle +|\rho_{2}|+\rangle_{2} |+\rangle_{22}\langle+|+{}_{2}\langle +|\rho_{2}|-\rangle_{2} |+\rangle_{22}\langle-|-{}_{2}\langle -|\rho_{2}|+\rangle_{2} |-\rangle_{22}\langle+|-{}_{2}\langle -|\rho_{2}|-\rangle_{2} |-\rangle_{22}\langle-|)
\\ &
+\nu_{2}({}_{2}\langle +|\rho_{2}|+\rangle_{2} |+\rangle_{22}\langle+|-{}_{2}\langle +|\rho_{2}|-\rangle_{2} |+\rangle_{22}\langle-|+{}_{2}\langle -|\rho_{2}|+\rangle_{2} |-\rangle_{22}\langle+|-{}_{2}\langle -|\rho_{2}|-\rangle_{2} |-\rangle_{22}\langle-|)
\\&
=2\nu_{2}({}_{2}\langle +|\rho_{2}|+\rangle_{2} |+\rangle_{22}\langle+|-{}_{2}\langle -|\rho_{2}|-\rangle_{2} |-\rangle_{22}\langle-|).
\
\end{split}\label{BHG}
\end{equation}
In the second line, we use the decomposition Eq. (\ref{decomposi}). This shows that the result of the antricomutator $\{H_{2},\rho_{2}\}$ is something that is diagonal in the eigenbasis of $H_{2}$. It can also be written as: $\{H_{2},\rho_{2}\}=2H_{2}\Delta_{2}(\rho_{2})$; see the diagonal part of $\rho_{2}$ in Eq. (\ref{decomposi}).
\end{proof}

The latter proved result would be important to simplify the variances of $Q_{M}$ and $W$. Using equation (\ref{BHG}), we have,
\begin{equation}
\begin{split}
\mathrm{Tr}\left[H_{2}\Phi(H_{2}\rho_{2}+\rho_{2}H_{2})\right]& =2\mathrm{Tr}\left[(\nu_{2}|+\rangle_{22}\langle+|-\nu_{2}|-\rangle_{22}\langle-|)\Phi(\nu_{2}({}_{2}\langle +|\rho_{2}|+\rangle_{2} |+\rangle_{22}\langle+|-{}_{2}\langle -|\rho_{2}|-\rangle_{2} |-\rangle_{22}\langle-|))\right]
\\ &
=2\nu_{2}^{2}\left({}_{2}\langle +|\rho_{2}|+\rangle_{2}(1-\theta)-{}_{2}\langle -|\rho_{2}|-\rangle_{2}\theta-{}_{2}\langle +|\rho_{2}|+\rangle_{2}\theta+{}_{2}\langle -|\rho_{2}|-\rangle_{2}(1-\theta)\right)
\\ &
=2\nu_{2}^{2}(1-2\theta).
\end{split}
\end{equation}
In the second line, we use Eq. (\ref{teth}), and in the third line, $\mathrm{Tr}\left[\rho_{2}\right]={}_{2}\langle +|\rho_{2}|+\rangle_{2}+{}_{2}\langle -|\rho_{2}|-\rangle_{2}=1$. Putting everything together, we obtain the next compact equation,
\begin{equation}
\mathrm{Re}[\langle Q_{M}^{2}\rangle_{c}]=4\nu_{2}^{2}\theta-\langle Q_{M}\rangle^{2}=4\nu_{2}^{2}(\theta-((\theta_{c}-\delta')\tanh(\beta\nu_{1}))^{2}).\label{resu2}
\end{equation}
The imaginary part of the second cumulant of $Q_{M}$ is given by $-2\mathrm{Im}\{\mathrm{Tr}\left[H_{2}\Phi(H_{2}\rho_{2})\right]\}$. The latter can be written as follows:
\begin{equation}
\mathrm{Im}[\langle Q_{M}^{2}\rangle_{c}] = -2\mathrm{Im}\left\{\mathrm{Tr}\left[H_{2}\Phi(H_{2}\rho_{2})\right]\right\} = i\mathrm{Tr}\left[H_{2}\Phi\left( [H_{2},\rho_{2}]\right)\right].
\end{equation}
We see that when the commutator $[H_{2},\rho_{2}] = 0$ (i.e., the state $\rho_{2}$ is incoherent in the eigenbasis of $H_{2}$), the imaginary part of the second cumulant of $Q_{M}$ vanishes. Thus, the state $\rho_{2}$ has zero off-diagonal elements in the eigenbasis of $H_{2}$.

Similarly to the variance of $Q_{C}$ and $Q_{M}$ we have,
\begin{equation}
\mathrm{Re}[\langle W^{2}\rangle_{c}]:=\mathrm{Re}\bigg\{\frac{\partial^{2} \mathrm{log} (\chi_{UDE}(\gamma_{1}=\gamma_{W},\gamma_{2}=-\gamma_{W},\gamma_{3}=\gamma_{W},\gamma_{4}=-\gamma_{W}))}{\partial (i\gamma_{W})^{2}}\bigg\rvert_{\gamma_{W}=0}
\bigg\}.
\end{equation}
And, its explicit expression is given by
\begin{equation}
\begin{split}
\mathrm{Re}[\langle W^{2}\rangle_{c}] & =\mathrm{Tr}\left[H_{1}^{2}\rho_{1}\right]+\mathrm{Tr}\left[H_{2}^{2}\rho_{2}\right]+\mathrm{Tr}\left[H_{2}^{2}\rho_{3}\right]+\mathrm{Tr}\left[H_{1}^{2}\rho_{4}\right]-2\mathrm{Tr}\left[H_{2}U(H_{1}\rho_{1})U^{\dagger}\right]+2\mathrm{Tr}\left[H_{2}\Phi(U(H_{1}\rho_{1})U^{\dagger})\right]
\\ &
-2\mathrm{Tr}\left[H_{1}V(\Phi(U(H_{1}\rho_{1})U^{\dagger}))V^{\dagger}\right]-2\mathrm{Re}\{\mathrm{Tr}\left[H_{2}\Phi(H_{2}\rho_{2})\right]\}+2\mathrm{Re}\{\mathrm{Tr}\left[H_{1}V(\Phi(H_{2}\rho_{2}))V^{\dagger}\right]\}-2\mathrm{Re}\{\mathrm{Tr}\left[H_{1}V(H_{2}\rho_{3})V^{\dagger}\right]\}
\\ &
-\langle W\rangle^{2}.\label{iop}
\end{split}
\end{equation}
The imaginary part of the second cumulant of $W$ is given by:
\begin{equation}
\begin{split}
\mathrm{Im}[\langle W^{2}\rangle_{c}] & = -2\mathrm{Im}\{\mathrm{Tr}\left[H_{2}\Phi(H_{2}\rho_{2})\right]\}+2\mathrm{Im}\{\mathrm{Tr}\left[H_{1}V(\Phi(H_{2}\rho_{2}))V^{\dagger}\right]\}-2\mathrm{Im}\{\mathrm{Tr}\left[H_{1}V(H_{2}\rho_{3})V^{\dagger}\right]\} \\
& = i\mathrm{Tr}\left[H_{2}\Phi\left( [H_{2},\rho_{2}]\right)\right]
-i\mathrm{Tr}\left[H_{1}V(\Phi\left( [H_{2},\rho_{2}]\right))V^{\dagger}\right]
+i\mathrm{Tr}\left[H_{1}V( [H_{2},\rho_{3}])V^{\dagger}\right].
\end{split}
\end{equation}
When the commutator $[H_{2},\rho_{2}] = [H_{2},\rho_{3}] = 0$ (i.e., the states $\rho_{2}$ and $\rho_{3}$ have no coherence in the eigenbasis of $H_{2}$), the imaginary term of the second cumulant of $W$ is zero.

We see from the second cumulants of $Q_{M}$ and $W$ that the imaginary part comes from the average of the next products: $E_{2}E_{3}$, $E_{2}E_{4}$, and $E_{3}E_{4}$. However, the average of the product of $E_{1}$ with $E_{2}$, $E_{3}$, and $E_{4}$ is always real, hence the nullity of the imaginary part of the second cumulant of $Q_{C}$. Furthermore, following the same reasoning as we did for $Q_{C}$ and $Q_{M}$, i.e., by applying lemma \ref{lemma}, one can prove that the work fluctuations are given by
\begin{equation}
\mathrm{Re}[\langle W^{2}\rangle_{c}]=4\nu_{1}\nu_{2}(\delta'+\zeta-\theta_{c}-\zeta_{c})+4(\nu_{1}^{2}\zeta^{c}+\nu_{2}^{2}\theta)-\langle W\rangle^{2}.\label{resu3}
\end{equation}
Finally, note that Eqs. (\ref{qcun})-(\ref{qmun})-(\ref{wun}) and Eqs. (\ref{resu1})-(\ref{resu2})-(\ref{resu3}) are important results of the paper since they compress the expressions of the averages and the fluctuations into simple expressions.
\section{The fluctuations of the heat released to the cold bath}\label{mpmpm}
Following Ref. \cite{Abdelkader2}, one can easily see that the variance of the heat released to the cold bath is given by
\begin{equation}
var(Q_{C}') = \langle E_{1}^{2}\rangle_{c} + \langle E_{4}^{2}\rangle_{c} = \mathrm{Tr}\left[H_{1}^{2}\rho_{1}\right] - \mathrm{Tr}\left[\rho_{1}H_{1}\right]^{2} + \mathrm{Tr}\left[H_{1}^{2}V(\Phi(U\rho_{1}U^{\dagger}))V^{\dagger}\right] - \mathrm{Tr}\left[\rho_{4}H_{1}\right]^{2}.
\end{equation}
Using Eqs. (\ref{e1}--\ref{e4}) and after simple mathematical steps, one can show that
\begin{equation}
var(Q_{C}') = \nu_1^2 \left( 2 - \tanh^2(\beta \nu_1) \left[ (1 - 2\zeta^c)^2 + 1 \right] \right).
\end{equation}

\section{Lower bounds on the RFs of $Q_{C}$}\label{lowerrr}
From Eq. (\ref{qcun}) and Eq. (\ref{resu1}), we have,
\begin{equation}
\langle \Sigma\rangle\left(\frac{\langle Q_{C}^{2}\rangle_{c}}{\langle Q_{C}\rangle^{2}}+1\right)=-\beta\langle Q_{C}\rangle\left(\frac{\langle Q_{C}^{2}\rangle_{c}}{\langle Q_{C}\rangle^{2}}+1\right)=2\beta\nu_{1}\coth(\beta\nu_{1})\geq2.\label{yy}
\end{equation}
The lower bound is achieved when $\beta=0$, i.e., \textit{in the high-temperature regime}. From equation (\ref{yy}), we see that
\begin{equation}
\frac{\langle Q_{C}^{2}\rangle_{c}}{\langle Q_{C}\rangle^{2}}\geq\frac{2\beta\nu_{1}\coth(\beta\nu_{1}}{\langle \Sigma\rangle}-1\geq\frac{2}{\langle \Sigma\rangle}-1.
\end{equation}
Furthermore, one can also show that when $\zeta^{c}\leq1/2$, it follows that
\begin{equation}
\langle Q_{C}^{2}\rangle_{c}\geq\langle Q_{C}\rangle^{2}.
\end{equation}
Consider an arbitrary unitary $U$. When $V=U^{\dagger}$ and for the unital channel of Sec. \ref{casestu}, one can show that
\begin{equation}
0\leq \zeta^{c}\leq1/2.
\end{equation}
\begin{proof}
After simple lines of algebra, one can show that $\zeta^{c}$ is given as follows:
\begin{equation}
\begin{split}
\zeta^{c} & =\sum_{j}{}_{1}\langle-|U^{\dagger}\pi_{j}U|+\rangle_{11}\langle+|U^{\dagger}\pi_{j}U|-\rangle_{1}
\\ &
=|_{1}\langle-|U^{\dagger}|\pi_{1}\rangle|^{2}|\langle\pi_{1}|U|+\rangle_{1}|^{2}+|_{1}\langle-|U^{\dagger}|\pi_{2}\rangle|^{2}|\langle\pi_{2}|U|+\rangle_{1}|^{2}.
\end{split}
\end{equation}
Where we replace $\pi_{j}$ by $|\pi_{j}\rangle\langle\pi_{j}|$ for $j=1$ and $2$. And by defining $p_{1}:=|\langle\pi_{1}|U|+\rangle_{1}|^{2}$ and $p_{2}:=|\langle\pi_{2}|U|+\rangle_{1}|^{2}$, one can find that,
\begin{equation}
\zeta^{c}=(1-p_{1})p_{1}+(1-p_{2})p_{2}.
\end{equation}
This follows from the microreversibility principle, i.e.,
\begin{equation}
\begin{split}
|_{1}\langle-|U^{\dagger}|\pi_{1}\rangle|^{2} & =|\langle\pi_{1}|U|-\rangle_{1}|^{2}
\\ &
=\langle\pi_{1}|U(\mathbb{1}_{2}-|+\rangle_{11}\langle+|)U^{\dagger}|\pi_{1}\rangle
\\ &
=1-|\langle\pi_{1}|U|+\rangle_{1}|^{2}
\\ &
=1-p_{1}.
\end{split}
\end{equation}
In the same manner, we have $|_{1}\langle-|U^{\dagger}|\pi_{2}\rangle|^{2}=1-p_{2}$. Further, one can also prove that $p_{2}=1-p_{1}$ as follows:
\begin{equation}
\begin{split}
p_{2} & =|\langle\pi_{2}|U|+\rangle_{1}|^{2}
\\ &
={}_{1}\langle+|U^{\dagger}|\pi_{2}\rangle\langle\pi_{2}|U|+\rangle_{1}
\\ &
={}_{1}\langle+|U^{\dagger}(\mathbb{1}_{2}-|\pi_{1}\rangle\langle\pi_{1}|)U|+\rangle_{1}
\\ &
=1-|\langle\pi_{1}|U|+\rangle_{1}|^{2}
\\ &
=1-p_{1}.
\end{split}
\end{equation}
Putting everything together, we obtain
\begin{equation}
\zeta^{c}=2p_{1}(1-p_{1}).
\end{equation} 
And from the fact that $0\leq p_{1}\leq1$, we conclude that
\begin{equation}
0\leq \zeta^{c}\leq1/2.
\end{equation}
The highest value is achieved when $p_{1}=1/2$. 
\end{proof}
\section{Proof of $\zeta_{c}=\theta_{c}$ and $\delta'=\zeta$, when $V=U^{\dagger}$}\label{rwqm}
For arbitary $U$ and $V$ that satisfy $V=U^{\dagger}$, and for the unital channel of Sec. \ref{casestu}, one can show that $\zeta_{c}=\theta_{c}$.
\begin{proof}
We have,
\begin{equation}
\begin{split}
\zeta_{c} & ={}_{1}\langle -|V\Phi(|+\rangle_{22}\langle +|)V^{\dagger}|-\rangle_{1}
\\ &
=\sum_{j} {}_{1}\langle -|U^{\dagger}\pi_{j}|+\rangle_{22}\langle +|\pi_{j}U|-\rangle_{1}
\\ &
=\sum_{j}{}_{2}\langle +|\pi_{j}U|-\rangle_{11}\langle -|U^{\dagger}\pi_{j}|+\rangle_{2}
\\ &
={}_{2}\langle +|\Phi(U|-\rangle_{11}\langle -|U^{\dagger})|+\rangle_{2}
\\ &
=\theta_{c}.
\end{split}\label{olpm}
\end{equation}
In the second line, we replace $\Phi$ by its expression and $V$ by $U^{\dagger}$.
\end{proof}
Now let's prove thtat when $V=U^{\dagger}$ we have $\delta'=\zeta$.
\begin{proof}
We have,
\begin{equation}
\begin{split}
\zeta & =|{}_{1}\langle +|V|-\rangle_{2}|
\\ &
=|{}_{1}\langle +|U^{\dagger}|-\rangle_{2}|
\\ &
={}_{2}\langle-|U|+\rangle_{11}\langle +|U^{\dagger}|-\rangle_{2}
\\ &
={}_{2}\langle-|U(\mathbb{1}_{2}-|-\rangle_{11}\langle -|)U^{\dagger}|-\rangle_{2}
\\ &
=1-{}_{2}\langle-|U|-\rangle_{11}\langle -|U^{\dagger}|-\rangle_{2}
\\ &
=1-({}_{1}\langle-|U^{\dagger}(\mathbb{1}_{2}-|+\rangle_{22}\langle +|)U|-\rangle_{1})
\\ &
=1-(1-{}_{1}\langle-|U^{\dagger}|+\rangle_{22}\langle +|U|-\rangle_{1})
\\ &
=|_{2}\langle +|U|-\rangle_{1}|^{2}
\\ &
=\delta'.
\end{split}
\end{equation}
\end{proof}

\section{Proof of theorem 2 (Eq. (\ref{UppLoE}))}\label{EERF}

\begin{proof}
Using the results of Appendix \ref{rwqm}, one can show the next series of equalities;
\begin{equation}
\begin{split}
\mathrm{Re}[\langle W^{2}\rangle_{c}] & =4\nu_{1}\nu_{2}(\delta'+\zeta-\theta_{c}-\zeta_{c})+4(\nu_{1}^{2}\zeta^{c}+\nu_{2}^{2}\theta)-\langle W\rangle^{2}
\\ &
=4\nu_{1}\nu_{2}(2\delta'-2\theta_{c})+4(\nu_{1}^{2}\zeta^{c}+\nu_{2}^{2}\theta)-\langle W\rangle^{2}
\\ &
=4\nu_{1}(-2\nu_{2}(\theta_{c}-\delta')+\nu_{1}\zeta^{c})+4\nu_{2}^{2}\theta-\langle W\rangle^{2}
\\ &
=4\nu_{1}(-\coth(\beta\nu_{1})\langle Q_{M}\rangle-\coth(\beta\nu_{1})\langle Q_{C}\rangle/2)+4\nu_{2}^{2}\theta-\langle W\rangle^{2}
\\ &
=-2\nu_{1}\coth(\beta\nu_{1})(\langle Q_{M}\rangle+\langle W\rangle)+4\nu_{2}^{2}\theta-\langle W\rangle^{2}
.
\end{split}\label{mpmp}
\end{equation}
In the fourth line, we use Eqs. (\ref{qcun1})-(\ref{qmun1}). Further, now consider the difference $\mathrm{Re}[\langle Q_{M}^{2}\rangle_{c}]-\mathrm{Re}[\langle W^{2}\rangle_{c}]$. One can show the next series of equalities,
\begin{equation}
\begin{split}
\mathrm{Re}[\langle Q_{M}^{2}\rangle_{c}]-\mathrm{Re}[\langle W^{2}\rangle_{c}] & =4\nu_{2}^{2}\theta-\langle Q_{M}\rangle^{2}- (-2\nu_{1}\coth(\beta\nu_{1})(\langle Q_{M}\rangle+\langle W\rangle)+4\nu_{2}^{2}\theta-\langle W\rangle^{2})
\\ &
=2\nu_{1}\coth(\beta\nu_{1})(\langle Q_{M}\rangle+\langle W\rangle+(\langle W\rangle-\langle Q_{M}\rangle)(\langle Q_{M}\rangle+\langle W\rangle)
\\ &
=(\langle W\rangle+\langle Q_{M}\rangle)(2\nu_{1}+\langle Q_{C}\rangle\tanh(\beta\nu_{1}))\coth(\beta\nu_{1})
\\ &
=2\nu_{1}\coth(\beta\nu_{1})(\langle W\rangle+\langle Q_{M}\rangle)(1-\zeta^{c}\tanh^{2}(\beta\nu_{1})).
\end{split}
\end{equation}
Note that the latter result is always $\geq0$ in the heat engine region. Further, the term $(1-\zeta^{c}\tanh^{2}(\beta\nu_{1}))$ is always $\geq0$, since $0\leq \zeta^{c}\tanh^{2}(\beta\nu_{1})\leq1$. On the other hand, the term $\langle W\rangle+\langle Q_{M}\rangle$ is $\geq(\leq)0$ when both $\langle W\rangle$ and $\langle Q_{M}\rangle$ are $\geq(\leq)0$. Of course, $\langle W\rangle+\langle Q_{M}\rangle$ can still be $\geq0$ even when $\langle W\rangle\leq0$ is $\leq$ and $\langle Q_{M}\rangle$ $\geq0$ such that their sum is $\geq0$. From all this, we see that in the heat engine region, we have $\mathrm{Re}[\langle Q_{M}^{2}\rangle_{c}]-\mathrm{Re}[\langle W^{2}\rangle_{c}]\geq0$, thus
\begin{equation}
\frac{\mathrm{Re}[\langle W^{2}\rangle_{c}]}{\mathrm{Re}[\langle Q_{M}^{2}\rangle_{c}]}\leq1.
\end{equation}

\end{proof}
Note that they become equal at the point where efficiency goes to 1. In this case, $W$ and $Q_{M}$ converge to 0. But note that numerically, we found that fluctuations of both $W$ and $Q_{M}$ are non-zero.

\section{Difference between the relative fluctuations}\label{diffref}
Let's compute the difference between the relative fluctuations of work $W$ and heat $Q_{M}$. We have,
\begin{equation}
\begin{split}
\frac{\mathrm{Re}[\langle W^{2}\rangle_{c}]}{\langle W\rangle^{2}}-\frac{\mathrm{Re}[\langle Q_{M}^{2}\rangle_{c}]}{\langle Q_{M}\rangle^{2}} & =\frac{\mathrm{Re}[\langle W^{2}\rangle_{c}]\langle Q_{M}\rangle^{2}-\mathrm{Re}[\langle Q_{M}^{2}\rangle_{c}]\langle W\rangle^{2}}{\langle W\rangle^{2}\langle Q_{M}\rangle^{2}}
\\ &
=\frac{(-2\nu_{1}\coth(\beta\nu_{1})(\langle Q_{M}\rangle+\langle W\rangle)+4\nu_{2}^{2}\theta-\langle W\rangle^{2})\langle Q_{M}\rangle^{2}-(4\theta\nu_{2}^{2}-\langle Q_{M}\rangle^{2})\langle W\rangle^{2}}{\langle W\rangle^{2}\langle Q_{M}\rangle^{2}}
\\ &
=\frac{-2\nu_{1}\coth(\beta\nu_{1})(\langle Q_{M}\rangle+\langle W\rangle)\langle Q_{M}\rangle^{2}+4\theta\nu_{2}^{2}(\langle Q_{M}\rangle^{2}-\langle W\rangle^{2})}{\langle W\rangle^{2}\langle Q_{M}\rangle^{2}}
\\ &
=\frac{-2\nu_{1}\coth(\beta\nu_{1})(\langle Q_{M}\rangle+\langle W\rangle)\langle Q_{M}\rangle^{2}+4\theta\nu_{2}^{2}(\langle Q_{M}\rangle-\langle W\rangle)(\langle Q_{M}\rangle+\langle W\rangle)}{\langle W\rangle^{2}\langle Q_{M}\rangle^{2}}
\\ &
=\frac{(\langle Q_{M}\rangle+\langle W\rangle)(-2\nu_{1}\coth(\beta\nu_{1})\langle Q_{M}\rangle^{2}-4\theta\nu_{2}^{2}\langle Q_{C}\rangle)}{\langle W\rangle^{2}\langle Q_{M}\rangle^{2}}
\\ &
=\frac{(\langle Q_{M}\rangle+\langle W\rangle)(8\theta\zeta^{c}\nu_{1}\nu_{2}^{2}\tanh(\beta\nu_{1})-8\nu_{1}\nu_{2}^{2}\tanh(\beta\nu_{1})(\theta_{c}-\delta')^{2})}{\langle W\rangle^{2}\langle Q_{M}\rangle^{2}}
\\ &
=\frac{8\nu_{1}\nu_{2}^{2}\tanh(\beta\nu_{1})(\langle Q_{M}\rangle+\langle W\rangle)(\theta\zeta^{c}-(\theta_{c}-\delta')^{2})}{\langle W\rangle^{2}\langle Q_{M}\rangle^{2}}.
\end{split}
\end{equation}
In the second line, we use Eqs. (\ref{resu22})-(\ref{mpmp}). In the sixth line, we use Eqs. (\ref{qcun1})-(\ref{qmun1}). Similarly to those steps of calculus, one can show that, 
\begin{equation}
\frac{\mathrm{Re}[\langle W^{2}\rangle_{c}]}{\langle W\rangle^{2}}-\frac{\langle Q_{C}^{2}\rangle_{c}}{\langle Q_{C}\rangle^{2}}=\frac{(4\nu_{1}\nu_{2}\tanh(\beta\nu_{1}))^{2}(\theta\zeta^{c}-(\theta_{c}-\delta')^{2})}{\langle W\rangle^{2}\langle Q_{C}\rangle^{2}}.
\end{equation}

\end{document}